\documentclass[preprint,pra,twocolumn,10pt,twoside]{revtex4}%
\pdfoutput=1
\usepackage{amsfonts}
\usepackage{amsmath}
\usepackage{amssymb}
\usepackage{graphicx}%
\graphicspath{{pics/}} 

\usepackage[caption=false]{subfig}
\usepackage{hyperref}
\usepackage{pra_my}
\usepackage{pra_math_my}
\usepackage[acronyms]{glossaries}%
\glsdisablehyper
\setacronymstyle{long-short}
\newglossary[nlg]{notation}{not}{ntn}{}
\newcommand{\hiddengls}[1]{\glslink{#1}{}}
\newglossaryentry{graph}{%
    type=notation,
    name={\ensuremath{G}},
    description={A graph},
    sort={graph}
}%
\newglossaryentry{adjMatG}{%
    type=notation,
    name={\ensuremath{A_G}},
    description={The adjacency matrix of the graph $G$},
    sort={adjacency matrix}
}%
\newglossaryentry{gState}{%
    type=notation,
    name={\ensuremath{\ket{G}}},
    description={A graph state},
    sort={graph state}
}%
\newglossaryentry{vertexOfG}{%
    type=notation,
    name={\ensuremath{V_G}},
    description={The vertex set of $G$},
    sort={graph theory: vertex set}
}%
\newglossaryentry{edgeOfG}{%
    type=notation,
    name={\ensuremath{E_G}},
    description={The edge set of $G$},
    sort={graph theory: edge set}
}%
\newglossaryentry{neighborhoodV}{%
    type=notation,
    name={\ensuremath{N_v}},
    description={The neighborhood of $v$},
    sort={graph theory: neighborhood set}
}%
\newglossaryentry{inducedG}{%
    type=notation,
    name={\ensuremath{G[\xi]}},
    description={The subgraph of $G$ induced by vertices $\xi$},
    sort={graph theory: induced subgraph}
}%
\newcommand{\iBasis}[2][]{i_{#1}^{(#2)}} 
\newglossaryentry{iState}{%
    type=notation,
    name={\ensuremath{|\iBasis[\alpha]{\iota}\rangle}},
    description={The $\alpha$-basis state with binary number corresponding to the index set $\iota$, $\alpha\in\{X,Y,Z\}$},
    sort={iState}
}%
\newglossaryentry{gStateGenerator}{%
    type=notation,
    name={\ensuremath{g_i}},
    description={The graph state stabilizer generator associated to $i$th vertex},
    sort={graph state: 1 stabilizer generator}
}%
\newglossaryentry{inducedS}{%
    type=notation,
    name={\ensuremath{s_G^{(\xi)}}},
    description={The graph state stabilizer induced by the vertex subset $\xi$},
    sort={graph state: 2 induced stabilizer}
}%
\newglossaryentry{stabGroup}{%
    type=notation,
    name={\ensuremath{\mathcal{S}_G}},
    description={The graph state stabilizer group of $G$},
    sort={graph state: 3 stabilizer group}
}%
\newglossaryentry{powerset}{%
    type=notation,
    name={\ensuremath{\mathcal{P}(V_G)}},
    description={The power set of the vertex set $V_G$},
    sort={powerset}
}%
\newglossaryentry{corrIndex}{%
    type=notation,
    name={\ensuremath{c_{\xi}}},
    description={The correlation index of the vertex subset $\xi$},
    sort={correlation index}
}%
\newglossaryentry{corrIndexSet}{%
    type=notation,
    name={\ensuremath{\mathcal{C}_G}},
    description={The set all correlation indices in $G$},
    sort={correlation index set}
}%
\newglossaryentry{stabParity}{%
    type=notation,
    name={\ensuremath{\pi_G(\xi)}},
    description={The stabilizer parity of $\xi$ in $G$},
    sort={stabilizer parity}
}%
\newglossaryentry{XResourceSet}{%
    type=notation,
    name={\ensuremath{\mathcal{X}_G^{(c)}}},
    description={The set of all X-resources of $c$-correlation},
    sort={X-resource set}
}%
\newglossaryentry{Xchain}{%
    type=notation,
    name={\ensuremath{\mathcal{X}_G^{(\emptyset)}}},
    description={The set of X-chains},
    sort={X-chain}
}%
\newglossaryentry{XchainGroup}{%
    type=notation,
    name={\ensuremath{\braket{\Gamma_G}}},
    description={The X-chain group generated by its generating set $\Gamma_G$},
    sort={X-chain group}
}%
\newglossaryentry{corrGroupG}{%
    type=notation,
    name={\ensuremath{\braket{\mathcal{K}_G}}},
    description={The correlation group of $G$ generated by its generating set $\mathcal{K}_G$},
    sort={correlation group of G}
}%
\newglossaryentry{corrGroup}{%
    type=notation,
    name={\ensuremath{\braket{\mathcal{K}}}},
    description={A general correlation subgroup of $\braket{\mathcal{K}_G}$},
    sort={correlation group 1}
}%
\newglossaryentry{XchainStateBasic}{%
    type=notation,
    name={\ensuremath{\ket{i^{(x_{\Gamma})}}}},
    description={The basic X-chain state},
    sort={X-chain state: basic}
}%
\newglossaryentry{XchainState}{%
    type=notation,
    name={\ensuremath{\ket{\psi_{\emptyset}(\xi)}}},
    description={An X-chain state},
    sort={X-chain state: general}
}%
\newglossaryentry{corrStateSet}{%
    type=notation,
    name={\ensuremath{\Psi_{\mathcal{K}'}^{(\mathcal{K})}}},
    description={The set $\mathcal{K}$-correlation states $\ket{\psi_{\mathcal{K}}(\xi)}$ with $\xi\in \braket{\mathcal{K}'}/\braket{\mathcal{K}}$},
    sort={correlation state set}
}%
\newglossaryentry{biasDegree}{%
    type=notation,
    name={\ensuremath{\beta(\ket{G})}},
    description={The Z-bias degree of the graph state $\ket{G}$},
    sort={bias degree}
}%
\newglossaryentry{SchmidtRank}{%
    type=notation,
    name={\ensuremath{r_S}},
    description={The Schmidt rank},
    sort={Schmidt rank}
}%
\newglossaryentry{corrGroupToB}{%
    type=notation,
    name={\ensuremath{\braket{\mathcal{K}^{(B)}}}},
    description={The correlation subgroup, whose elements possess correlation index only in the subsystem $B$},
    sort={correlation group to B}
}%
\newglossaryentry{corrGroupAToA}{%
    type=notation,
    name={\ensuremath{\braket{\mathcal{K}_{A}^{(A)}}}},
    description={The correlation subgroup, whose elements and their corresponding correlation index are both in the subsystem $A$},
    sort={correlation group A to B}
}%
\newglossaryentry{corrGroupSimBToA}{%
    type=notation,
    name={\ensuremath{\braket{\mathcal{K}_{\sim B}^{(A)}}}},
    description={A special correlation subgroup},
    sort={correlation group sim B to A}
}%
\newglossaryentry{corrGroupASepB}{%
    type=notation,
    name={\ensuremath{\braket{\mathcal{K}^{A\rfloor B}}}},
    description={The correlation subgroup, whose corresponding correlation state are the $A|B$-separable Schmidt basis},
    sort={correlation group A sep B}
}%
\newglossaryentry{corrState}{%
    type=notation,
    name={\ensuremath{\ket{\psi_{\mathcal{K}}(\xi)}}},
    description={A $\mathcal{K}$-correlation state},
    sort={correlation state}
}%
\newglossaryentry{corrASepBState}{%
    type=notation,
    name={\ensuremath{\ket{\psi_{A\rfloor B}(\xi)}}},
    description={A $\mathcal{K}^{(A\rfloor B)}$-correlation state},
    sort={correlation state A sep B}
}%
\newglossaryentry{corrGroupAandB}{%
    type=notation,
    name={\ensuremath{\braket{\mathcal{K}^{A\leftharpoonup B}}}},
    description={The correlation subgroup obtained by the quotient group $\mathcal{K}_G/\mathcal{K}^{A\rfloor B}$},
    sort={correlation group A sep B}
}%
\newglossaryentry{corrAsepBStateOnA}{%
    type=notation,
    name={\ensuremath{\ket{\phi_{A\rfloor B}^{(A)}(\xi)}}},
    description={The state projected from the $A\rfloor B$-correlation state $\ket{\psi_{A\rfloor B}(\xi)}$ onto the subsystem $A$},
    sort={correlation state A sep B on A}
}%
\newglossaryentry{corrAsepBStateOnB}{%
    type=notation,
    name={\ensuremath{\ket{\phi_{A\rfloor B}^{(B)}(\xi)}}},
    description={The state projected from the $A\rfloor B$-correlation state $\ket{\psi_{A\rfloor B}(\xi)}$ onto the subsystem $B$},
    sort={correlation state A sep B on B}
}%
\newglossaryentry{BiEntGeoMeas}{%
    type=notation,
    name={\ensuremath{\mathcal{E}^{A|B}_g}},
    description={$A|B$-bipartite geometric measure of entanglement},
    sort={geometric measure}
}%

\makenoidxglossaries
\newcommand{\midrule}{\colrule}
\newcommand{\bottomrule}{\botrule}

\usepackage{multirow}
\usepackage{calc}
\usepackage{array}
\newcolumntype{L}[1]{>{\raggedright\let\newline\\\arraybackslash\hspace{0pt}}m{#1}}
\newcolumntype{C}[1]{>{\centering\let\newline\\\arraybackslash\hspace{0pt}}m{#1}}
\newcolumntype{R}[1]{>{\raggedleft\let\newline\\\arraybackslash\hspace{0pt}}m{#1}}
\usepackage{Xfactorizationdiagram}
\newlength{\colwidthA}
\newlength{\colwidthB}
\newlength{\colwidthC}
\newlength{\colwidthD}
\newlength{\colwidthE}
\newlength{\lastcolwidth} 
\newlength{\floatwidth}
\def\graphscale{0.9} 
\begin{document}
\title{Group structures and representations of graph states}
\author{Jun-Yi Wu}
\affiliation{Institut f\"ur Theoretische Physik III, Heinrich-Heine-Universit\"at D\"usseldorf, D-40225 D\"usseldorf, Germany}
\author{Hermann Kampermann}
\affiliation{Institut f\"ur Theoretische Physik III, Heinrich-Heine-Universit\"at D\"usseldorf, D-40225 D\"usseldorf, Germany}
\author{Dagmar Bru\ss}
\affiliation{Institut f\"ur Theoretische Physik III, Heinrich-Heine-Universit\"at D\"usseldorf, D-40225 D\"usseldorf, Germany}

\pacs{PACS number}

\begin{abstract}
A special configuration of graph state stabilizers, which contains only Pauli $\sigma_X$ operators, is studied. The vertex sets $\xi$ associated with such configurations are defined as what we call X-chains of graph states.
The X-chains of a general graph state can be determined efficiently.
They form a group structure such that one can obtain the explicit representation of graph states in the X-basis via the so-called X-chain factorization diagram.
We show that graph states with different X-chain groups can have different probability distributions of X-measurement outcomes, which allows one to distinguish certain graph states with X-measurements.
We provide an approach to find the Schmidt decomposition of graph states in the X-basis.
The existence of X-chains in a subsystem facilitates error correction in the entanglement localization of graph states.
In all of these applications, the difficulty of the task decreases with increasing number of X-chains.
Furthermore, we show that the overlap of two graph states can be efficiently determined via X-chains, while its computational complexity with other known methods increases exponentially.
\end{abstract}

%
%
%
%

\maketitle

\section{Introduction}
Graph states \cite{BriegelRaussendorf2001-FirstGrSt, RaussendorfBriegel2001-05,RaussendorfBriegel2002-10,RaussendorfBrowneBriegel2003-08,RaussendorfPHDThesis2003,HeinEisertBriegel2004-06, GraphStateReviews2006} represent specific multipartite entangled quantum systems.
They are an important resource for measurement-based quantum computation:
there, the multipartite entanglement of cluster states (a special class of graph states) is consumed by local measurements on subsystems.
Depending on the measurement outcomes, local unitary transformations of the remaining systems are performed. In this way, certain quantum operations can be implemented.
\par
Graph states can be represented in the stabilizer formalism as eigenstates of certain tensor products of Pauli $\sigma_X$- and $\sigma_Z$-operators (the graph state stabilizers).
The explicit structure of the stabilizer operators depends on the structure of the underlying graph.
The stabilizers form a group (under multiplication), which is generated from $n$ generators, where $n$ is the number of vertices of the graph.
\par
In this paper, we will discuss what we call X-chains.
X-chains are subsets of vertices of a given graph which correspond to graph state stabilizers that consist \emph{only} of Pauli $\sigma_X$-operators.
We will show that these X-chains form a group.
Not every graph contains an X-chain.
However, it will be shown that if a graph does contain X-chains, this fact can be used as an efficient tool to determine essential properties of the corresponding graph state, such as its overlap with other graph states, its entanglement characteristics, and the existence of error correcting code words in subsystems of graph states.
Note that the overlap of two graph states cannot be determined efficiently to date. The X-chains provide an efficient method to solve this problem.
\par
While graph states are usually given in the Z-basis, the concepts and methods developed in this paper show that it is often favorable to represent graph states in the X-basis, in particular when one wants to study overlaps of graph states or determine their entanglement properties.
The reason for this fact is that for all graph states originating from the same number of vertices, the probability distribution of outcomes of local Z-measurements is uniform,
while it is nonuniform for outcomes of local X-measurements.
Different X-measurement outcomes of two graph states reflect their difference in the X-chain groups,
as the existence of an X-chain in a graph state implies vanishing probability of certain X-measurement outcomes.
Conversely, X-chain groups of graph states determine their representation in the X-basis.
\par
In the present paper we will focus on introducing the concept of X-chains, illustrating it with examples, and presenting some applications.
The X-chain group of a given graph state can be efficiently determined, the search of X-chains in a given graph state will be studied in detail elsewhere \cite{WuKampermannBruss2015-XChainAlgo} and a MATHEMATICA package is available in the Supplemental
Material \cite{XchainMpackage_Wu}.
%
%
\par
This paper is organized as follows.
In section \ref{sec::basic_concepts}, we review the essential concepts of graph theory and graph states.
In section \ref{sec::representation_of_graph_states}, we review the representation of graph states in the Z-basis and point out its disadvantage in distinguishing graph states.
Then we introduce X-chains and study their properties in section \ref{sec::X-chains}.
The representation of graph states in the X-basis is derived via the so-called X-chain factorization in section \ref{sec::X-chain_factorization_of_graph_states}, where we show how X-chain groups feature the X-measurement outcomes on graph states.
In section \ref{sec::application_of_X-chains}, we discuss several applications of X-chains, namely the calculation of the overlap of two graphs states (section \ref{sec::graph_state_overlap}), the Schmidt decomposition of graph states in the X-basis (section \ref{sec::schmidt_decomposition}) and the entanglement localization \cite{PoppVMCirac2005-LocalizableEnt} of graph states against errors (section \ref{sec::unilateral_projection_against_errors}).
The proofs are presented in Appendix \ref{apdx::proofs_for_X-chains}. A list of notations and symbols is given in Appendix \ref{apdx::list_of_notations}.
\subsection{Basic concepts}
\label{sec::basic_concepts}
Here we review the concepts of graphs \cite{Diestel_GraphTheory} and graph states \cite{HeinEisertBriegel2004-06,GraphStateReviews2006}, and introduce the notation used in the main text.\par
\bigskip
\paragraph{Graph theory \cite{Diestel_GraphTheory}:}
\hiddengls{graph}
\hiddengls{vertexOfG}
\hiddengls{edgeOfG}
\hiddengls{neighborhoodV}
\hiddengls{inducedG}
A \emph{graph} $G=(V,E)$ consists of $n$ vertices $V$ and $l$ edges $E$.
The \emph{vertices}, denoted by $V_{G}=\left\{  v_{1},...,v_{n}\right\}  $, are depicted as dots and represent locations, particles etc.
The \emph{edges}, denoted by $E_{G}=\left\{  e_{1},...,e_{l}\right\}  $, describe a relation network between vertices.
A symmetric relation between two vertices $v_{1}$ and $v_{2}$, e.g. a two-way bridge between two islands, can be represented by the vertex set $e=\{v_{1},v_{2}\}$, which is called \emph{undirected edge}.
Let $\xi_{a},\xi_{b}\subseteq V_{G}$ be two subsets of
$V_{G}$, then the edges between $\xi_{a}$ and $\xi_{b}$ are the edges
$e=\{v_{a},v_{b}\}$, which have one vertex $v_{a}\in\xi_{a}$ and another
vertex $v_{b}\in\xi_{b}$. The set of these edges is denoted by $E_{G}(\xi
_{a}:\xi_{b})$. A vertex $v_{1}$ is a \emph{neighbor} of $v_{2}$, if they are connected
by an edge. The set of all neighbors of $v$, called the neighborhood of $v$,
is denoted as $N_{v}$.
In Table \ref{table::different_types_of_graphs} we list two of the relevant types of
graphs, which will be considered in the main text.
\begin{table}[t]
  \centering
\setlength{\colwidthA}{0.2\columnwidth}
\setlength{\colwidthB}{0.15\columnwidth}
\setlength{\lastcolwidth}{0.9\columnwidth-\colwidthA-\colwidthB}

\begin{tabular}[c]{|C{\colwidthA}|C{\colwidthB}|m{\lastcolwidth}|}
    \toprule
    Name of graph & Graph & Definition\\
    \midrule
    \emph{Star graph} $S_{n}$ &
    \includegraphics[width=\graphscale\linewidth]{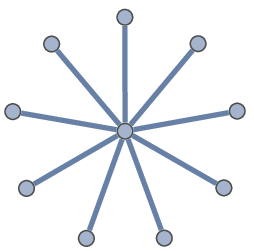} &
    Graphs, for which the center vertex has $n-1$ neighbors and all the others have the center vertex as their only neighbor.
    \\
    \midrule
    \emph{Cycle graph} $C_{n}$ &
    \includegraphics[width=\graphscale\linewidth]{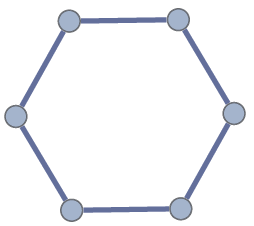} &
    Graphs, for which every vertex has degree $2$. They are closed paths.
    \\
    \bottomrule
\end{tabular} 
  \caption{The graphs considered in this paper.}
  \label{table::different_types_of_graphs}
\end{table}
\par
A graph $F$ is a \emph{subgraph} of $G$, if its vertices and edges are subsets
of the vertex set and the edge set of $G$, respectively, i.e., $V_{F}\subseteq
V_{G}$ and $E_{F}\subseteq E_{G}$.
A subgraph induced by a vertex set $\xi\subseteq V_{G}$ is defined as the
graph
\begin{equation}
G[\xi]:=(\xi,E_{G}(\xi:\xi)),
\end{equation}
which has the edge set $E_{G}(\xi:\xi)$ consisting of edges between vertices
inside the set $\xi$.
\par
\paragraph{Binary notation:}
\hiddengls{iState}
In this paper, we use binary numbers to denote a subset of vertices of
graphs. Let $G$ be a graph with vertices $V_{G}=\left\{  1,...,n\right\}  $
and $\xi\subseteq V_{G}$ be a vertex subset. We denote the binary number of
$\xi$ as%
\begin{equation}
i^{(\xi)}:=i_{1}\cdots i_{n},
\end{equation}
with
\[
i_{j}=\left\{
\begin{array}
[c]{cc}%
0 & ,j\not \in \xi\\
1 & ,j\in\xi
\end{array}
\right.  .
\]
For example, in a 4-vertex graph, $0110=i^{\left\{  2,3\right\}  }$.
The tensor product of Pauli-operators $\sigma_{\alpha}$ with $\alpha\in\{x,y,z\}$ is denoted as%
\begin{equation}
\sigma_{\alpha}^{(\xi)}:=\sigma_{\alpha}^{i_{1}\left(  \xi\right)  }\otimes\cdots \otimes\sigma_{\alpha}^{i_{n}\left(  \xi\right)  },
\end{equation}
with $\sigma_{\alpha}^{0}=\id,\sigma_{\alpha}^{1}=\sigma_{\alpha}$. For example, for $n=4$,
$\sigma_{\alpha}^{\left\{  2,3\right\}  }:=\id\otimes\sigma_{\alpha}\otimes\sigma_{\alpha}\otimes\id$.

\section{Representation of graph states}
\label{sec::representation_of_graph_states}
\hiddengls{gStateGenerator}
\hiddengls{inducedS}
\hiddengls{stabGroup}
\hiddengls{powerset}
We review the representation of graph states \cite{HeinEisertBriegel2004-06,GraphStateReviews2006}.
A given graph with $n$ vertices has a
corresponding quantum state by associating each vertex $v_i$ with a graph
state \emph{stabilizer generator} $g_{i}$,
\begin{equation}
g_{i}=\sigma_{X}^{(i)}\sigma_{Z}^{(N_{i})}.
\end{equation}
Here, $N_{i}$ is the neighborhood of the vertex $v_i$. A graph state $|G\rangle$
is the $n$-qubit state stabilized by all $g_{i}$, i.e.,
\begin{equation}
g_{i}|G\rangle=|G\rangle,\text{for all }i=1,...,n.
\end{equation}
The $n$ graph state stabilizer generators, $g_{i}$, generate the whole stabilizer group
$\left(  \mathcal{S}_{G},\cdot\right)  $ of $|G\rangle$ with multiplication as
its group operation. The group $\mathcal{S}_{G}$ is Abelian and contains
$2^{n}$ elements. These $2^{n}$ stabilizers uniquely represent a graph state
on $n$ vertices.
Let us define the``induced stabilizer'', which is uniquely associated with a given vertex subset.

\begin{definition}[Induced stabilizer]\label{def::induced_stabilizer}
\hiddengls{inducedS}
Let $G$ be a graph on vertices $V_{G}=\left\{
v_{1},v_{2},\cdots,v_{n}\right\}  $. Let $\xi$ be a subset of $V_{G}$. We call the product of all $g_{i}$ with $i\in\xi$, i.e.
\begin{equation}
s_{G}^{(\xi)}:=\prod_{i\in\xi}g_{i}, \label{eq::induced_stabilizer}%
\end{equation}
the \emph{$\xi$-induced stabilizer} of the graph state $|G\rangle$. Here,
$g_{i}$ is the graph state stabilizer generator of $|G\rangle$ associated with
the $i$-th vertex.
\end{definition}

Since this $\xi$-induction map is bijective,
it maps the group $\left(  \mathcal{P}\left(  V_{G}\right)  ,\Delta\right)  $ into
the stabilizer group $\left(  \mathcal{S}_{G},\cdot\right)  $,
where $\mathcal{P}\left(  V_{G}\right)  :=\left\{  \xi\subseteq V_{G}\right\}
$ is the power set (the set of all subsets) of $V_{G}$ and $\Delta$ is the symmetric difference operation
acting on two sets as $\xi_1\Delta\xi_2=(\xi_1\setminus\xi_2)\cup(\xi_2\setminus\xi_1)$.

\begin{proposition}
[Isomorphism of $\xi$-induction]%
\label{prop::isomorphism_of_vertex-induction_operation}
\Needspace*{7\baselineskip}
Let $\left(
\mathcal{S}_{G},\cdot\right)  $ be the stabilizer group of a graph state
$|G\rangle$ and $\mathcal{P}(V_{G})$ be the power set of the vertex set of $G$.
The vertex-induction operation $s_{G}^{(\xi)}$ is a group isomorphism between
$(\mathcal{P}(V_{G}),\Delta)$ and $(\mathcal{S}_{G},\cdot)$, i.e.%
\begin{equation}
(\mathcal{P}(V_{G}),\Delta)\overset{s_{G}^{(\xi)}}{\sim}\left(  \mathcal{S}%
_{G},\cdot\right)  ,
\end{equation}
where $\Delta$ is the symmetric difference operation.

\begin{proof}
See Appendix \ref{sec::graph_states_proofs}.
\end{proof}
\end{proposition}

The summation operation maps the stabilizer group $\mathcal{S}_{G}$ to its
stabilized space, i.e. the density matrix of the graph state $|G\rangle$ \cite{GraphStateReviews2006},
\begin{equation}
\mathcal{S}_{G}\overset{\Sigma}{\longrightarrow}|G\rangle\langle G|=\frac
{1}{2^{n}}\sum_{s\in S_{G}}%
s.\label{eq::graph_state_density_matrix_n_stabilizers}%
\end{equation}
Hence there also exists an operation mapping the group $\mathcal{P}(V_{G})$ to
graph states%
\begin{equation}
\mathcal{P}(V_{G})\overset{\Sigma\circ s_{G}^{(\xi)}}{\longrightarrow
}|G\rangle\langle G|=\frac{1}{2^{n}}\sum_{\xi\subseteq V_{G}}s_{G}^{(\xi)}
=\prod_{i=1}^{n}\frac{1+g_{i}}{2}.
\label{eq::stabilizer_projected_graph_state}
\end{equation}
%
\par
This is a well-known representation of graph states \cite{GraphStateReviews2006}.
The representation of a graph state in the computational Z-basis $|i_{Z}\rangle$ \cite{VandenNest2005-PhDthesis} is given by
\begin{equation}
\mathcal{P}(V_{G})\overset{\Sigma\circ s_{G}^{(\xi)}\text{ in }|i_{Z}\rangle
}{\longrightarrow}|G\rangle=\frac{1}{2^{n/2}}\sum_{i\in\{0,1\}^{\otimes n}%
}\left(  -1\right)  ^{\left\langle i,i\right\rangle _{A_{G}}}|i_{Z}%
\rangle.\label{eq::graph_states_in_Z-basis}%
\end{equation}
Here $\sigma_{z}^{\otimes n}|i_{Z}\rangle=\left(  -1\right)
^{\left\vert i\right\vert }|i_{Z}\rangle$, where $\left\vert i\right\vert $
is the Hamming weight of $i$.
\hiddengls{adjMatG}
$A_{G}$ is the adjacency matrix of the graph
$G$, and $\braket{i,i}_{A_{G}}=(i_{1},...,i_{n})A_{G}(i_{1},...,i_{n})^{\mathrm{T}}%
$. For all graph states with $n$ vertices, the probability amplitudes of
Z-basis states $\langle i_{Z}|G\rangle$ are homogenously distributed for all
$|i_{Z}\rangle$ up to a phase $-1$, i.e. $|\braket{i_Z|G}|=1/2^{n/2}$.
Therefore graph states with the same vertex set all
have the equivalent probability distribution of local $\sigma_Z$-measurement outcomes.
This means that the Z-basis representation conceals the inner structure of graph states.
\bigskip
\par
Different from the Z-basis, the representation of graph states in the computational
X-basis $|i_{X}\rangle$ ( i.e. $\sigma_{X}^{\otimes n}|i_{X}\rangle=\left(
-1\right)  ^{\left\vert i\right\vert }|i_{X}\rangle$) reveals the structure of graph states to a certain degree.
One aim in this paper is to find an efficient algorithm, i.e. a
mapping from $\mathcal{P}(V_{G})$ to $|G\rangle$, to represent graph states in
the computational X-basis:%
\begin{equation}
\mathcal{P}(V_{G})\overset{?}{\longrightarrow}|G\rangle \text{ in }|i_{X}\rangle.
\label{eq::aim_of_paper}%
\end{equation}
In the rest of the paper, we denote the X-basis $|i_{X}\rangle$ as $|i\rangle$;
that is, $|0\rangle=|+_{Z}\rangle=(\ket{0_Z}+\ket{1_Z})/\sqrt{2}$ and $|1\rangle=|-_{Z}\rangle
=(\ket{0_Z}-\ket{1_Z})/\sqrt{2}$.

\section{X-chains and their properties}
\label{sec::X-chains}
%
The commutativity of the measurement setting with graph state stabilizers determines
whether one can obtain information about a graph state in the laboratory.
Graph state stabilizers that commute with $\sigma_{X}$-measurements are the stabilizers consisting of solely $\sigma_{X}$ operators.
They are the key ingredient in the representation of graph states in the X-basis.
We will call the vertex sets $\xi$ inducing such configurations \emph{X-chains} of graph states.
In this section, the concept of X-chains will be introduced and their properties will be investigated.
\bigskip
\par
The number of $\sigma_Z$-operators in the graph state stabilizer $s_{G}^{(\xi)}$ depends on the neighborhoods within the vertex set $\xi$.
If a vertex $v$ has an even number of neighbors within $\xi$, then the Pauli operator $\sigma_{Z}^{(v)}$ appears an even number of times in $s_{G}^{(\xi)}$, such that the product becomes the identity.
Therefore to find the X-chain configurations of graph states, one needs to study the symmetric difference of neighborhoods within the vertex set $\xi$, which we define as the \emph{correlation index} of $\xi$ as follows.
\begin{definition}[Correlation index]
\label{def::correlation_index}
\Needspace*{6\baselineskip}%
\hiddengls{corrIndex}
Let $\xi$ be a vertex subset of a graph $G$. Its correlation index is defined as the symmetric difference of neighbourhoods within $\xi$,%
  \begin{equation}
    c_{\xi}:=N_{v_{1}}\Delta N_{v_{2}}\cdots\Delta N_{v_{k}},%
  \end{equation}
where $N_{v_i}$ is the neighbourhood of $v_i$ and $\xi=\{v_1,...,v_k\}$.
\end{definition}%
\noindent
The name ``correlation index'' will become clearer in Theorem \ref{theorem::graph_states_as_X-factorized_states} and refers to the fact that for vanishing correlation index the corresponding stabilized state is factorized.
(These states are called X-chain states in Def. \ref{def::X-chain_states_K-correlation_states}.)
Note that the set $c_\xi$ occurs as an ``index'' for the $\sigma_Z$ operator of the induced stabilizer $s_G^{(\xi)}$ (see Proposition \ref{prop::induced_stabilizer_math_formula}).
\par
Besides the correlation index, due to the anticommutativity of $\sigma_{X}$ and
$\sigma_{Z}$, the graph state stabilizers also depend on the so-called
\emph{stabilizer parity} of $\xi$.
\begin{definition}[Stabilizer parity]
\label{def::stabilizer_parity}
\Needspace*{6\baselineskip}%
\hiddengls{stabParity}%
Let $\xi$ be a vertex subset of a graph $G$.
Its stabilizer parity in $\ket{G}$ is defined as the parity of the edge number $\left\vert E_{G[\xi]}\right\vert $ of the $\xi$-induced subgraph $G[\xi]$
  \begin{equation}
    \pi_{G}\left(  \xi\right)  :=(-1)^{\left\vert E_{G[\xi]}\right\vert}.%
    \label{eq::G-parity_formula}%
  \end{equation}
\end{definition}%
\noindent
The stabilizer parity of $\xi$, $\pi_{G}\left(  \xi\right)$ is positive if the edge number $E(G[\xi])$ is even, otherwise it is negative.
The explicit form of the induced stabilizers is given in the following proposition.
\begin{proposition}
[Form of the induced stabilizer]%
\label{prop::induced_stabilizer_math_formula}
\Needspace*{6\baselineskip}%
Let $\xi$ be a vertex subset of a graph $G$. The $\xi$-induced stabilizer (see Def. \ref{def::induced_stabilizer}) of a graph
state $|G\rangle$ is given by
\begin{equation}
s_{G}^{(\xi)}=\pi_{G}\left(  \xi\right)  \sigma_{X}^{(\xi)}\sigma_{Z}%
^{(c_{\xi})},\label{eq::induced_stabilizer_math_formula}%
\end{equation}
where $c_{\xi}$ is the \emph{correlation index} of $\xi$
and $\pi_{G}\left(\xi\right)$ is the stabilizer parity of $\xi$.
\begin{proof}
See Appendix \ref{sec::graph_states_proofs}.
\end{proof}
\end{proposition}
\begin{figure*}[t!]
\subfloat[]{
\includegraphics[width=0.15\textwidth]{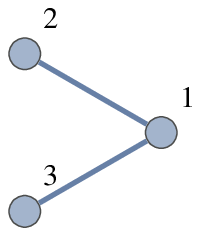}
\label{fig::S3_labeled}
}
\subfloat[]{
\includegraphics[width=0.8\textwidth]{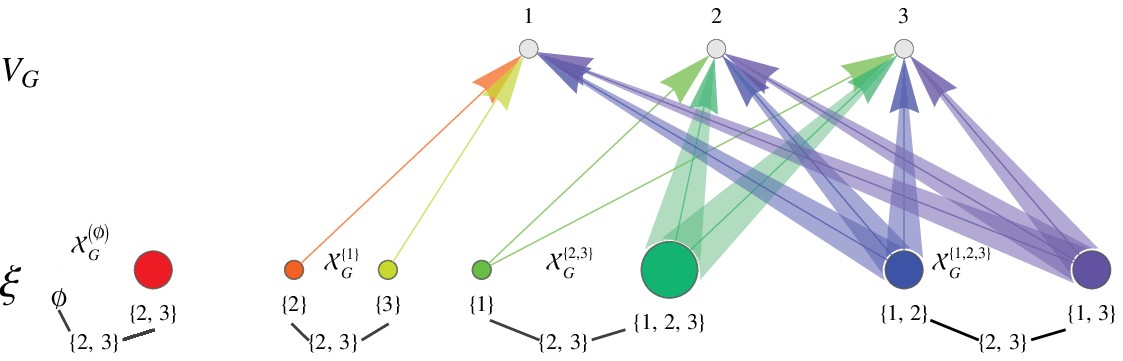}
\label{fig::correlation_index_and_x-resource_InStr}
\label{fig::S3_X-chain_illustation}
}
\newline
\subfloat[]{
\begin{tabular}[c]{|c|c|c|c|c|}
\toprule
$\xi$ & $\emptyset$ and $\{2,3\}$ & $\{2\}$ and $\{3\}$ & $\{1\}$ and $\{1,2,3\}$ & $\{1,2\}$ and
$\{1,3\}$ \\
\midrule
$c_{\xi}\in\mathcal{C}_G$ & $\emptyset$ & $\{1\}$ & $\{2,3\}$ & $\{1,2,3\}$\\
\midrule
$|E_{G[\xi]}|$ & $0$ & $0$ & $0$ and $2$ & $1$ \\
\midrule
$\pi_{G}\left(  \xi\right)  $ & $1$ & $1$ & $1$ & $-1$ \\
\midrule
$s_{G}^{(\xi)}$ & $\id$, $\sigma_{X}^{\{2,3\}}$ & $\sigma_{X}^{\{2\}}\sigma_{Z}^{\{1\}}$, $\sigma_{X}^{\{3\}}\sigma_{Z}^{\{1\}}$
& $\sigma_{X}^{(1)}\sigma_{Z}^{\{2,3\}}$, $\sigma_{X}^{\{1,2,3\}}\sigma_{Z}^{\{2,3\}}$
& $-\sigma_{x}^{\{1,2\}}\sigma_{Z}^{\{1,2,3\}}$, $-\sigma_{X}^{\{1,3\}}\sigma_{Z}^{\{1,2,3\}}$ \\
\midrule
$\xi\in\langle\mathcal{K}_G\rangle$ & $\emptyset$ & $\{2\}$ & $\{1\}$ & $\{1,2\}$ \\
\midrule
 & \multicolumn{4}{c|}{ $\Gamma_G=\{\{2,3\}\}$, $\mathcal{K}_G=\{\{1\}, \{2\}\}$ }\\
\bottomrule
\end{tabular}
\label{fig::correlation_index_and_x-resource_table}
}
\caption{\colorfig Correlation indices and X-resources:
(a) $3$-vertex star graph.
(b) The mapping from X-resources to correlation indices is illustrated in the incidence structure
\cite{Rosen1999combinatorial} of the graph $S_3$. The upper line is
the correlation index, while the lower line are the vertex subsets (including
the empty set). The arrows go from lower vertex subsets $\xi$ to the
upper vertices corresponding to the nonzero entries of their correlation
index $c_{\xi}$. For example, the vertex set $\{1,2,3\}$ points to the vertices $\{2,3\}$, indicating that the correlation index of $\{1,2,3\}$ is $c_{\{1,2,3\}}=\{2,3\}$. In particular, the vertex set $\emptyset$
and $\{2,3\}$ are X-chains (see Def. \ref{def::X-chain}), since their correlation index is $\emptyset$.
The resources in the sets $\mathcal{X}_G^{(\emptyset)}=\{\emptyset,\{2,3\}\}$, $\mathcal{X}_G^{\{1\}}=\{\{2\},\{3\}\}$, $\mathcal{X}_G^{\{2,3\}}=\{\{1\},\{1,2,3\}\}$ and $\mathcal{X}_G^{\{1,2,3\}}=\{\{1,2\},\{1,3\}\}$ are all ``connected'' by $\{2,3\}$ via the symmetric difference operation $\Delta$.
(c) Grouping of vertex subsets according to the correlation index. $\Gamma_G$ and $\mathcal{K}_G$ are the X-chain group generators and correlation group generators, respectively.
}
\label{fig::correlation_index_and_x-resource}
\end{figure*}
Let us illustrate these concepts with an example.
The star graph state $|S_{3}\rangle$ is shown in Fig. \ref{fig::S3_labeled}.
Its stabilizers can be represented in the following binary matrix:
\[
\left(
\begin{array}
[c]{ccc|ccc}%
0 & 0 & 0 & 0 & 0 & 0\\
1 & 0 & 0 & 0 & 1 & 1\\
0 & 1 & 0 & 1 & 0 & 0\\
0 & 0 & 1 & 1 & 0 & 0\\
1 & 1 & 0 & 1 & 1 & 1\\
1 & 0 & 1 & 1 & 1 & 1\\
0 & 1 & 1 & 0 & 0 & 0\\
1 & 1 & 1 & 0 & 1 & 1
\end{array}
\right)  ,
\]
in which each row represents a stabilizer.
The bit strings on the left hand side
of the divider are the possible vertex sets $\xi$ occurring as a superscript for the Pauli
$\sigma_{X}$ operators in Eq. \eqref{eq::induced_stabilizer_math_formula}, while the right hand side is their correlation
indices $c_{\xi}$ occurring as a superscript for the Pauli $\sigma_{Z}$ operators.
This is the so-called binary representation of graph states
\cite{Gottesman1997-phdQEC, Gottesman1996-QECSQHB,
CalderbankRSSloane1997-QECOrthGeo}.
We interpret this binary representation as an incidence structure
\cite{Rosen1999combinatorial} in Fig.
\ref{fig::correlation_index_and_x-resource_InStr}, in which the vertex sets
$\xi$ are depicted as the nodes in the lower row, while the upper row
interprets the correlation indices $c_{\xi}$ . In the example of
$|S_{3}\rangle$, one observes that the correlation indices $c_{\xi}$ do not
cover all possible $3$-bit binary numbers. The vertex subsets are regrouped
according to their correlation indices in Fig.
\ref{fig::correlation_index_and_x-resource_table}. The concept of regrouping
is introduced via the definition of the so-called \emph{X-resources} as follows.
\begin{definition}
[X-resources of correlation indices]%
\label{def::X-resouces_of_correlation_indices}
\Needspace*{8\baselineskip}
We denote the set of \emph{correlation indices} of a graph $G$ as%
\[
\mathcal{C}_{G}:=\left\{  c_{G}(\xi):\xi\subseteq V_{G}\right\}  .
\]
If a vertex set $\xi$ has correlation index $c$, i.e. $c_{G}(\xi)=c$, then we
call $\xi$ an \emph{(X-)resource of $c$-correlation} in $G$. The
\emph{(X-)resource set of $c$-correlation} is written as
\begin{equation}
\mathcal{X}_{G}^{\left(  c\right)
}:=\left\{  \xi\subseteq V_{G}:c_{G}(\xi)=c\right\}  .
\end{equation}
\hiddengls{corrIndexSet}\hiddengls{XResourceSet}%
\end{definition}%
\noindent
Since in the example of $\ket{S_3}$ the correlation index of $\{2,3\}$ is $\emptyset$, each correlation index $c\in\mathcal{C}_{S_3}$ has two X-resources $\xi^{(c_1)}$ and $\xi^{(c_2)}$ with $\xi^{(c_1)}=\xi^{(c_2)}\Delta\{2,3\}$.
The number of X-resources of $\ket{S_3}$ is $2^3$.
Therefore the graph state $|S_{3}\rangle$ generates $4$ correlation indices corresponding to $4$ binary numbers.
The other $4$ correlation indices are excluded due to the existence of the non-trivial $\emptyset$-correlation resource $\{2,3\}$.
This non-trivial $\emptyset$-correlation resource decreases
the correlations of the graph state in the X-basis.
Explicitly a non-trivial $\emptyset$-correlation resource induces a stabilizer consisting solely of $\sigma_{X}$ operators as follows.
\[
s_{G}^{(\xi)}=\pi_{G}\left(  \xi\right)  \sigma_{X}^{(\xi)}\text{, for all
}\xi\in\mathcal{X}_{G}^{\left(  \emptyset\right)  }.
\]
We will call such vertex sets \emph{X-chains}.
\begin{definition}
[X-chains]\label{def::X-chain}
\Needspace*{3\baselineskip}
Let $|G\rangle$ be a graph state. An \emph{X-resource of $\emptyset$-correlation} in $G$ is
called an \emph{X-chain} of $G$. The set of all X-chains is denoted as
$\mathcal{X}_{G}^{\left(  \emptyset\right)  }$.
\hiddengls{Xchain}
\end{definition}
The X-chains of a graph state $\ket{G}$ can be determined efficiently with standard linear algebra tools \cite{WuKampermannBruss2015-XChainAlgo}.
As examples, the X-chains of the graph states $|S_{3}\rangle$, $|S_{4}\rangle$ and $|C_{3}\rangle$ are given in Table \ref{table::x-chain_group_of_graph_state_eg}.
Besides, the X-chains for certain types
of graph states (i.e. linear graph states $\ket{L_n}$, cycle graph states $\ket{C_n}$, complete graph states $\ket{K_n}$ and star graph states $\ket{S_n}$) are studied in
\cite{WuKampermannBruss2015-XChainAlgo}.
A MATHEMATICA package is provided for finding X-chains in general graph states; see Supplemental Material \cite{XchainMpackage_Wu}.
\begin{table}[t]

\newlength{\Xgraphwidth}
\setlength{\Xgraphwidth}{0.18\columnwidth}
\newlength{\thirdcolwidth}
\setlength{\lastcolwidth}{0.9\columnwidth-3\Xgraphwidth-\thirdcolwidth}

\def\gwidth{0.7}
\begin{tabular}
[c]{
|C{\Xgraphwidth}
|C{2\Xgraphwidth}
|C{\lastcolwidth}
|
}%
\toprule

Graph $G$ & X-chains $\braket{\Gamma_G}$ & X-chain generators $\Gamma_G$\\
\midrule
\includegraphics[width=\linewidth]{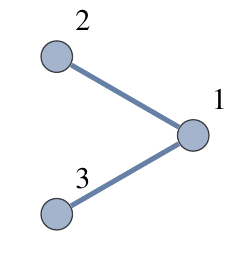}
\newline
$S_3$
&%
\includegraphics[width=0.9\linewidth]{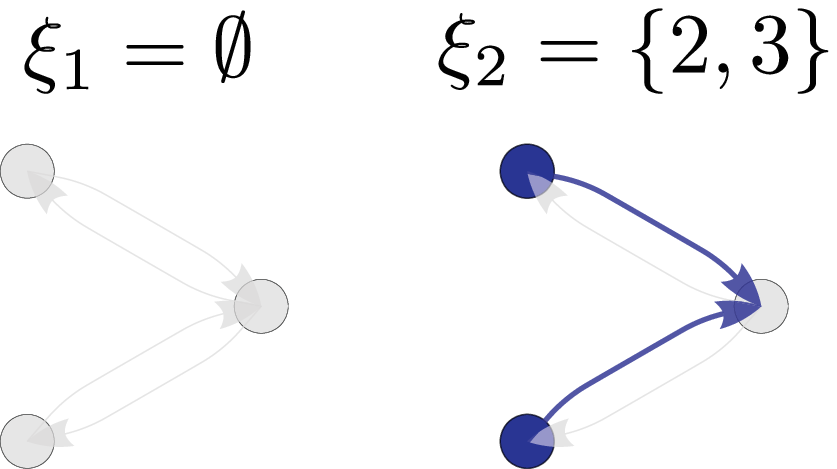}
& \{\{2,3\}\}\\

\hline
\includegraphics[width=\linewidth]{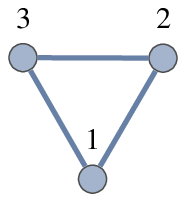}
\newline
$C_3$
&%
\includegraphics[width=0.9\linewidth]{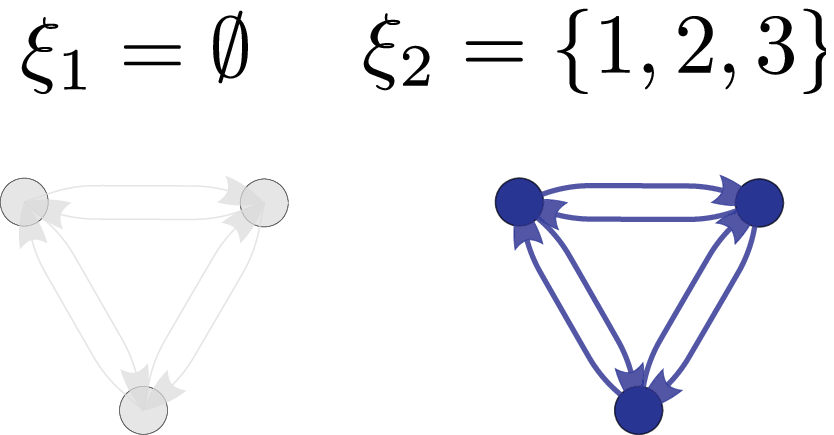}
& $\{\{1,2,3\}\}$\\

\hline
\includegraphics[width=\linewidth]{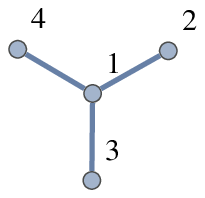}
\newline
$S_4$
&%
\includegraphics[width=0.9\linewidth]{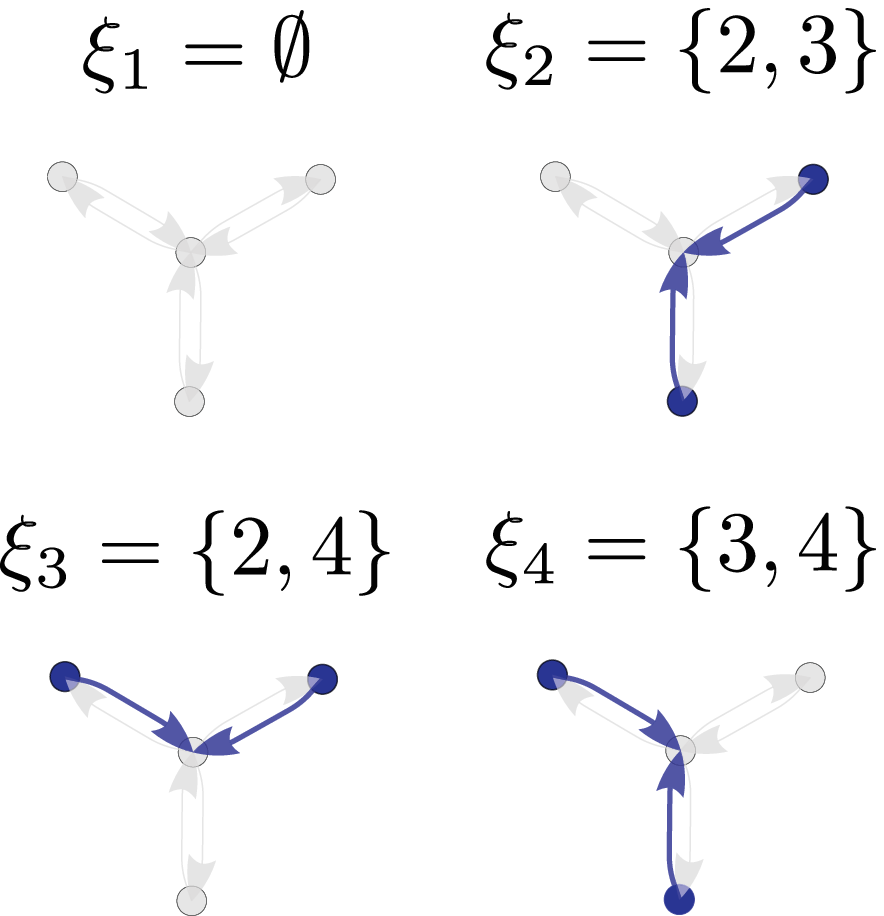}
& $\{\{2,3\},\{2,4\}\}$ \\
\bottomrule
\end{tabular} 
\caption{X-chain groups of simple graphs: The directed graphs shown under the X-chains illustrate the criterion of X-chains. Once a vertex is selected in a vertex subset, one draws arrows from it to its neighbors. A vertex subset $\xi$ is an X-chain if and only if all vertices of the graph are incident by even number of arrows.
The X-chain groups $\braket{\Gamma_G}$ are generated by their generating sets $\Gamma_G$.}%
\label{table::x-chain_group_of_graph_state_eg}%
\end{table}
\par
We point out that the X-chains form a group with the symmetric difference operation.
%
\begin{lemma}
[X-chain groups and correlation groups]%
\label{lemma::X-chain_factorized_group}
\Needspace*{15\baselineskip}
Let $|G\rangle$ be a graph state.
The set of X-chains together with the symmetric difference
$(\mathcal{X}_{G}^{\left(  \emptyset\right)  },\Delta)$,
is a normal subgroup of $\left(  \mathcal{P}\left(  V_{G}\right)  ,\Delta\right)$.
The quotient group $(\mathcal{P} \left(  V_{G}\right)  /\mathcal{X}_{G}^{\left(\emptyset\right)},\Delta)$ is identical to the set of all resource sets
\begin{equation}
\mathcal{P}\left(  V_{G}\right)  /\mathcal{X}_{G}^{\left(  \emptyset\right)
}=\left\{  \mathcal{X}_{G}^{\left(  c\right)  }:c\in\mathcal{C}_{G}\right\},
\end{equation}
which we call call the \emph{correlation group} of $\ket{G}$.
Let $\Gamma_{G}$ and $\mathcal{K}_{G}$ denote the generating sets of
$(\mathcal{X}_{G}^{\left(  \emptyset\right)  },\Delta)$ and $(\mathcal{P}(V_{G})/\mathcal{X}_{G}^{(\emptyset)}, \Delta
)$, respectively.
The stabilizer group $\left(  \mathcal{S}_{G},\cdot\right)  $ is isomorphic to the direct product of the X-chain group and the correlation group,
\begin{equation}
\left(  \mathcal{S}_{G},\cdot\right)  \sim\left(  \left\langle \Gamma
_{G}\right\rangle ,\Delta\right)  \times\left(  \left\langle \mathcal{K}%
_{G}\right\rangle ,\Delta\right)
,\label{eq::prop_X-chain_group_correlation_group}%
\end{equation}
As a result, the graph state $|G\rangle$ is the product of the X-chain group and correlation group inducing stabilizers, i.e.%
\begin{equation}
|G\rangle\langle G|=\prod_{\kappa\in\mathcal{K}_{G}}\frac{1+s_{G}^{(\kappa)}%
}{2}\prod_{\gamma\in\Gamma_{G}}\frac{1+s_{G}^{(\gamma)}}{2}%
.\label{eq::X-chain_factorization_of_graph_state_projector}%
\end{equation}
\hiddengls{XchainGroup}\hiddengls{corrGroupG}\hiddengls{corrGroup}
\begin{proof}
See Appendix \ref{sec::graph_states_proofs}.
\end{proof}
\end{lemma}
Note that the brackets $\braket{\Gamma_G}$ and $\braket{\mathcal{K}_G}$ denote the group generated by $\Gamma_G$ and $\mathcal{K}_G$, respectively
The correlation group represents the partition of the powerset of vertex set $\mathcal{P}(  V_{G})$ regarding the correlation index of the vertex subsets $\xi\in\mathcal{P}(  V_{G})$.
The members in the correlation group $\xi\in\braket{\mathcal{K}_G}$ possess distinct correlation indices.
All of the members in the $c$-correlation resource set
$\xi\in\mathcal{X}^{(c)}$ are connected by X-chains. Let $\xi_{1}^{(c)}%
\in\mathcal{X}^{(c)}$ and $\xi_{2}^{(c)}\in\mathcal{X}^{(c)}$ be two
X-resources for the same correlation index $c$. Then there must exist an
X-chain $\gamma\in\Gamma_{G}$, such that
\begin{equation}
\xi_{2}^{(c)}=\xi_{1}^{(c)}\Delta\gamma.
\end{equation}
For instance, in the example of $|S_{3}\rangle$ (Fig.
\ref{fig::S3_X-chain_illustation}), the resources of correlation $i^{(c)}=111$
(i.e. $c=\left\{  1,2,3\right\}  $) are connected by the X-chain $\{2,3\}$,
i.e. $\{1,3\}=\{1,2\}\Delta\{2,3\}$.
Therefore, one can choose one member in
$\mathcal{X}^{(c)}$ to represent the whole resource set $\mathcal{X}^{(c)}$.
Hence, after the X-chain factorization the group $\left(  \mathcal{P}\left(
V_{G}\right)  ,\Delta\right)  $ for $S_{3}$ becomes $\left\langle
\mathcal{K}_{G}\right\rangle =\left\{  \emptyset,\left\{  1\right\}  ,\left\{
2\right\}  ,\left\{  1,2\right\}  \right\}  $ with $\mathcal{K}_{G}=\left\{
\left\{  1\right\}  ,\left\{  2\right\}  \right\}  $
\par
In Eq. (\ref{eq::X-chain_factorization_of_graph_state_projector}), the Hilbert
space $\mathbb{H}_G$ of the graph state $\ket{G}$ is first projected onto the subspace stabilized by the stabilizers
$s_{G}^{(\gamma)}$ with $\gamma\in\Gamma_{G}$. It is the subspace,
$\mathrm{span}\left(  \Psi_{\emptyset}\right)  $, spanned by the stabilized states
$|\psi_{\emptyset}\rangle$ with%
\begin{equation}
\Psi_{\emptyset}:=\left\{  |\psi_{\emptyset}\rangle:s_{G}^{(\gamma)}%
|\psi_{\emptyset}\rangle=|\psi_{\emptyset}\rangle\text{, for all }\gamma
\in\Gamma_{G}\right\}  .
\end{equation}
In this projection, $\ket{\psi_{\emptyset}}$ are all product states, since every X-chain
stabilizer $s_{G}^{(\gamma)}$ commutes with the $\sigma^{\otimes n}_{X}$ operator. After the
first projection, the graph state is then obtained via projecting the subspace
$\mathrm{span}\left(  \Psi_{\emptyset}\right)  $ into the state that is stabilized by
the stabilizers $s_{G}^{(\kappa)}$ induced by the correlation group;
that is,%
\begin{equation}
\mathbb{H}_{G}\overset{\Gamma_{G}}{\longrightarrow}\Psi_{\emptyset}%
\overset{\mathcal{K}_{G}}{\longrightarrow}|G\rangle.
\end{equation}
This approach will be employed in the next section to derive the
representation of graph states in the X-basis.

\section{X-chain factorization of graph states}
\label{sec::X-chain_factorization_of_graph_states}
We express $|G\rangle$ in the X-basis as $|G\rangle=\sum_{i=1}^{2^n}\alpha_{i} |i_X\rangle$, with $\sum_i|\alpha_{i}|^2=1$.
Since the X-chain stabilizers $s_{G}^{(\gamma)}$ stabilize $|G\rangle$, it holds that
\begin{equation}
\sum\alpha_{i}|i_{X}\rangle=\sum \alpha_{i} s_{G}^{(\gamma)}|i_{X}\rangle.
\label{eq::no_name_1}
\end{equation}
Since $s_{G}^{(\gamma)}$ solely contains $\sigma_{X}$-operators, $s_{G}%
^{(\gamma)}|i_{X}\rangle=\pm|i_{X}\rangle$.
In order to fulfill Eq. \eqref{eq::no_name_1}, however, it follows that only the plus sign is possible, i.e.
\begin{equation}
s_{G}^{(\gamma)}|i_{X}\rangle=|i_{X}\rangle\text{ for all }\alpha_{i}\not =0.
\end{equation}
That means that the possible X-measurement outcomes are solely those X-basis states $|i_{G}\rangle$,
which are stabilized by all X-chain stabilizers $s_{G}^{(\gamma)}$.
A graph state $|G\rangle$ is hence a superposition of such particular X-basis states.
\par
For example, the star graph state $|S_{3}\rangle$ in Fig. \ref{fig::correlation_index_and_x-resource} is stabilized by the X-chain stabilizer $s_{S_{3}}^{\{2,3\}}=\sigma_{X}^{\{2,3\}}$. Therefore $\ket{S_3}$ belongs to the space spanned by the states stabilized by $s_{S_{3}}^{\{2,3\}}$.
From the table in Fig. \ref{fig::correlation_index_and_x-resource},
one observes that the X-basis $|i^{(c)}\rangle$, with $c\in\mathcal{C}_{S_{3}}$ (see Def. \ref{def::X-resouces_of_correlation_indices}) corresponding to the correlation indices of $|S_{3}\rangle$, are stabilized by $s_{S_{3}}^{\{2,3\}}$,
i.e. $\sigma_{X}^{(\{2,3\})}\ket{i^{(c)}}=\ket{i^{(c)}}$ for all $c\in\mathcal{C}_{S_{3}}$.
That means $|S_{3}\rangle$ belongs to the subspace, $\mathrm{span}(\Psi)$, spanned by $\Psi=\left\{  |i^{(c)}\rangle,c\in\mathcal{C}_{S_{3}}\right\}=\left\{  |000\rangle,|100\rangle,|011\rangle,|111\rangle\right\}  $.
Thus $|S_{3}\rangle$ can be represented solely in $4$ X-basis states instead of $8$ $Z$-basis states.
\bigskip
\par
In this section, we will derive a general mapping from the X-chain group and
correlation group to graph states in the X-basis.
This is the question we raised in section \ref{sec::representation_of_graph_states}.
We first introduce X-chain states and $\mathcal{K}%
$-correlation states (Definition
\ref{def::X-chain_states_K-correlation_states}), which span the subspace
stabilized by X-chain stabilizers and $\mathcal{K}$-correlation stabilizers,
respectively. Given the explicit form of the X-chain states and correlation states
in the X-basis (Proposition \ref{prop::X-chain_states_in_X-basis} and
\ref{prop::correlation_states_explicit_form}), one arrives at the X-chain
factorization representation of graph states in Theorem
\ref{theorem::graph_states_as_X-factorized_states}.

\begin{definition}
[X-chain states and correlation states]%
\label{def::X-chain_states_K-correlation_states}
\Needspace*{8\baselineskip}%
Let $|G\rangle$ be a graph
state with the X-chain group $\left\langle \Gamma_{G}\right\rangle $ and the
correlation group $\left\langle \mathcal{K}_{G}\right\rangle $.
We define the X-basis state $|i^{(x_{\Gamma_G})}\rangle$ (in short form $|i^{(x_{\Gamma})}\rangle$) as the state stabilized by the Pauli $\sigma_{X}$ operators such that
\begin{enumerate}
\item \label{item::def_X-chain_state_condition_1} $\pi_{G}\left(  \gamma\right)
\sigma_{X}^{(\gamma)}|i^{(x_{\Gamma})}\rangle=|i^{(x_{\Gamma})}\rangle$, for
all $\gamma\in\Gamma_{G}$,
\item \label{item::def_X-chain_state_condition_2} $\sigma_{X}^{(\kappa
)}|i^{(x_{\Gamma})}\rangle=|i^{(x_{\Gamma})}\rangle$, for all $\kappa
\in\mathcal{K}_{G}$.
\end{enumerate}
The local unitary transformed states
\begin{equation}
|\psi_{\emptyset}(\xi)\rangle=s_{G}^{(\xi)}\left\vert i^{(x_{\Gamma}%
)}\right\rangle ,\xi\in\braket{\mathcal{K}_G}\label{eq::def_X-chain_state}%
\end{equation}
are called \emph{X-chain states}.
Let $\braket{\mathcal{K}}\subseteq\braket{\mathcal{K}_G}$ be a correlation subgroup,
then a $\mathcal{K}$-\emph{correlation state} of graph state $|G\rangle$,
$|\psi_{\mathcal{K}}\left(  \xi\right)  \rangle$, is defined as%
\begin{equation}
|\psi_{\mathcal{K}}\left(  \xi\right)  \rangle=s_{G}^{(\xi)}\prod_{\kappa
\in\mathcal{K}}\frac{1+s_{G}^{(\kappa)}}{\sqrt{2}}|i^{(x_{\Gamma})}%
\rangle\label{eq::def_correlation_states}%
\end{equation}
with $\xi\in\braket{\mathcal{K}_G}/\braket{\mathcal{K}}$. Let $\left\langle
\mathcal{K}\right\rangle \mathcal{\mathcal{\subseteq}\left\langle
\mathcal{K}^{\prime}\right\rangle \mathcal{\subseteq}}\left\langle
\mathcal{K}_{G}\right\rangle $, where a set of $\mathcal{K}$-correlation states is
denoted as%
\begin{equation}
\Psi_{\mathcal{K}^{\prime}}^{(\mathcal{K)}}=\left\{  |\psi_{\mathcal{K}%
}\left(  \xi\right)  \rangle:\xi\in
\braket{\mathcal{K'}}/\braket{\mathcal{K}}\right\}
\label{eq::def_K_set_of_K-correlation_states}%
\end{equation}
\hiddengls{XchainState}
\hiddengls{XchainStateBasic}
\hiddengls{corrState}
\hiddengls{corrStateSet}
\end{definition}

In this notation, the set of X-chain states is then written as $\Psi
_{\mathcal{K}_{G}}^{(\emptyset)}$, or in short form as $\Psi^{(\emptyset)}$, while the
set of all $\mathcal{K}$-correlation states is denoted by $\Psi_{\mathcal{K}%
_{G}}^{(\mathcal{K)}}$, or in short form as $\Psi^{(\mathcal{K)}}$. Note that
$|\psi_{\mathcal{\emptyset}}\left(  \mathcal{\emptyset}\right)  \rangle
=|i^{(x_{\Gamma})}\rangle$, and X-chain states $|\psi_{\mathcal{\emptyset}%
}\left(  \xi\right)  \rangle$ are $\mathcal{\emptyset}$-correlation states.
The X-basis $|i^{(x_{\Gamma})}\rangle$ is the fundamental state from which the
non-vanishing X-basis components in graph states can be derived.
According to its definition, $|i^{(x_{\Gamma})}\rangle$ depends on the
generating set of the X-chain group and the correlation group of a given graph state.
One can employ the following approach to obtain the fundamental X-chain state $|i^{(x_{\Gamma})}\rangle$.
\begin{proposition}
[X-chain states in X-basis]\label{prop::X-chain_states_in_X-basis}
\Needspace*{8\baselineskip}
Let $|G\rangle$ be a graph state with the X-chain group $\left\langle \Gamma
_{G}\right\rangle $ and the correlation group $\left\langle \mathcal{K}%
_{G}\right\rangle $. Let $\Gamma_{G}=\left\{  \gamma_{1},\gamma_{2}%
,...\right\}  $, and $\gamma_{i}=\left\{  v_{i_{1}},v_{i_{2}},\cdots\right\}
$. The generating set $\Gamma_{G}$ and $\mathcal{K}_{G}$ can be chosen as

\begin{enumerate}
\item $\Gamma_G=\{\gamma_1,...,\gamma_k\}$ such that $\gamma_i\not\subseteq\gamma_j$ for all $\gamma_i, \gamma_j\in\Gamma_G$,
\item $\mathcal{K}_{G}=\left\{  \left\{  v\right\}  :v\in V_{G}\backslash
\bigcup_{i=1}^{k}\left\{  v_{i_{1}}\right\}  \right\}  $.
\end{enumerate}
Here, the first element of $\gamma_i=\{v_{i_1},v_{i_2},...\}$ is selected in a way such that $v_{i_{1}}\neq v_{j_{1}}$ for all $i\not =j$.
Then the X-chain state $\ket{\psi_{\emptyset}(\emptyset)}$ of $|G\rangle$ is an X-basis state, $|i^{(x_{\Gamma})}\rangle$, with%
\begin{equation}
x_{\Gamma}=\left\{  v_{i_{1}}:\pi_{G}\left(  \gamma_{i}\right)  =-1\right\}  .
\end{equation}

\begin{proof}
See Appendix \ref{sec::graph_states_proofs}.
\end{proof}
\end{proposition}
\noindent
The vertices $v_{i_{1}}$ are the key for the determination of $\ket{x_{\Gamma}}$.
First of all, we choose the X-chain generators $\Gamma_{G}$, such that
$\gamma_{i}\not \subseteq \gamma_{j}$, for all $\gamma_{i},\gamma_{j}\in
\Gamma_{G}$.
That means each X-chain generator possesses at least a vertex
$v_{i_{1}}$ as its own exclusive vertex, i.e. $v_{i_{1}}\in\gamma
_{i}\backslash\left(  \cup_{j\not =i}\gamma_{j}\right)  $.
In other words, the vertex $v_{i_{1}}$ uniquely represents the X-chain generator $\gamma_{i}$.
The correlation group generators are then chosen as the single vertex sets $\{v\}$ with $v\in V_G\setminus \bigcup_{i}\left\{  v_{i_{1}}\right\}$.
At the end, the corresponding vertex set $x_{\Gamma}$ of the fundamental X-chain
state $|i^{(x_{\Gamma})}\rangle$ is the set of $v_{i_{1}}$, whose X-chain
generator $\gamma_{i}$ possesses a negative stabilizer-parity. Note that in
general the choice of the X-chain generators $\Gamma_{G}$ is not unique,
and therefore neither are the fundamental X-chain states $|i^{(x_{\Gamma})}\rangle$.
However, the above mentioned approach still arrives to the same set $\Psi^{(\emptyset)}$ of X-chain states, since the X-chain group is unique.
\begin{figure*}[t!]
  \centering
  \subfloat[]{
  \includegraphics[width=0.2\textwidth]{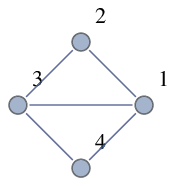}
  \label{fig::determination_of_X-chain_state_graph}
  }
  \subfloat[]{
  \Xchaindiagram[
    GState=|K_{4}^{\neg 1}\rangle,
    GammaG={\{\{1,2,3\},\{2,4\}\}},
    GammaGNum=2,
    xG=1000,
    KappaG={\{\{2\},\{3\}\}},
    KappaGNum=2,
    heightScale=2,
    widthScale=2,
    scale=0.5
  ]
  \label{fig::determination_of_X-chain_state_XchainDiagram}
  }
  \\
  \subfloat[]{
\begin{tabular}{C{0.4\textwidth}C{0.55\textwidth}}
\includegraphics[width=0.95\linewidth]{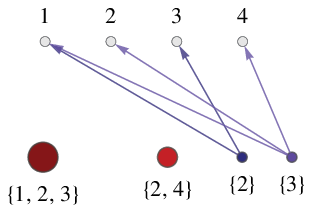} &
\begin{tabular}
[c]{|c|c|c|c|c|c|c|c|}
\toprule
$\gamma\in\Gamma_{G}$ & $\left\{  1,2,3\right\}  $ & $\left\{  2,4\right\}  $
& $\mathcal{K}_{G}$ & \multicolumn{4}{c|}{$\left\{  \left\{  2\right\}
,\left\{  3\right\}  \right\}  $}\\
\midrule
$v_{i_{1}}$ & $1$ & $2$ & $\xi\in\left\langle \mathcal{K}_{G}\right\rangle $ &
$\emptyset$ & $\left\{  2\right\}  $ & $\left\{  3\right\}  $ & $\left\{
2,3\right\}  $\\
\midrule
$\pi_{G}\left(  \gamma\right)  $ & $-1$ & $1$ & $\pi_{G}\left(  \xi\right)  $
& $1$ & $1$ & $1$ & $-1$\\
\midrule
$x_{\Gamma}$ & \multicolumn{2}{c|}{$\left\{  1\right\}  $} &
$i^{(c_{\xi})}$ & $0$ & $1010$ & $1101$ & $0111$\\
\midrule
$|i^{(x_{\Gamma})}\rangle$ & \multicolumn{2}{c|}{$|1000\rangle$} & $\ket{\psi_{\emptyset}\left(
\xi\right)}  $ & $|1000\rangle$ & $|0010\rangle$ & $|0101\rangle$ &
$-|1111\rangle$\\
\bottomrule
\end{tabular}
\\
\end{tabular}
\label{fig::determination_of_X-chain_state_GammaG_choice_1}
  }
  \\
  \subfloat[]{
\begin{tabular}{C{0.4\textwidth}C{0.6\textwidth}}
\includegraphics[width=0.95\linewidth]{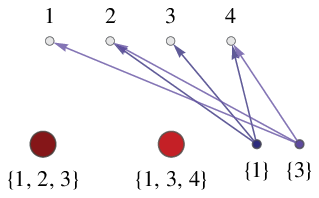} &
  \begin{tabular}
[c]{|c|c|c|c|c|c|c|c|}
\toprule
$\gamma\in\Gamma_{G}$ & $\left\{  2,1,3\right\}  $ & $\left\{  4,1,3\right\}
$ & $\mathcal{K}_{G}$ & \multicolumn{4}{c|}{$\left\{  \left\{  1\right\}
,\left\{  3\right\}  \right\}  $}\\
\midrule
$v_{i_{1}}$ & $2$ & $4$ & $\xi\in\left\langle \mathcal{K}_{G}\right\rangle $ &
$\emptyset$ & $\left\{  1\right\}  $ & $\left\{  3\right\}  $ & $\left\{
1,3\right\}  $\\
\midrule
$\pi_{G}\left(  \gamma\right)  $ & $-1$ & $-1$ & $\pi_{G}\left(  \xi\right)  $
& $1$ & $1$ & $1$ & $-1$\\
\midrule
$x_{\Gamma}$ & \multicolumn{2}{c|}{$\left\{  2,4\right\}  $} &
$i^{(c_{\xi})}$ & $0$ & $0111$ & $1101$ & $1010$\\
\midrule
$|i^{(x_{\Gamma})}\rangle$ & \multicolumn{2}{c|}{$|0101\rangle$} & $\ket{\psi_{\emptyset}\left(
\xi\right)}  $ & $|0101\rangle$ & $|0010\rangle$ & $|1000\rangle$ &
$-|1111\rangle$\\
\bottomrule
\end{tabular}
  \\
\end{tabular}
 \label{fig::determination_of_X-chain_state_GammaG_choice_2}
  }
  \\
  \caption{\colorfig Example for the determination of X-chain states (see section \ref{sec::X-chain_factorization_of_graph_states} for details).
  (a) The graph state $|K_{4}^{\neg 1}\rangle$.
  (b) The factorization diagram of $|K_{4}^{\neg 1}\rangle$ (for an explanation, see Algorithm \ref{algo::factorization_diagram_graph_states}).
  (c),(d) The incidence structure of the X-chain generators and correlation group generators of $|K_{4}^{\neg 1}\rangle$. The choices of these generators are not unique, and lead to different fundamental X-chain states $\ket{i^{(x_{\Gamma})}}$. However, they arrive at the identical set of X-chain states $\{\ket{\psi_{\emptyset}(\xi)}, \xi\in\braket{\mathcal{K}_G}\}$.
  }
  \label{fig::determination_of_X-chain_states}
\end{figure*}
\par
Let us illustrate these concepts by an example, the graph state $|K_{4}^{\neg 1}\rangle$ (Fig. \ref{fig::determination_of_X-chain_state_graph}),
which corresponds to the graph with one edge missing from the complete graph $K_{4}$.
Its X-chain generators can be chosen as $\Gamma_{G}=\left\{  \gamma_{1},\gamma_{2}\right\}  =\left\{  \left\{  1,2,3\right\}  ,\left\{  2,4\right\}\right\}  $ (see Fig. \ref{fig::determination_of_X-chain_state_GammaG_choice_1}).
The exclusive vertex $v_{1}$ for $\gamma_{1}$ can be chosen as $1$, while $v_{2}$ for
$\gamma_{2}$ is $4$. Only $\gamma_{1}$ has negative parity, and therefore
$x_{\Gamma}=\left\{  1\right\}  $ and the fundamental X-chain state is
$\ket{i^{(x_{\Gamma})}}=|1000\rangle$.
\par
From the fundamental X-chain state $\ket{i^{(x_{\Gamma})}}$ one can derive all the X-chain states and correlation states with the following proposition.
\begin{proposition}
[Form of X-chain states, $\mathcal{K}$-correlation states]%
\label{prop::correlation_states_explicit_form}
\Needspace*{8\baselineskip}
Let $\xi\in\left\langle
\mathcal{K}_{G}\right\rangle $ be an X-resource and $\left\langle
\mathcal{K}\right\rangle \subseteq\left\langle \mathcal{K}_{G}\right\rangle $.
An X-chain state is given as
\begin{equation}
|\psi_{\emptyset}\left(  \xi\right)  \rangle=\pi_{G}\left(  \xi\right)
\left\vert i^{(x_{\Gamma})}\oplus i^{(c_{\xi})}\right\rangle,
\end{equation}
where $\pi_{G}\left(  \xi\right)  $ is the stabilizer parity of $\xi$ (see Eq. \eqref{eq::G-parity_formula}),
and $c_{\xi}$ is the correlation index of $\xi$.

A $\mathcal{K}$-correlation state is the superposition of \emph{X-chain
states},\emph{ }
\begin{equation}
|\psi_{\mathcal{K}}\left(  \xi\right)  \rangle=\frac{1}{2^{\left\vert
\mathcal{K}\right\vert /2}}\sum_{\xi^{\prime}\in\left\langle \mathcal{K}%
\right\rangle }|\psi_{\emptyset}\left(  \xi\Delta\xi^{\prime}\right)  \rangle.
\label{eq::correlation_state_explicit_form}
\end{equation}

\begin{proof}
See Appendix \ref{sec::graph_states_proofs}.
\end{proof}
\end{proposition}
\noindent
According to this proposition, the X-chain states of $\ket{K_4^{\lnot1}}$ derived from $\ket{i^{(x_\gamma)}}=\ket{1000}$ are given in the table in Fig.
\ref{fig::determination_of_X-chain_state_GammaG_choice_1}.
Alternatively, one can also choose the X-chain generators $\Gamma_{G}=\left\{  \gamma_{1}%
,\gamma_{2}\right\}  =\left\{  \left\{  2,1,3\right\}  ,\left\{
4,1,3\right\}  \right\}  $ (see Fig.
\ref{fig::determination_of_X-chain_state_GammaG_choice_2}). In this case
$v_{1}=2$ and $v_{2}=4$. The parities of $\gamma_{1}$ and $\gamma_{2}$ are
both negative, hence $|i^{(x_{\Gamma})}\rangle=|0101\rangle$. However, the sets
of obtained X-chain states $\Psi^{(\emptyset)}$ are identical in both cases.
\par
The correlation states $|\psi_{\mathcal{K}}\left(  \xi\right)  \rangle$ are then the superposition of their corresponding X-chain states. For example, in Fig. \ref{fig::determination_of_X-chain_state_GammaG_choice_1} the correlation state $|\psi_{\{\{2,3\}\}}\left(  \emptyset\right)  \rangle=(\ket{1000}-\ket{1111})/\sqrt{2}$. The correlation states have the following properties.
\begin{corollary}
[Properties of $\mathcal{K}$-correlation states]%
\label{coro::properties_of_correlation_states}
\Needspace*{8\baselineskip}
Let $\braket{\mathcal{K}}\subseteq
\braket{\mathcal{K}_G}$ be a correlation index subgroup, then

\begin{enumerate}
\item \label{item::correlation_states_property_1} $|\psi_{\mathcal{K}}\left(
\xi\right)  \rangle$ is stabilized by all stabilizers $s_{G}^{(\kappa)}$ with
$\kappa\in\left\langle \Gamma_{G}\right\rangle \times\left\langle
\mathcal{K}\right\rangle $%
\begin{equation}
s_{G}^{(\kappa)}|\psi_{\mathcal{K}}\left(  \xi\right)  \rangle=|\psi
_{\mathcal{K}}\left(  \xi\right)  \rangle.
\end{equation}
Therefore the space $\mathrm{span}(\Psi^{(\mathcal{K})})$, see Eq. \eqref{eq::def_K_set_of_K-correlation_states}, is also stabilized by $s_{G}^{(\kappa)}$ with $\kappa\in\left\langle \Gamma_{G}\right\rangle \times\left\langle
\mathcal{K}\right\rangle $.

\item \label{item::correlation_states_property_2} For $\xi_{1}\in\left\langle
\mathcal{K}_{G}\right\rangle $ and $\xi_{1}\not \in \left\langle
\mathcal{K}\right\rangle $, it holds that%
\begin{equation}
s_{G}^{(\xi_{1})}|\psi_{\mathcal{K}}\left(  \xi_{2}\right)  \rangle
=|\psi_{\mathcal{K}}\left(  \xi_{1}\Delta\xi_{2}\right)  \rangle.
\end{equation}

\item \label{item::correlation_states_property_3} For $\kappa\in
\mathcal{K}_{G}$ and $\kappa\not \in \mathcal{K}$, the $\mathcal{K
}\cup\left\{  \mathcal{\kappa}\right\}  $-correlation state can be obtained by%
\begin{equation}
|\psi_{\mathcal{K}\cup\{\kappa\}}\left(  \xi\right)  \rangle=\frac
{1+s_{G}^{(\kappa)}}{\sqrt{2}}|\psi_{\mathcal{K}}\left(  \xi\right)  \rangle.
\end{equation}

\end{enumerate}

\begin{proof}
See Appendix \ref{sec::graph_states_proofs}.
\end{proof}
\end{corollary}
With these properties one can derive the representation of graph states in the X-basis.
\begin{theorem}
[X-chain state representation of graph states]%
\label{theorem::graph_states_as_X-factorized_states}
\Needspace*{8\baselineskip}
Let $|G\rangle$ be a graph
state. Then $|G\rangle$ is a $\mathcal{K}_{G}$-correlation state, which is a
superposition of X-chain states $|\psi_{\emptyset}\left(  \xi\right)  \rangle
$, i.e.%
\begin{equation}
|G\rangle=|\psi_{\mathcal{K}_{G}}\rangle=\frac{1}{2^{\left\vert \mathcal{K}%
_{G}\right\vert /2}}\sum_{\xi\in\left\langle \mathcal{K}_{G}\right\rangle
}|\psi_{\emptyset}\left(  \xi\right)  \rangle.
\label{eq::graph_state_in_X_basis}
\end{equation}
\begin{proof}
According to property \ref{item::correlation_states_property_1} in Corollary
\ref{coro::properties_of_correlation_states}, one can infer that
$\ket{\psi_{\mathcal{K}_{G}}}$ is stabilized by all graph state stabilizers
$s_{G}^{(\xi)}$ with $\xi\in\left\langle \Gamma_{G}\right\rangle
\times\left\langle \mathcal{K}_{G}\right\rangle $. As a result of Lemma
\ref{lemma::X-chain_factorized_group}, $\ket{\psi_{\mathcal{K}_{G}}}$ is stabilized
by the whole graph state stabilizer group $\mathcal{S}_{G}$. According to the definition
of graph states in the stabilizer formalism, one can infer that $|G\rangle
=\ket{\psi_{\mathcal{K}_{G}}}$. The explicit form of $\ket{\psi_{\mathcal{K}_{G}}}$ in Eq. \eqref{eq::graph_state_in_X_basis} is obtained by Proposition \ref{prop::correlation_states_explicit_form}.
\end{proof}
\end{theorem}
\par
Note that the graph state obtained by this theorem may differ from the real one by a global phase $-1$, i.e. $\ket{G}=-\ket{\psi_{\mathcal{K}_G}}$ \footnote{This global phase can be corrected by the sign of the sum of the parities of all X-resources in the correlation group $\braket{\mathcal{K}_G}$, $\alpha=Sign(\sum_{\xi\in\braket{\mathcal{K}_G}}(\pi_G(\xi)))$, i.e. $\ket{G}=\alpha\ket{\psi_{\mathcal{K}_G}}$.}.
We summarize the approach of X-chain factorization of a graph state representation in a so-called factorization diagram.
\begin{algorithm}
[Factorization diagram]\label{algo::factorization_diagram_graph_states}
\Needspace*{8\baselineskip}
The X-chain factorization of graph states can be described in the \emph{factorization diagram} shown in Fig. \ref{fig::factorization_diagram_graph_states}.
\begin{enumerate}
  \item One decomposes the group $\mathcal{P}(V_G)$ into the direct product of the X-chain group $\braket{\Gamma_G}$ and the correlation group $\braket{\mathcal{K}_G}$ (Lemma \ref{lemma::X-chain_factorized_group}).
  \item From the X-chain group $\braket{\Gamma_G}$, one obtains the set of X-chain states $\Psi^{\emptyset}_{\mathcal{K}_G}$ (Proposition \ref{prop::X-chain_states_in_X-basis}).
  \item From the correlation group $\braket{\mathcal{K}_G}$, one obtains graph states via the superposition of the X-chain states in $\Psi^{\emptyset}_{\mathcal{K}_G}$ (Theorem \ref{theorem::graph_states_as_X-factorized_states}).
\end{enumerate}
\end{algorithm}
\begin{figure}[ht!]
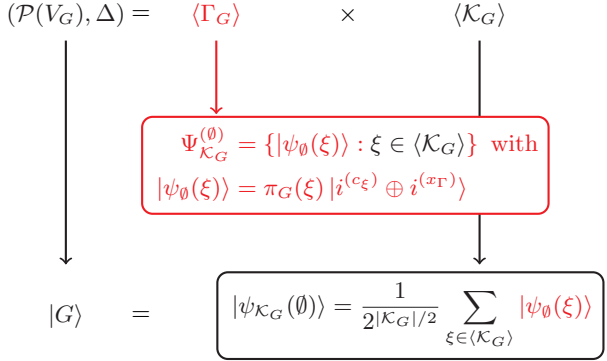

  \centering
  \Xchaindiagram
  \caption{\colorfig X-chain factorization diagram of graph states: A graphical summary of Proposition \ref{prop::X-chain_states_in_X-basis}, \ref{prop::correlation_states_explicit_form} and Theorem \ref{theorem::graph_states_as_X-factorized_states}. This diagram illustrates the algorithm for representing a graph state in the X-basis.}
  \label{fig::factorization_diagram_graph_states}
\end{figure}
\noindent
The arrows in the factorization diagram can be interpreted as a mapping from the sets of
X-resources to their corresponding stabilized Hilbert subspaces.
As we already discussed at the end of the section
\ref{sec::representation_of_graph_states}, a graph state is mapped from the
powerset of vertices by stabilizer induction, which is depicted in the left
hand side of the equality in the diagram. The equation in the first row is the
X-chain factorization of the group $(\mathcal{P}(V_{G}),\Delta)$ (Lemma
\ref{lemma::X-chain_factorized_group}).
The arrow from the X-chain group $\Gamma_{G}$ to the X-chain states $\Psi^{(\emptyset)}$ interprets the
mapping from the X-chain group to the stabilized subspace spanned by
$\Psi^{(\emptyset)}$ (Definition
\ref{def::X-chain_states_K-correlation_states} and Proposition
\ref{prop::X-chain_states_in_X-basis}).
The arrow from the correlation group $\langle\mathcal{K}_{G}\rangle$ through the X-chain states
$\Psi^{(\emptyset)}$ to the $\mathcal{K}_{G}$-correlation state is a mapping
from the subspace $\mathrm{span}(\Psi^{(\emptyset)})$ to the $\mathcal{K}_{G}%
$-correlation state $|\psi_{\mathcal{K}_{G}}\rangle$, which is stabilized by
the $\mathcal{K}_{G}$-stabilizers. This arrow-represented mapping is the
summation (superposition) of the X-chain states over the correlation group
$\braket{\mathcal{K}_G}$ (Proposition
\ref{prop::correlation_states_explicit_form}). Since the graph state $\ket{G}$
is the only stabilized state of the stabilizers induced by the group
$\braket{\Gamma_G}\times\braket{\mathcal{K}_G}$, it is identical to
the $\mathcal{K}_{G}$-correlation state $|\psi_{\mathcal{K}_{G}}\rangle$
(Theorem \ref{theorem::graph_states_as_X-factorized_states}), which is
represented by the equality of the last line in the factorization diagram.
With the help of the factorization diagram in Fig.
\ref{fig::determination_of_X-chain_state_XchainDiagram}, the graph state
$|K_{4}^{\lnot1}\rangle$ is given by%
\begin{equation}
|K_{4}^{\lnot1}\rangle=\frac{1}{2}\left(  \left\vert 1000\right\rangle
+\left\vert 0010\right\rangle +\left\vert 0101\right\rangle -\left\vert
1111\right\rangle \right)  .
\end{equation}
\bigskip
\par
Since the edge number $|E_{G[\xi]}|$ is identical to the product $\langle i_{Z}^{(\xi)},i_{Z}^{(\xi)}\rangle_{A_{G}}$ in Eq. \eqref{eq::graph_states_in_Z-basis},
according to the definition of the stabilizer-parity (Def. \ref{def::stabilizer_parity}),
\begin{equation}
  \pi_{G}\left(  \xi\right) = (-1)^{\langle i_{Z}^{(\xi)},i_{Z}^{(\xi)}\rangle_{A_{G}}}.
\end{equation}
Hence the representation of graph states in the Z-basis in Eq.
\eqref{eq::graph_states_in_Z-basis} can be reformulated as%
\begin{equation}
|G\rangle=\frac{1}{2^{n/2}}\sum_{\xi\subseteq V_{G}}\pi_{G}\left(  \xi\right)
|i_{Z}^{(\xi)}\rangle.
\label{eq::graph_states_in_Z-basis_with_pi}
\end{equation}
Comparing this Z-representation with the representation of a graph state in the X-basis given
in Eq. \eqref{eq::graph_state_in_X_basis}, the number of terms in the
representation is reduced from $2^{|V_{G}|}$ to $2^{|\mathcal{K}_{G}|}$.
The correlation group $\braket{\mathcal{K}_{G}}$ can be directly obtained if one knows the X-chain group.
The X-chain group can be searched by a criterion that the cardinality of the intersection of the vertex neighborhood with the X-chain $|N_v\cap \xi|$ should be even for all $v\in V_G$ \cite{WuKampermannBruss2015-XChainAlgo}.
The search of the X-chains of a graph state $\ket{G}$ is equivalent to finding the 2-modulus-kernel of the adjacency matrix of the graph $G$.
As this is efficient, the representation of graph states in the X-basis is feasible.
The larger the X-chain group that a graph state possesses, the smaller is its correlation group and hence
the more efficient is its X-chain factorization.
%
%
%
\par
Note that not every graph state has non-trivial X-chains (non-trivial means not the empty set).
For graph states without non-trivial X-chains, their X-chain factorization contains all X-basis states, and thus has the same difficulty as their Z-representation.
\par
Besides, the X-chain factorization of graph states in Theorem \ref{theorem::graph_states_as_X-factorized_states} implies that the possible outcomes of X-measurements are only the X-chain states,
$|\psi_{\emptyset}\left(  \xi\right)  \rangle$.
Consequently two graph states with different X-chains can have different X-chain states, and hence are distinguishable via the X-measurement outcomes.
In Table \ref{table::X-chain_states_of_3-vertex_graph_states}, we list the X-chain
generators and X-chain states of graph states with $3$ vertices.
Since the X-chain states of these graph states are different from each other,
one can therefore distinguish these $8$ graph states via local X-measurements with non-zero probability of success.
\def\threeVGSWidth{0.2\columnwidth}
\begin{table}[ht!]
\begin{tabular}
[c]{|C{\threeVGSWidth}|C{0.25\columnwidth}|C{0.5\columnwidth}|}%
\toprule

$\ket{G}$ & $\Gamma_{G}$ & $\Psi^{(\emptyset)}_{\mathcal{K}_{G}}=\{\ket{\psi_{\emptyset}(\xi)}:\xi\in\braket{\mathcal{K}_G}\}$\\
\midrule

\includegraphics[width=\threeVGSWidth]{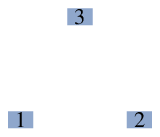}
& $\left\{  \{1\},\{2\},\{3\}\right\}  $ & $\left\{  \left\vert
000\right\rangle \right\}  $\\
\midrule

\includegraphics[width=\threeVGSWidth]{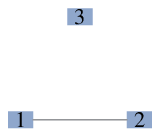}
&
$\left\{  \{3\} \right\}  $ &
$\left\{\ket{000},\ket{010},\ket{100},-\ket{110}\right\}$
\\
\midrule
\includegraphics[width=\threeVGSWidth]{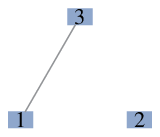}
&
$\left\{  \{2\} \right\}  $ &
$\left\{\ket{000},\ket{001},\ket{100},-\ket{101}\right\}$
\\
\midrule
\includegraphics[width=\threeVGSWidth]{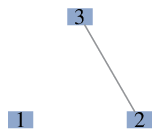}
&
$\left\{  \{1\} \right\}  $ &
$\left\{\ket{000},\ket{001},\ket{010},-\ket{011}\right\}$
\\
\midrule

\includegraphics[width=\threeVGSWidth]{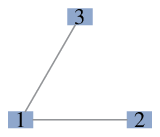}
&
$\left\{  \{2,3\} \right\}  $ &
$\left\{\ket{000},\ket{100},\ket{011},-\ket{111}\right\}$
\\
\midrule
\includegraphics[width=\threeVGSWidth]{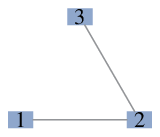}
&
$\left\{  \{1,3\} \right\}  $ &
$\left\{\ket{000},\ket{100},\ket{101},-\ket{111}\right\}$
\\
\midrule
\includegraphics[width=\threeVGSWidth]{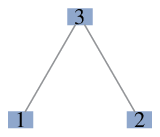}
&
$\left\{  \{1,2\} \right\}  $ &
$\left\{\ket{000},\ket{001},\ket{110},-\ket{111}\right\}$
\\
\midrule

\includegraphics[width=\threeVGSWidth]{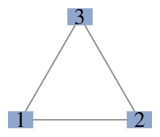}
&
$\left\{  \{1,2,3\} \right\}  $ &
$\left\{  \left\vert 100\right\rangle
,\left\vert 010\right\rangle ,\left\vert 001\right\rangle ,-\left\vert
111\right\rangle \right\}  $
\\

\bottomrule
\end{tabular}
\caption{X-chain states of 3-vertex graph states}
\label{table::X-chain_states_of_3-vertex_graph_states}

\end{table}

\section{Application of the X-chain factorization}
\label{sec::application_of_X-chains}
The representation of graph states in the X-chain factorization reveals certain substructures of graph states.
In this section, we discuss its usefulness for the calculation of graph state overlaps, the Schmidt decomposition and unilateral projections in bipartite systems.

\subsection{Graph state overlaps}
\label{sec::graph_state_overlap}
In \cite{WuRKSKMBruss2014-05}, the overlaps of graph states are the basis for genuine multipartite entanglement detection of randomized graph states with projector-based witnesses $W_{G}=\id/2-\projector{G}$, see \cite{AcinBLSanpera2001-07,GuhneHBELMSanpera2002-12}, where $G$ is a connected graph.
An expectation value $\mathrm{tr}(|H\rangle\langle H|G\rangle\langle G|)>1/2$ indicates the presence of genuine multipartite entanglement of the graph state $\ket{H}$.
\par
In general, a graph state $|G\rangle=\prod_{e\in E_{G}}U_{Z}^{(e)}|0_{X}\rangle$ is created by control-Z operators $U_{Z}^{(e)}$, where
\begin{equation}
U_{Z}^{\{v_a,v_b\}}:=\projector{0}^{(a)}\otimes\id^{(b)}+\projector{1}^{(a)}\otimes\sigma_Z^{(b)}.
\end{equation}
Since the operators $U_{Z}^{(e)}$ commute for different edges $e$ and are unitary and Hermitian, the overlap $\braket{G|H}$ is calculated by%
\begin{equation}
\label{eq::graph_state_overlap_1}
\left\langle G|H\right\rangle =\langle0^{\otimes n}_{X}|\prod_{e\in E_{G}\Delta E_{H}%
}U_{Z}^{(e)}|0^{\otimes n}_{X}\rangle=\langle0^{\otimes n}_{X}|G\Delta H\rangle.
\end{equation}
According to Eq. \eqref{eq::graph_states_in_Z-basis},
\begin{equation}
\label{eq::graph_state_overlap_2}
\langle G|H\rangle=\frac{1}{2^{n/2}}\sum_{i=0}^{2^{n}-1} \left(
-1\right)  ^{\left\langle i_{Z},i_{Z}\right\rangle _{A_{G\Delta H}}},%
\end{equation}
where $G\Delta H$ is the symmetric difference of the graphs $G$ and $H$.
$G\Delta H$ is the graph $(V_{G\Delta H},E_{G\Delta H})$, whose vertices and edges are $V_{G\Delta H}=V_{G}=V_{H}$ and $E_{G\Delta H}=E_{G}\cup E_{H}\setminus E_{G}\cap E_{H}$, respectively.
However, the complexity of this calculation increases exponentially with the size of the system.
%
\par
The quantity obtained from Eq. \eqref{eq::graph_states_in_Z-basis},
\begin{equation}
\left\langle 0^{\otimes n}_{X}|G\right\rangle =\frac{1}{2^{n/2}}\sum_{i=0}^{2^{n}-1}\left(-1\right)  ^{\left\langle i_{Z},i_{Z}\right\rangle _{A}},%
\end{equation}
corresponds to the difference of the positive and negative amplitudes of $\ket{G}$ in the Z-basis.
We can define for each graph state $\ket{G}$ a Boolean function
$f_{G}:=\braket{i_Z,i_Z}_{A} \pmod 2$ with $A$ being the adjacency matrix.
The function $f_{G}$ is balanced, if and only if $\braket{0^{\otimes n}_X|G}=0$, otherwise it is biased.
We introduce the bias degree of a graph state and define its Z-balance as follows.

\begin{definition}
[Bias degree and Z-balanced graph states]%
\label{def::bias_degree_balanced_graph_states}
\Needspace*{8\baselineskip}
The \emph{(Z-)bias degree} $\beta$ of
a graph state $|G\rangle$ with $n$ vertices is defined as the overlap
\begin{equation}
\beta(|G\rangle):=\langle0_{X}^{\otimes n}|G\rangle,
\label{eq::def_bias_degree_of_graphs}%
\end{equation}
where $|0_{X}\rangle=\left(  |0_{Z}\rangle+|1_{Z}\rangle\right)  /\sqrt{2}$.
A graph state with zero bias degree is called \emph{Z-balanced}.
\hiddengls{biasDegree}
\end{definition}
\par
The bias degree is related to the weight of a graph
state, $\omega^{-}\left(  G\right)  :=\left\vert \left\{  i_{Z}:\langle
i_{Z}|G\rangle/\left\vert \langle i_{Z}|G\rangle\right\vert =-1\right\}
\right\vert $, which is equal to the number of minus amplitudes in $|G\rangle$
in the Z-basis \cite{CosentinoSimone2009-Weight}.
The probability of finding a negative amplitude in the Z-basis is $1/2-\beta(|G\rangle)/2$, which is equal to
$\omega^{-}\left(  G\right)  /2^{n}$.
Note that as a result of Eq. \eqref{eq::graph_states_in_Z-basis_with_pi}, the bias degree of a graph state is equal to the sum of its stabilizer parities.
\begin{equation}
  \beta(\ket{G}) = \sum_{\xi\subseteq V_G}\pi_G(\xi).
\end{equation}
\bigskip
\par
As a result of Theorem \ref{theorem::graph_states_as_X-factorized_states}, the bias degree $\Braket{0_x|G}$, depends only on the number of X-chain generators and the parity of their corresponding X-resources.
\begin{corollary}
[Graph state overlaps and bias degrees]\label{coro::balanced_and_orthogonal_graph_states}
\Needspace*{10\baselineskip}
The overlap of two graph states $\ket{G}$ and $\ket{H}$ is equal to the bias degree of the graph state $\ket{G \Delta H}$, i.e.
\begin{equation}\label{eq::graph_state_overlap_coro}
\braket{G|H} = \beta(\ket{G\Delta H}).
\end{equation}
The bias degree of a graph state $|G\rangle$ is equal to%
\begin{equation}
|\beta(|G\rangle)|=\frac{1}{2^{\left(  n-\left\vert \Gamma_{G}\right\vert
\right)  /2}}\prod_{\gamma\in\Gamma_{G}}\delta_{\pi_{G}(\gamma)}^{1},
\end{equation}
where $\Gamma_{G}$ is the X-chain generating set of $|G\rangle$, $\delta$ is the Kronecker-delta and $\pi
_{G}(\gamma)$ is the stabilizer-parity of X-chain generators $\gamma$.

\begin{proof}
First we prove that there does not exist $\xi$ such that $c_{\xi}=x_{\Gamma}$.
Assume $c_{\xi}=x_{\Gamma}$, then $|c_{\xi}\cap\gamma|\overset
{\operatorname{mod}2}{=}|\xi\cap c_{\gamma}|=0$.
However, according to the definition of $x_{\Gamma}$ (Def. \ref{prop::X-chain_states_in_X-basis}),
$|c_{\xi}\cap\gamma|=|x_{\Gamma}\cap\gamma|=1$ which contradicts $|c_{\xi}%
\cap\gamma|=0\operatorname{mod}2$.
Then the only possible zero X-chain state is $\ket{i^{x_{(\Gamma)}}}$.
Therefore Theorem \ref{theorem::graph_states_as_X-factorized_states} leads to%
\begin{equation}
|\beta(|G\rangle)|=\frac{1}{2^{\left(  n-\left\vert \Gamma_{G}\right\vert
\right)  /2}}\langle 0_{X}|i^{(x_{\Gamma})}\rangle .
\end{equation}
According to the definition of the X-chain basis, $x_{\Gamma}=\emptyset$ if and
only if $\pi_{G}\left(  \gamma\right)  =1$ for all X-chain generators
$\gamma\in\Gamma_{G}$, that means $\left\langle 0_{X}|i^{(x_{\Gamma}%
)}\right\rangle =\prod_{\gamma\in\Gamma_{G}}\delta_{\pi_{G}(\gamma)}^{1}$.
\end{proof}
\end{corollary}
\noindent
In \cite{CosentinoSimone2009-Weight}, the authors relate the weight
$\omega^{-}\left(  G\right)  $ to the binary rank of the adjacency matrix of graphs.
Our Corollary \ref{coro::balanced_and_orthogonal_graph_states} is a similar result showing that the bias degree depends on the binary rank of the adjacency matrix, which is equal to $n-\left\vert \Gamma_{G}\right\vert $.
%
%
\bigskip
\par
Here, we focus on the bias degree and Z-balance of graph states.
Since the X-chain group of a graph state can be efficiently determined, instead of Eq. \eqref{eq::graph_state_overlap_2}, Corollary \ref{coro::balanced_and_orthogonal_graph_states} provides an efficient method to calculate the graph state overlap.
As a result of Corollary \ref{coro::balanced_and_orthogonal_graph_states}, we arrive at the following corollary.
\begin{corollary}[Z-balanced graph states]
  A graph state is Z-balanced, if and only if it has at least one X-chain generator $\gamma^{-}$ with negative stabilizer-parity, i.e. $\left\vert E\left(  G\left[  \gamma^{-}\right]\right)  \right\vert $ is odd.
  Two graph states are orthogonal, if and only if $|G\Delta H\rangle$ is Z-balanced.
\end{corollary}
%
%
%
%
\par
Knowing all the Z-balanced graph states with vertex number $n$ allows one to identify
all pairs of orthogonal graph states with $n$ vertices.
Note that relabeling a graph state (graph isomorphism) does not change its bias degree,
since the structure of the X-chain group does not change under graph isomorphism.
%
\par
In Fig. \ref{fig::balanced_graph_states_up_to_5v}, the Z-balanced graph states
up to five vertices are listed.
Every graph in the figure represents an isomorphic class.
From these balanced graph states one can obtain orthogonal
graph states via the graph symmetric difference. Examples of orthogonal graph
states derived from the Z-balanced graph states $|C_{3}\rangle$ and
$|C_{5}\rangle$ are shown in Figs. \ref{fig::orthogonal_graph_states_C3} and
\ref{fig::orthogonal_graph_states_C5} , respectively, ($C_{3}$ and $C_{5}$ are
the first and fifth graph in Fig. \ref{fig::balanced_graph_states_up_to_5v} ).
\begin{figure*}[ht!]

  \centering
  \begin{tabular}{|C{0.8\textwidth}|}\centering
  \includegraphics[width=0.95\linewidth]{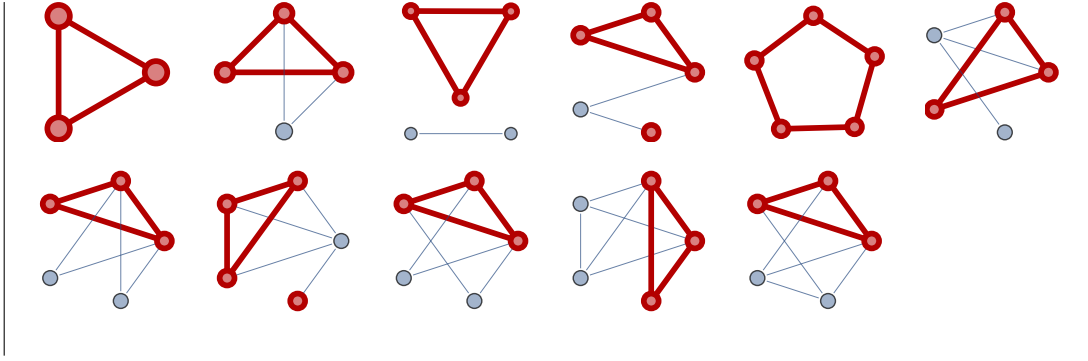}\\
  \end{tabular}

  \caption{\colorfig Z-balanced graph states (see Def. \ref{def::bias_degree_balanced_graph_states}) up to $5$ vertices: Each graph represents a graph isomorphic class. Each balanced graph state has at least one X-chain $\gamma^-$ with negative parity. In each graph, the $\gamma^-$-induced subgraph $G[\gamma^-]$ is highlighted in red with bold edges. Every highlighted $\gamma^-$-induced subgraph has an odd edge number.}
  \label{fig::balanced_graph_states_up_to_5v}

\end{figure*}%
\begin{figure*}[ht]
  \centering
\def\gwidth{0.1\textwidth}
\def\cellwidth{0.12\textwidth}
\begin{tabular}{
|>{\centering\arraybackslash}m{\cellwidth}>{\centering\arraybackslash}m{\cellwidth}
|>{\centering\arraybackslash}m{\cellwidth}>{\centering\arraybackslash}m{\cellwidth}
|>{\centering\arraybackslash}m{\cellwidth}>{\centering\arraybackslash}m{\cellwidth}
|}
\hline \noalign{\smallskip}
%
%
\includegraphics[width=\gwidth]{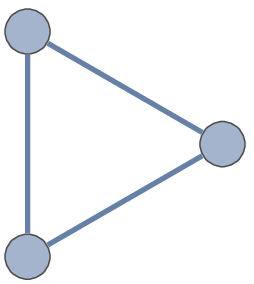} & \includegraphics[width=\gwidth]{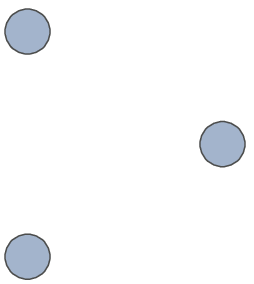} &

\includegraphics[width=\gwidth]{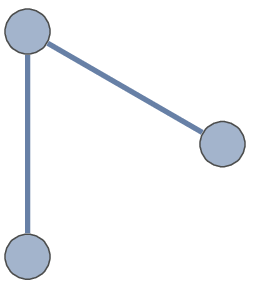} & \includegraphics[width=\gwidth]{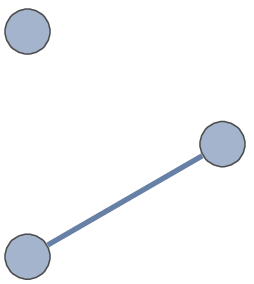}&

\includegraphics[width=\gwidth]{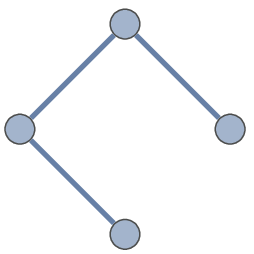} & \includegraphics[width=\gwidth]{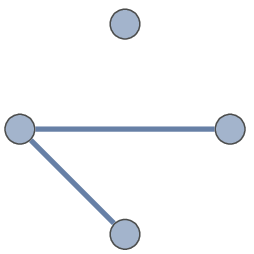}%
\\ \noalign{\smallskip} \hline
\end{tabular}

  \caption{Orthogonal graph states derived from the Z-balanced graph state $|C_3\rangle$: The graph states in each cell are orthogonal to each other. Their symmetric difference is identical to the cycle graph $C_3$, where $C_3$ is the first graph in Fig. \ref{fig::balanced_graph_states_up_to_5v}.}
  \label{fig::orthogonal_graph_states_C3}
\end{figure*}%
\begin{figure*}[ht]
  \centering
\def\gwidth{0.12\textwidth}
\def\cellwidth{0.14\textwidth}
\begin{tabular}{
|>{\centering\arraybackslash}m{\cellwidth}>{\centering\arraybackslash}m{\cellwidth}
|>{\centering\arraybackslash}m{\cellwidth}>{\centering\arraybackslash}m{\cellwidth}
|>{\centering\arraybackslash}m{\cellwidth}>{\centering\arraybackslash}m{\cellwidth}
|}
\hline\noalign{\smallskip}
%
%
%
\includegraphics[width=\gwidth]{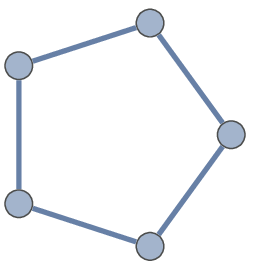} & \includegraphics[width=\gwidth]{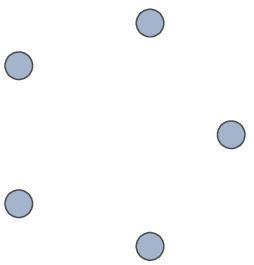} &

\includegraphics[width=\gwidth]{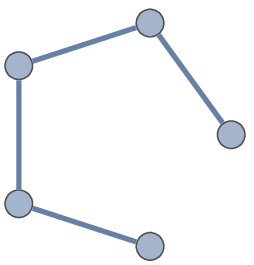} & \includegraphics[width=\gwidth]{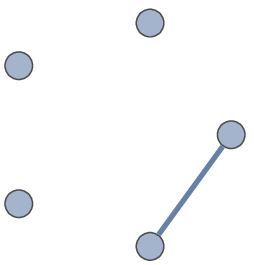} &

\includegraphics[width=\gwidth]{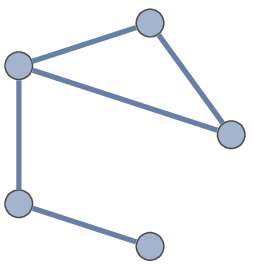} & \includegraphics[width=\gwidth]{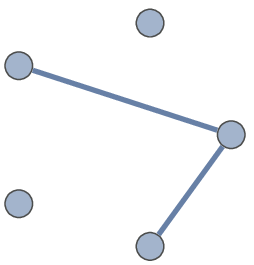} %
\\
\noalign{\smallskip}
\hline
\noalign{\smallskip}

\includegraphics[width=\gwidth]{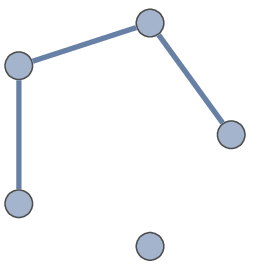} & \includegraphics[width=\gwidth]{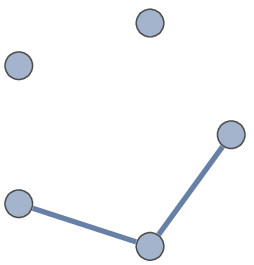} &

\includegraphics[width=\gwidth]{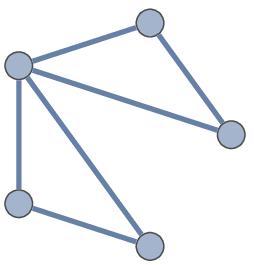} & \includegraphics[width=\gwidth]{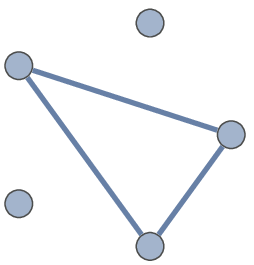} &

\includegraphics[width=\gwidth]{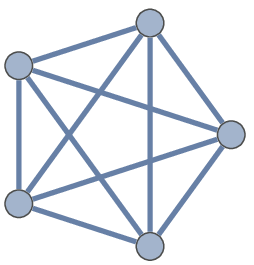} & \includegraphics[width=\gwidth]{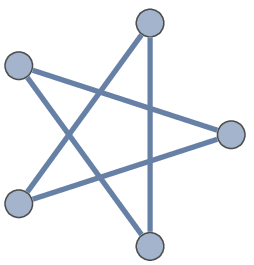}%
\\ \noalign{\smallskip} \hline
\end{tabular}

  \caption{Orthogonal graph states derived from the Z-balanced graph state $|C_5\rangle$: The graph states in each cell are orthogonal. Their symmetric difference is identical to the cycle graph $C_5$, where $C_5$ is the fifth graph in Fig. \ref{fig::balanced_graph_states_up_to_5v}.}
  \label{fig::orthogonal_graph_states_C5}
\end{figure*}%

\subsection{Schmidt decomposition}
\label{sec::schmidt_decomposition}
In this section, we discuss the Schmidt decomposition of graph states represented in
the X-basis, which is derived via the X-chain factorization.
The Schmidt decomposition of a graph state for an $A|B$-bipartition reads%
\begin{equation}
|G\rangle=\frac{1}{2^{r_{S}/2}}\sum_{i=1}^{r_S}\ket{\phi_{i}^{(A)}}\ket{\psi_{i}^{(B)}},
\label{eq::Schmidt_decomposition_of_graph_states_general}%
\end{equation}
where $\braket{\phi_i^{(A)}|\phi_j^{(A)}}=\delta_{ij}$ and $\braket{\psi_i^{(B)}|\psi_j^{(B)}}=\delta_{ij}$.
\hiddengls{SchmidtRank}
Here $r_{S}$ is the Schmidt rank of the graph state $|G\rangle$ with respect to the partition $A$ versus $B$.
Its value%
\begin{equation}
r_{S}=\left\vert S_{A}\right\vert :=\left\vert \left\{  s_{G}^{(\xi)}\in
S_{G}:\text{supp}(s_{G}^{(\xi)})\subseteq A\right\}  \right\vert
\end{equation}
is studied in the section III.B of Ref. \cite{HeinEisertBriegel2004-06} via the
Schmidt decomposition of graph states in the Z-basis, where $\mathrm{supp}%
(s_{G}^{(\xi)})$ is the support of the stabilizer $s_{G}^{(\xi)}$. The
$\mathrm{supp}(s_{G}^{(\xi)})$ is equal to the projection on the Hilbert space spanned by qubits corresponding to the vertices $\xi\cup c_{\xi}$, which is the set of vertices on which the stabilizer $s_{G}^{(\xi)}$ acts non-trivially (i.e. not equal to the identity).
\bigskip
\par
We derive the Schmidt decomposition of graph states in the X-basis in the following steps. First, we generalize the X-chain factorization of graph states (Theorem \ref{theorem::graph_states_as_X-factorized_states}) to
the X-chain factorization of arbitrary correlation states (Theorem
\ref{theorem::X-factorization_of_correlation_states}). Second, we introduce
three correlation subgroups, whose correlation states are $A|B$-biseparable
(Lemma \ref{lemma::biseparable_AB-correlation_states}). Third, we prove the
orthonormality of these correlation states (Lemma
\ref{lemma::orthonormality_of_AB-correlation_states}). At the end, we arrive
at the Schmidt decomposition in Theorem
\ref{theorem::Schmidt_decomposition_in_AB-correlation_states}.

The X-chain factorization of graph states in Theorem
\ref{theorem::graph_states_as_X-factorized_states} can be generalized to
correlation states (introduced in Eqs. \eqref{eq::def_correlation_states} and \eqref{eq::correlation_state_explicit_form}) as follows.
\begin{theorem}
[X-chain factorization of $\mathcal{K}$-correlation states]%
\label{theorem::X-factorization_of_correlation_states}
\Needspace*{10\baselineskip}
Let $\left\langle\mathcal{K}_{1}\right\rangle ,\left\langle \mathcal{K}_{2}\right\rangle
\subseteq\left\langle \mathcal{K}_{G}\right\rangle $
be two disjoint correlation subgroups of a graph state $|G\rangle$, and $\mathcal{K}=\mathcal{K}_1\cup\mathcal{K}_2$.
Then the $\mathcal{K}$-correlation state is a superposition of
$\mathcal{K}_{1}$-correlation states,%
\begin{equation}
|\psi_{\mathcal{K}}\left(  \xi\right)  \rangle
=\frac{1}{2^{\left\vert \mathcal{K}_{2}\right\vert /2}}\sum_{\xi^{\prime}%
\in\left\langle \mathcal{K}_{2}\right\rangle }|\psi_{\mathcal{K}_{1}}\left(
\xi\Delta\xi^{\prime}\right)  \rangle
\end{equation}
with $\xi\in\braket{\mathcal{K}_{G}}/\braket{\mathcal{K}}  $ being an element in their quotient group.
Theorem \ref{theorem::graph_states_as_X-factorized_states} is a special case of this
theorem related by $\braket{\mathcal{K}}=\left\langle \mathcal{K}_{1}\right\rangle \times\left\langle
\mathcal{K}_{2}\right\rangle =\mathcal{\emptyset}\times\left\langle
\mathcal{K}_{G}\right\rangle $.

\begin{proof}
According to the definition in Eq. (\ref{eq::def_correlation_states}) it
holds
\begin{equation}
|\psi_{\mathcal{K}_{1}\cup\mathcal{K}_{2}}\left(  \xi\right)  \rangle
=s_{G}^{(\xi)}\prod_{\kappa\in\mathcal{K}_{2}}\frac{1+s_{G}^{(\kappa)}}%
{\sqrt{2}}\prod_{\kappa\in\mathcal{K}_{1}}\frac{1+s_{G}^{(\kappa)}}{\sqrt{2}%
}\left\vert i^{(x_{\Gamma})}\right\rangle .
\end{equation}
Due to the commutativity of the graph state stabilizers it follows%
\begin{equation}
|\psi_{\mathcal{K}}\left(  \xi\right)  \rangle=|\psi_{\mathcal{K}_{1}%
\cup\mathcal{K}_{2}}\left(  \xi\right)  \rangle=\prod_{\kappa\in
\mathcal{K}_{2}}\frac{1+s_{G}^{(\kappa)}}{\sqrt{2}}|\psi_{\mathcal{K}_{1}%
}\left(  \xi\right)  \rangle.
\end{equation}
According to Proposition \ref{prop::isomorphism_of_vertex-induction_operation}%
, $s_{G}^{(\kappa_{1}\Delta\kappa_{2})}=s_{G}^{(\kappa_{1})}s_{G}^{(\kappa
_{2})}$, the product of $(1+s_{G}^{(\kappa)})$ with $\kappa\in\mathcal{K}_{2}$
becomes the sum of the stabilizers $s_{G}^{(\xi^{\prime})}$ with $\xi^{\prime
}\in\left\langle \mathcal{K}_{2}\right\rangle $.
\begin{align}
|\psi_{\mathcal{K}}\left(  \xi\right)  \rangle &  =\frac{1}{2^{\left\vert
\mathcal{K}_{2}\right\vert /2}}\sum_{\xi^{\prime}\in\left\langle
\mathcal{K}_{2}\right\rangle }s_{G}^{(\xi^{\prime})}|\psi_{\mathcal{K}_{1}%
}\left(  \xi\right)  \rangle\nonumber\\
&  =\frac{1}{2^{\left\vert \mathcal{K}_{2}\right\vert /2}}\sum_{\xi^{\prime
}\in\left\langle \mathcal{K}_{2}\right\rangle }|\psi_{\mathcal{K}_{1}}\left(
\xi\Delta\xi^{\prime}\right)  \rangle,
\end{align}
where the second equality is a result of property
\ref{item::correlation_states_property_2} in Corollary
\ref{coro::properties_of_correlation_states}.
\end{proof}%
\end{theorem}
\begin{algorithm}[Factorization diagram of correlation states]
\label{algo::factorization_diagram_of_correlation_states}
\Needspace*{8\baselineskip}
Theorem \ref{theorem::X-factorization_of_correlation_states} can be interpreted by the
factorization diagram in Fig. \ref{fig::factorization_diagram_correlation_states}.
\begin{enumerate}
  \item One decomposes the group $\mathcal{P}(V_G)$ into the direct product of the X-chain group $\braket{\Gamma_G}$ and the correlation group $\braket{\mathcal{K}_G}$.
  \item From the X-chain group $\braket{\Gamma_G}$, one obtains the set of X-chain states $\Psi^{\emptyset}_{\mathcal{K}_G}$.
  \item From the correlation group $\braket{\mathcal{K}_1}$, one obtains graph states via the superposition of the X-chain states in $\Psi^{\emptyset}_{\mathcal{K}_G}$ within $\braket{\mathcal{K}_1}$.
  \item At the end the correlation state $\ket{\psi_{\mathcal{K}_1\cup \mathcal{K}_2}(\xi)}$ is the superposition of the $\mathcal{K}_1$-correlation states $\ket{\psi_{\mathcal{K}_1}(\xi\Delta\xi')}\in\Psi^{(\mathcal{K}_1)}_{\mathcal{K}_G}$ inside the correlation group $\xi'\in\braket{\mathcal{K}_2}$ (Theorem \ref{theorem::X-factorization_of_correlation_states}).
\end{enumerate}
\end{algorithm}%
\begin{figure}[ht!]
  \centering
  \corrstatediagram[heightScale=0.7]
  \caption{\colorfig The X-chain factorization diagram of correlation states: A graphical summary of Theorem \ref{theorem::X-factorization_of_correlation_states}. The $\xi$ in $|\psi_{\mathcal{K}_{1}\cup\mathcal{K}_{2}}\left(  \xi\right)  \rangle$ are elements in the quotient group, $\xi\in \braket{\mathcal{K}_G}/\braket{\mathcal{K}_1\cup \mathcal{K}_2}$. }
  \label{fig::factorization_diagram_correlation_states}
\end{figure}
\noindent
The subspace of X-chain states $\mathrm{span}(\Psi_{K_{G}}^{(\emptyset)})$ are
projected via $\braket{\mathcal{K}_1}$-stabilizers to the space spanned by the
$\mathcal{K}_{1}$-correlation states $|\psi_{\mathcal{K}_{1}}\left(
\xi\right)  \rangle$. Further, the subspace $\mathrm{span}(\Psi_{K_{G}%
}^{(\mathcal{K}_{1})})$ are then projected via $\braket{\mathcal{K}_2}$%
-stabilizers to the $\mathcal{K}_{1}\cup\mathcal{K}_{2}$-correlation states
$|\psi_{\mathcal{K}_{1}\cup\mathcal{K}_{2}}\left(  \xi\right)  \rangle$.
With this theorem, one can obtain the Schmidt decomposition of graph states,
by appropriate selection of the correlation subgroup $\mathcal{K}_{1}$,
such that its corresponding $\mathcal{K}_{1}$-correlation states are $A|B$-separable and mutually orthonormal.
\bigskip
\par
Let $|G\rangle$ be a graph state with the correlation group $\left\langle
\mathcal{K}_{G}\right\rangle $ and $A|B$ be a bipartition of its vertices. In
order to find the Schmidt decomposition, we select $\braket{K_1}$ as the
union of three disjoint correlation subgroups specified as follows.
\begin{enumerate}
\item The correlation subgroup, whose elements possess a correlation index only
in $B$:%
\begin{equation}
\langle\mathcal{K}^{(B)}\rangle:=\left\{  \xi\in\left\langle \mathcal{K}%
_{G}\right\rangle :c_{\xi}\subseteq B\right\}
.\label{eq::def_KappaB_in_AB_bipartition}%
\end{equation}
\hiddengls{corrGroupToB}

\item The correlation subgroup, whose elements possess a correlation index only
in $A$ and only consists of vertices in $A$:
\begin{equation}
\langle\mathcal{K}_{A}^{(A)}\rangle:=\left\{  \xi\in\left\langle
\mathcal{K}_{G}\right\rangle :c_{\xi}\subseteq A,\xi\subseteq A\right\}
.\label{eq::def_KappaAA_in_AB_bipartition}%
\end{equation}
\hiddengls{corrGroupAToA}

\item The correlation subgroup, whose elements possess a correlation index only
in $A$, consists of vertices in $B$ and has an even number of edges between all $\beta
\in\left\langle \mathcal{K}^{(B)}\right\rangle $:%
\begin{align}
\langle\mathcal{K}_{\sim B}^{(A)}\rangle &  :=\left\{  \xi
\in\left\langle \mathcal{K}_{G}\right\rangle :c_{\xi}\subseteq A,\xi
\not \subseteq A\right\}  \cap\{\xi\in\left\langle \mathcal{K}_{G}%
\right\rangle :\nonumber\\
&  \left\vert E_{G}\left(  \xi:\beta\right)  \right\vert \overset
{\operatorname{mod}2}{=}0\text{, for all }\beta\in\left\langle \mathcal{K}%
^{(B)}\right\rangle \}.\label{eq::def_KappaA_B_in_AB_bipartition}%
\end{align}
\hiddengls{corrGroupSimBToA}
\end{enumerate}
These three groups form a special group
\begin{equation}
\langle\mathcal{K}^{A\rfloor B}\rangle:=\langle\mathcal{K}_{A}^{(A)}%
\cup\mathcal{K}_{\sim B}^{(A)}\rangle\times\langle\mathcal{K}_{{}}%
^{(B)}\rangle
\end{equation}
called \emph{$A\rfloor B$-correlation group}.
\hiddengls{corrGroupASepB}\hiddengls{corrASepBState}
(The notation ``$A\rfloor B$'' is used, as the group is not symmetric with respect to exchanging $A$ and $B$.)
We will show in Lemma \ref{lemma::biseparable_AB-correlation_states} that all $A\rfloor B$-correlation states $|\psi_{\mathcal{K}^{A\rfloor B}}(\xi)\rangle$ with $\xi\in\left\langle \mathcal{K}_{G}\right\rangle /\langle\mathcal{K}^{A\rfloor B}\rangle$, are $A|B$-separable.
The corresponding quotient group is denoted as%
\begin{equation}
\langle\mathcal{K}^{A\rightharpoonup B}\rangle:=\left\langle \mathcal{K}%
_{G}\right\rangle /\langle\mathcal{K}^{A\rfloor B}\rangle
\end{equation}
and called \emph{$\left(  A\rightharpoonup B\right)$-correlation group}.
\hiddengls{corrGroupAandB}
(The notation $A\rightharpoonup B$ is introduced, as there is again no symmetry under exchange of $A$ and $B$, as the correlation index $c_{\xi}$ of $\xi\in\braket{\mathcal{K}^{A\rightharpoonup B}}$ is always inside $A$.)
We will show in Theorem \ref{theorem::Schmidt_decomposition_in_AB-correlation_states} that the Schmidt rank of $\ket{G}$ is equal to the cardinality $|\braket{\mathcal{K}^{A\rightharpoonup B}}|$.
That means that the correlation subgroup $\mathcal{K}^{A\rightharpoonup B}$ generates the $A|B$ correlation in the graph state $\ket{G}$.
Note that we investigated many graphs and found their correlation subgroups $\langle\mathcal{K}_{\sim B}^{(A)}\rangle$ all to be empty.
That means the group $\langle\mathcal{K}_{\sim B}^{(A)}\rangle$ may not exist for any graph state.
However, this is still an open question.
\def\GammaG{\{\{1,2,3\}\}}
\def\KappaAA{\{\{2,3\}\}}
\def\KappaB{\{\{4,5\},\{2,3,4\}\}}
\def\KappaAsimB{\emptyset}
\def\KappaAfromB{\{\{2\}\}}
\par
\begin{figure*}[ht!]
  \centering
  \subfloat[]{
    \includegraphics[width=0.2\textwidth]{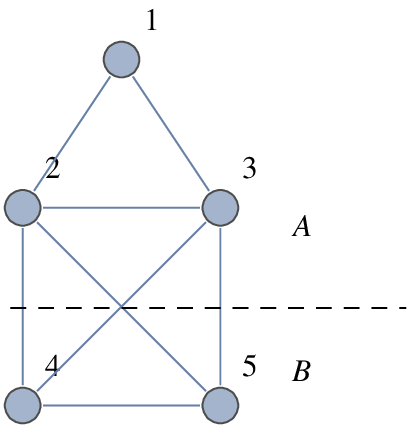}
    \label{fig::house_graph}
  }
  \subfloat[]{
    \includegraphics[width=0.8\textwidth]{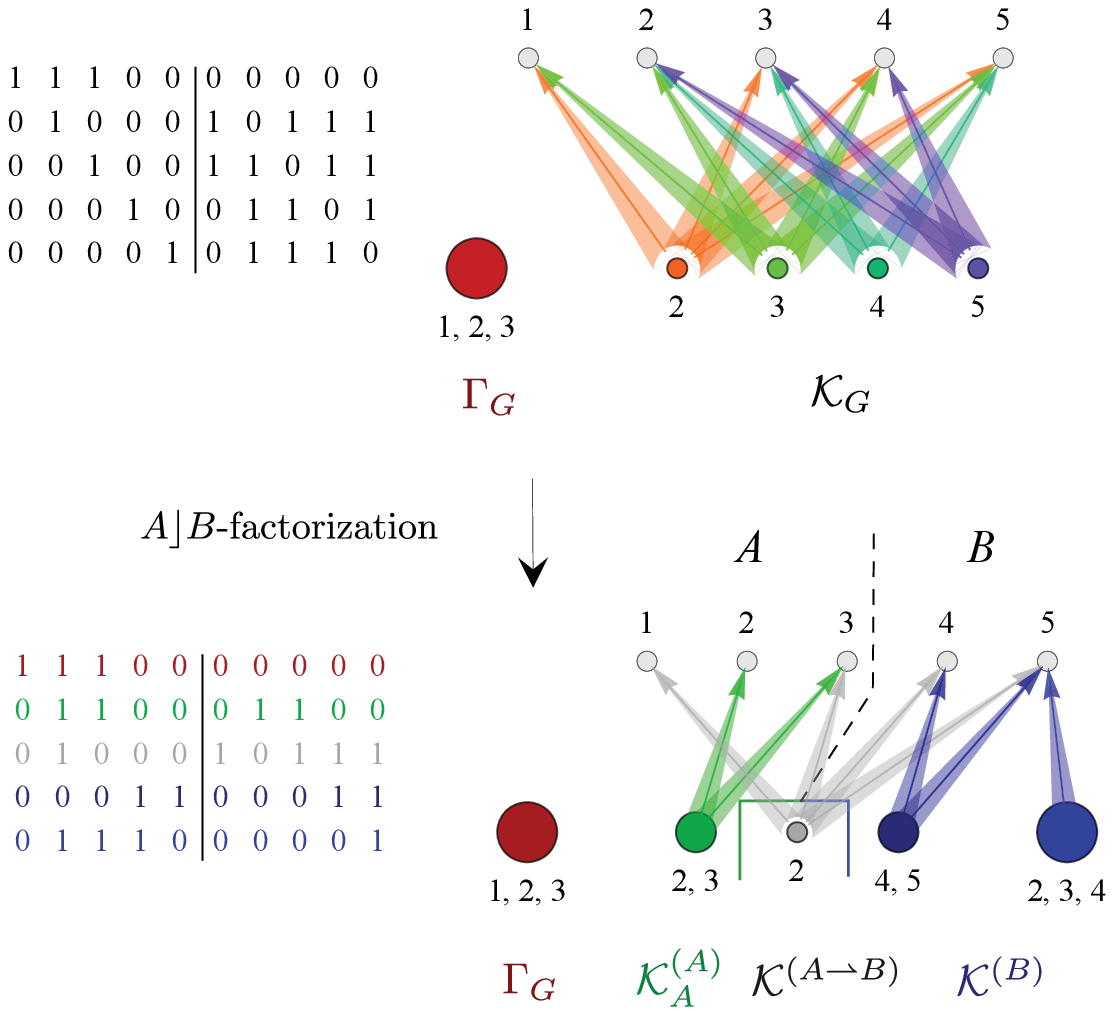}
    \label{fig::house_graph_binary_representation}
  }
  \\
  \subfloat[]{
  \ABfactorizationdiagram[
    GammaG={\{\{1,2,3\}\}},
    KappaAA={\{\{2,3\}\}},
    KappaB={\{\{4,5\},\{2,3,4\}\}},
    KappaAfromB={\{\{2\}\}},
    xG=10000,
    KappaAfromBNum=1,
    widthScale=1.1
    ]
  \label{fig::house_graph_factorization_diagram}
  }
  \\
  \caption{\colorfig $A\rfloor B$-factorization of graph states: (a) The graph state $\ket{G_{\mathrm{House}}}$ corresponding to a ``St. Nicholas's house'' is divided in two subsystems $A=\{1,2,3\}$ and $B=\{4,5\}$.
  (b) The binary representation of the X-chain factorization (the upper row) and $A\rfloor B$-factorization (the lower row). (c) The $A\rfloor B$-factorization diagram (see Algorithm \ref{algo::factorization_diagram_AB_correlation}) of the ``St. Nicholas's house'' graph state $\ket{G_{\mathrm{House}}}$.}
  \label{fig::house_graph_A-B_factorization}
\end{figure*}
\par
In this $A\rfloor B$-factorization, the correlation group $\mathcal{K}_{G}$ is
divided into four subgroups.
Let us take the graph of ``St. Nicholas's house'' in Fig. \ref{fig::house_graph} as an example.
This ``house'' state $|G_{\mathrm{House}}\rangle$
is divided into the bipartition $A=\{1,2,3\}$ versus $B=\{4,5\}$.
The correlation group factorization is shown in Fig.
\ref{fig::house_graph_binary_representation}. The X-chain group of
$\ket{G_{\mathrm{House}}}$ is $\GammaG$. The X-resources are factorized by the
X-chain group, $\mathcal{P}(V_{G})=\braket{\Gamma_G}\times
\braket{\mathcal{K}_G}$,
see the upper row in Fig. \ref{fig::house_graph_binary_representation}.
The array is the binary representation of the stabilizers induced by the X-chain generators
$\Gamma=\{\{1,2,3\}\}$ and correlation group generators $\mathcal{K}_G=\{\{2\},\{3\},\{4\},\{5\}\}$,
it corresponds to the incidence structure on its right hand side.
In the second row of Fig. \ref{fig::house_graph_binary_representation},
the X-resources, whose correlation indices lie in the
system $B$, are first grouped together into
$\braket{\mathcal{K}^{(B)}}=\braket{\KappaB}$. Second, the
X-resources $\xi$, whose correlation indices $c_{\xi}$ and $\xi$ itself are both contained by $V_{A}$,
are grouped into $\braket{\mathcal{K}_{A}^{(A)}}=\braket{\KappaAA}$. Third,
the group $\mathcal{K}_{\sim B}^{(A)}$ is empty. At the end, the $\left(
A\rightharpoonup B\right)  $-correlation group is then
$\braket{\mathcal{K}^{A\rightharpoonup B}}=\braket{\KappaAfromB}$.
\par
These three special correlation subgroups, $\braket{\mathcal{K}_{A}^{(A)}},$
$\braket{\mathcal{K}_{\sim B}^{(A)}}$ and $\braket{\mathcal{K}^{(B)}}$,
project the space spanned by the X-chain states into a subspace spanned by their
correlation states $|\psi_{A\rfloor B}(\xi)\rangle$.
These states are $A|B$-separable states, which is stated in the following lemma.
\begin{lemma}
[$A|B$-Separability of $A\rfloor B$-correlation states]%
\label{lemma::biseparable_AB-correlation_states}
\Needspace*{8\baselineskip}
For $\xi\in\langle
\mathcal{K}_{G}^{A\rightharpoonup B}\rangle$, the $\left(  A\rightharpoonup
B\right)  $-correlation states%
\begin{equation}
|\psi_{A\rfloor B}(\xi)\rangle=\pi_{G}\left(  \xi\right)  |\phi_{A\rfloor
B}^{(A)}(\xi)\rangle|\phi_{A\rfloor B}^{(B)}(\xi)\rangle
\label{eq::def_AB-correlation_states}%
\end{equation}
are $A|B$-separable with $|\phi_{A\rfloor B}^{(A)}(\xi)\rangle:=|\psi
_{\mathcal{K}_{A}^{(A)}\cup\mathcal{K}_{\sim B}^{(A)}}^{(A)}(\xi)\rangle$
and $|\phi_{A\rfloor B}^{(B)}(\xi)\rangle:=|\psi_{\mathcal{K}^{(B)}}^{(B)}(\xi)\rangle$ being the $(  \mathcal{K}_{A}^{(A)}\cup\mathcal{K}%
_{\sim B}^{(A)})$- and $\mathcal{K}^{(B)}$-correlation states
projected into the subspaces of $A$ and $B$, respectively.
\hiddengls{corrAsepBStateOnA}\hiddengls{corrAsepBStateOnB}
\begin{proof}
See Appendix \ref{sec::graph_states_proofs}.
\end{proof}
\end{lemma}
\noindent
Note that $|\psi_{A\rfloor B}(\xi)\rangle$ will be shown to be the Schmidt basis in Theorem \ref{theorem::Schmidt_decomposition_in_AB-correlation_states}.
There, one will also see that the global phase $\pi_G(\xi)$ ensures positive Schmidt coefficients.
\bigskip
\par
Let us continue to consider the ``St. Nicholas's house'' state as an example.
According to Proposition \ref{prop::X-chain_states_in_X-basis}, the
fundamental X-chain state of $|G_{\text{House}}\rangle$ is $|i^{x_{\Gamma}}%
\rangle=|10000\rangle$.
Then from the $\mathcal{K}_{A}^{(A)}$-correlation states,
\def\phiAempty{\frac{|100\rangle -|111\rangle }{\sqrt{2}}}%
\def\phiAone{\frac{|001\rangle +|010\rangle }{\sqrt{2}}}%
\def\phiBempty{\frac{|00\rangle -|01\rangle -|10\rangle -|11\rangle}{2}}%
\def\phiBone{\frac{-|00\rangle-|01\rangle-|10\rangle+|11\rangle}{2}}%
\def\xGammaA{\ket{100}}%
\def\xGammaB{\ket{00}}%
\begin{align}
|\psi_{\mathcal{K}_{A}^{(A)}}\left(  \emptyset\right)  \rangle &
=|\phi_{A\rfloor B}^{(A)}(\emptyset)\rangle \otimes\xGammaB\text{ and}\\
|\psi_{\mathcal{K}_{A}^{(A)}}\left(  \{2\}\right)  \rangle &
=|\phi_{A\rfloor B}^{(A)}\left(  \{2\}\right)\rangle\otimes\xGammaB\text{,}%
\end{align}
one can read off%
\begin{align}
|\phi_{A\rfloor B}^{(A)}(\emptyset)\rangle &  =\phiAempty \text{ and }\\
|\phi_{A\rfloor B}^{(A)}\left(  \{2\}\right)  \rangle &  =\phiAone.
\end{align}
From the $\mathcal{K}^{(B)}$-correlation states,
\begin{align}
|\psi_{\mathcal{K}^{(B)}}\left(  \emptyset\right)  \rangle &  =\xGammaA\otimes
|\phi_{A\rfloor B}^{(B)}(\emptyset)\rangle \text{ and}\\
|\psi_{\mathcal{K}^{(B)}}\left(  \{2\}\right)  \rangle &  =\xGammaA\otimes
|\phi_{A\rfloor B}^{(B)}\left(  \{2\}\right)  \rangle ,%
\end{align}
one can read off%
\begin{align}
|\phi_{A\rfloor B}^{(B)}(\emptyset)\rangle &  =\phiBempty\text{ and }\\
|\phi_{A\rfloor B}^{(B)}\left(  \{2\}\right)  \rangle &  =\phiBone.
\end{align}
According to Lemma \ref{lemma::biseparable_AB-correlation_states}, $A\rfloor
B$-correlation states are
\begin{align}
& |\psi_{A\rfloor B}(\emptyset)\rangle\nonumber\\
& =\left( \phiAempty \right)  \left(  \phiBempty \right)
\label{eq::house_state_A_B_correlation_state_1}%
\end{align}
and since $\pi_{G}(\{2\})=1$,%
\begin{align}
&  |\psi_{A\rfloor B}(\{2\})\rangle\nonumber\\
&  =\left( \phiAone \right)  \left(  \phiBone \right).
\label{eq::house_state_A_B_correlation_state_2}%
\end{align}
\bigskip
\par
Orthonormality of the states within the subspaces still needs to be verified.
This holds for the explicit example $\ket{G_{\mathrm{House}}}$ in Eqs.
\eqref{eq::house_state_A_B_correlation_state_1} and
\eqref{eq::house_state_A_B_correlation_state_2}.
In the general case, the orthonormality is shown in the following lemma.
\begin{lemma}[Orthonormality of $(A\rightharpoonup B)$-correlation states]
\label{lemma::orthonormality_of_AB-correlation_states}
\Needspace*{8\baselineskip}
The components of $A\rfloor B$-correlation states on subspace $A$ and $B$, $|\phi_{A\rfloor B}^{(A)}%
(\xi)\rangle$ and $|\phi_{A\rfloor B}^{(B)}(\xi)\rangle$, are orthonormal with
respect to $\xi\in\langle\mathcal{K}^{A\rightharpoonup B}\rangle$ within the subspaces
$A$ and $B$, respectively, i.e.,
\begin{equation}
\langle\phi_{A\rfloor B}^{(A)}(\xi_{1})|\phi_{A\rfloor B}^{(A)}(\xi
_{2})\rangle=0
\end{equation}
and%
\[
\langle\phi_{A\rfloor B}^{(B)}(\xi_{1})|\phi_{A\rfloor B}^{(B)}(\xi
_{2})\rangle=0
\]
for all $\xi_{1},\xi_{2}\in\langle\mathcal{K}^{A\rightharpoonup B}\rangle$ and
$\xi_{1}\not =\xi_{2}$.
%
%
\begin{proof}
See Appendix \ref{sec::graph_states_proofs}.
\end{proof}
\end{lemma}
\bigskip
\par
We can now construct the Schmidt decomposition of graph states with $A\rfloor B$-correlation states as follows.
\begin{theorem}
[Schmidt decomposition in $A\rfloor B$-correlation states]%
\label{theorem::Schmidt_decomposition_in_AB-correlation_states}
\Needspace*{12\baselineskip}%
The Schmidt decomposition of a graph state $|G\rangle$ is the superposition of its
$A\rfloor B$-correlation states,
\begin{equation}
|G\rangle=\frac{1}{2^{\left\vert \mathcal{K}^{A\rightharpoonup B}\right\vert
/2}}\sum_{\xi\in\left\langle \mathcal{K}^{A\rightharpoonup B}\right\rangle }\pi
_{G}\left(  \xi\right)  |\phi_{A\rfloor B}^{(A)}(\xi)\rangle|\phi_{A\rfloor
B}^{(B)}(\xi)\rangle.\label{eq::theorem_schmidt_decomp_in_AB_states}%
\end{equation}
The Schmidt rank $r_S$ and geometric measure of the $A|B$-bipartite entanglement \cite{Shimony1995-BiGeoM, BarnumLinden2001-GeoMeasure} can be expressed by%
\begin{equation}
\log_2(r_{S})=\mathcal{E}_{g}^{A|B}=\left\vert \mathcal{K}^{A\rightharpoonup B}\right\vert
\label{eq::theorem_schmidt_rank_of_graph_states}%
\end{equation}
with $\mathcal{E}_{g}^{(A|B)}\left(  |G\rangle\right)  :=-2\log_{2}\left(  \min_{\psi
}\left\vert \left\langle \psi_{A}\psi_{B}|G\right\rangle \right\vert \right)
$.
\hiddengls{BiEntGeoMeas}
\begin{proof}
Employing Theorem \ref{theorem::graph_states_as_X-factorized_states} and
\ref{theorem::X-factorization_of_correlation_states} together with Lemma
\ref{lemma::biseparable_AB-correlation_states} one can prove that the graph
state $|G\rangle$ is equal to the superposition of all biseparable $A\rfloor
B$-correlation states $|\psi_{A\rfloor B}(\xi)\rangle=\pi_{G}(\xi
)|\phi_{A\rfloor B}^{(A)}(\xi)\rangle|\phi_{A\rfloor B}^{(B)}(\xi)\rangle$. As
a result of the orthonormality of $|\phi_{A\rfloor B}^{(A)}(\xi)\rangle$ and
$|\phi_{A\rfloor B}^{(B)}(\xi)\rangle$ (Lemma
\ref{lemma::orthonormality_of_AB-correlation_states}), Eq.
(\ref{eq::theorem_schmidt_decomp_in_AB_states}) is a Schmidt decomposition.
The bipartite geometric measure of entanglement is equal to the maximum
singular value $s_{\max}$ of the matrix $M_{ij}:=\left\{  \langle i_{A}%
j_{B}|G\rangle\right\}  _{i,j}$ with $i=0,...,2^{\left\vert V_{A}\right\vert
}-1$ and $i=0,...,2^{\left\vert V_{B}\right\vert }-1$ \cite{Shimony1995-BiGeoM}. For the bipartite case
the singular value decomposition is equivalent to the Schmidt decomposition. Since
the Schmidt coefficients are all $2^{-\left\vert \mathcal{K}^{A\rightharpoonup
B}\right\vert /2}$, it follows that the geometric measure of bipartite entanglement of a
graph state, $\mathcal{E}^{A|B}_{g}:=-2\log_{2}\left(  s_{\max}\right)  $, is equal to the log of the
Schmidt rank, i.e. $\log_2(r_{S})=\left\vert \mathcal{K}^{A\rightharpoonup B}\right\vert $.
As a result, the $A\rfloor B$-correlation states $\pi_{G}(\xi)|\phi_{A\rfloor B}^{(A)}%
(\xi)\rangle|\phi_{A\rfloor B}^{(B)}(\xi)\rangle$ are the $A|B$-separable states, which are closest to $\ket{G}$.
\end{proof}
\end{theorem}
\par
According to \cite{GraphStateReviews2006}, the Schmidt rank is given by
$\log_{2}\left\vert \left\{  \sigma\in S_{G}:\text{supp}\left(  \sigma\right)
\subseteq V_{A}\right\}  \right\vert $ with $\left\vert V_{A}\right\vert
\leq\left\vert V_{B}\right\vert $, which is $\left\vert V_{A}\right\vert
-\left\vert \mathcal{K}_{A}^{(A)}\right\vert -\left\vert \Gamma_{G}\cap
\mathcal{P}\left(  V_{A}\right)  \right\vert $ in the language of the X-chain factorization.
The Schmidt rank is also equal to the cardinality of the matching \footnote{Note that the matching between two parties is not unique, but its cardinality is fixed.} between $A$ and $B$ \cite{HajdusekMurao2013-Direct}.
The matching is the set of edges between $A$ and $B$, which do not mutually share any common vertex \cite{Diestel_GraphTheory}.
Hence the cardinality $\left\vert
\mathcal{K}^{A\rightharpoonup B}\right\vert $ should be equal to the matching.
However the proof of this equality is still an open question.
\bigskip
\par
The result of this section can be summarized in an $A\rfloor B$-factorization diagram.
%
\begin{algorithm}[Factorization diagram: Schmidt decomposition of graph states]
\label{algo::factorization_diagram_AB_correlation}
\Needspace*{8\baselineskip}
The Schmidt decomposition of graph states in Theorem \ref{theorem::Schmidt_decomposition_in_AB-correlation_states} can be summarized in the factorization diagram of Fig. \ref{fig::factorization_diagram_AB_correlation}.
\begin{enumerate}
  \item The group $\mathcal{P}(V_G)$ is decomposed into the direct product of $\braket{\Gamma_G}$, $\braket{\mathcal{K}^{A\rfloor B}}=\braket{\mathcal{K}_A^{(A)}\cup\mathcal{K}_{\sim B}^{(A)}}\times\braket{\mathcal{K}^{(B)}}$ and $\braket{\mathcal{K}^{A\rightharpoonup B}}$.
  \item Via the X-chain group $\braket{\Gamma_G}$, one obtains the set of X-chain states $\Psi^{\emptyset}$.
  \item The Schmidt basis states $\ket{\phi_{A\rfloor B}^{(A)}(\xi})$ are constructed from the superposition of states in $\Psi^{\emptyset}$ inside the correlation group $\braket{\mathcal{K}_A^{(A)}\cup\mathcal{K}_{\sim B}^{(A)}}$ (Lemma \ref{lemma::biseparable_AB-correlation_states}).
  \item Similar to the previous step, one obtains the states $\ket{\phi_{A\rfloor B}^{(B)}(\xi)}$ via the correlation group $\braket{\mathcal{K}^{(B)}}$ (Lemma \ref{lemma::biseparable_AB-correlation_states}).
  \item Together with the stabilizer-parities $\pi_G(\xi)$, the set of $A\rfloor B$-correlation states $\Psi^{(A\rfloor B)}$ (Lemma \ref{lemma::biseparable_AB-correlation_states}) is constructed.
  \item Via the $(A\rightharpoonup B)$-correlation group $\braket{\mathcal{K}^{A\rightharpoonup B}}$, one obtains the Schmidt decomposition from the superposition of states in $\mathrm{span}(\Psi^{(A\rfloor B)})$ (Lemma \ref{lemma::orthonormality_of_AB-correlation_states} and Theorem \ref{theorem::Schmidt_decomposition_in_AB-correlation_states}).
\end{enumerate}
\end{algorithm}
\begin{figure*}[ht!]
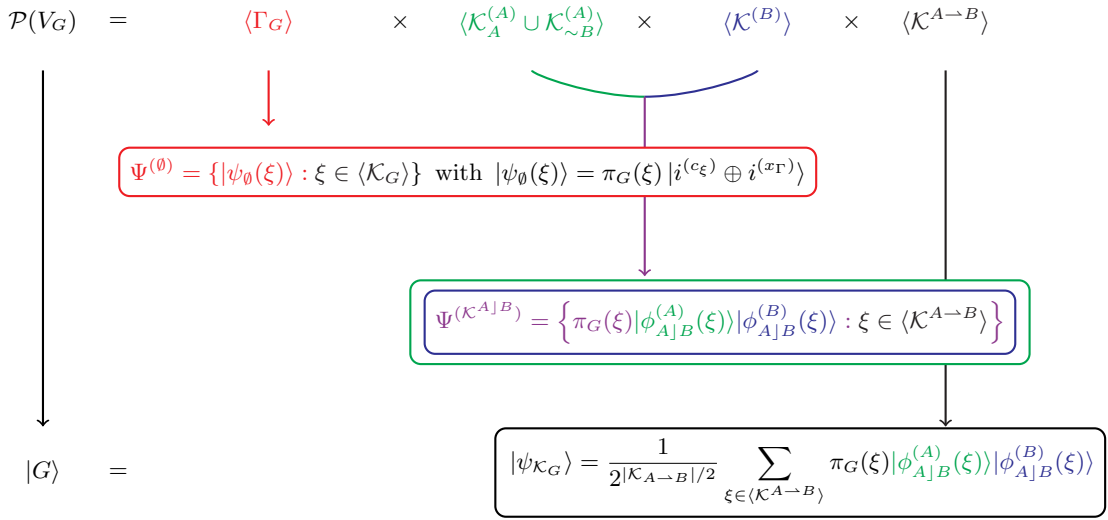

  \centering
  \ABfactorizationdiagram
  \caption{\colorfig X-chain factorization diagram for the Schmidt decomposition of graph states in the X-basis. A graphical summary of Lemmas \ref{lemma::biseparable_AB-correlation_states} and \ref{lemma::orthonormality_of_AB-correlation_states}, and Theorem \ref{theorem::Schmidt_decomposition_in_AB-correlation_states}.
  }
  \label{fig::factorization_diagram_AB_correlation}
\end{figure*}
\par
The $A\rfloor B$-factorization diagram of $|G_{\text{House}}\rangle$ is shown
in Fig. \ref{fig::house_graph_factorization_diagram}. As a result of this
theorem, the Schmidt decomposition of this state is
\begin{equation}
|G_{\text{House}}\rangle=\frac{1}{\sqrt{2}}\left(  |\psi_{A\rfloor
B}(\emptyset)\rangle+|\psi_{A\rfloor B}(\left\{ 2 \right\}  )\rangle
\right)
\end{equation}
with $|\psi_{A\rfloor B}(\emptyset)\rangle$ and $|\psi_{A\rfloor B}(\left\{2\right\}  )\rangle$ being given in Eqs.
\eqref{eq::house_state_A_B_correlation_state_1} and
\eqref{eq::house_state_A_B_correlation_state_2}. The house state has Schmidt
rank $r_{S}=2$ and the geometric measure of bipartite
entanglement $E_g^{(A|B)}=-2\log_{2}\left(  \min_{\psi}\left\vert \left\langle \psi
_{A}\psi_{B}|G_{\text{House}}\right\rangle \right\vert \right) = 1 $.
\subsubsection{Entanglement localization of graph states protected against errors}
\label{sec::unilateral_projection_against_errors}
\begin{figure*}[ht!]
  \centering
  \subfloat[]{
    \includegraphics[width=0.4\textwidth]{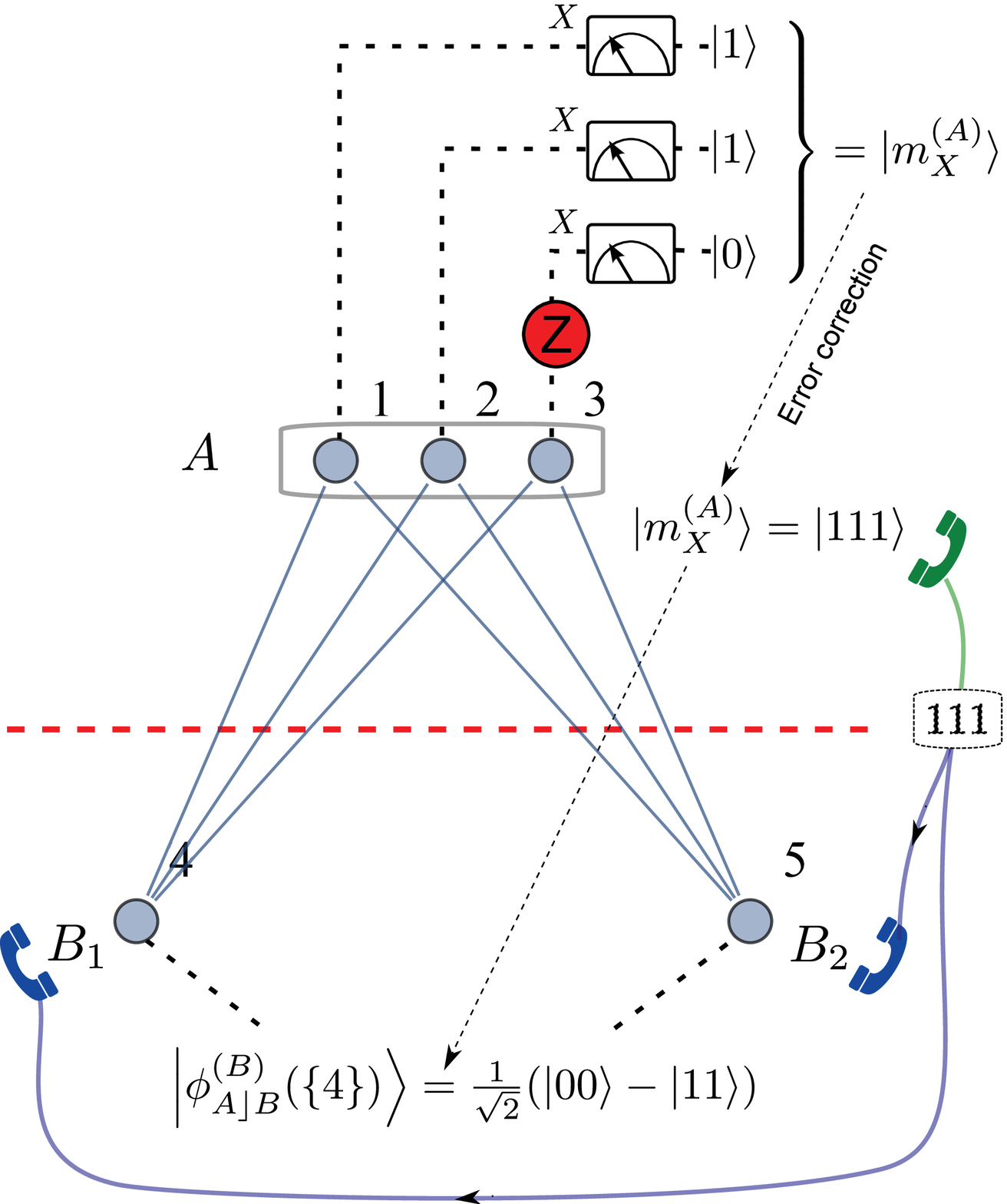}
    \label{fig::projection_against_error_bistar}
  }
  \subfloat[]{
    \includegraphics[width=0.55\textwidth]{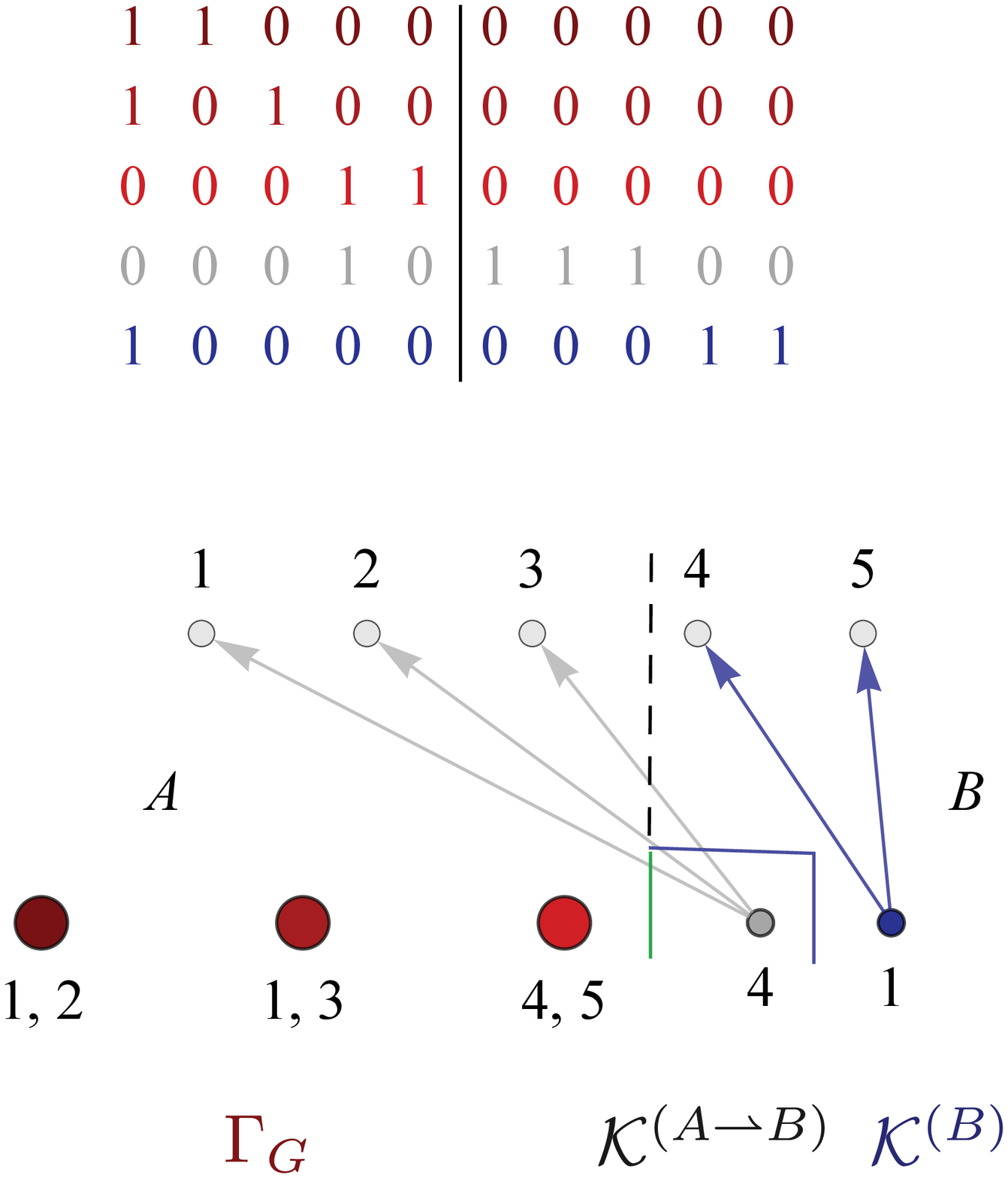}
    \label{fig::projection_against_error_AB-factorization}
  }%
  \caption{\colorfig An example of entanglement localization of graph states protected against errors:
  (a) Local X-measurements on subsystem $A$ project the graph state $\ket{G}$ onto the maximally entangled state $\ket{\phi_{A\rfloor B}^{(B)}(\xi)}$ for subsystem $B$.
  Under the assumption of a single qubit error, the outcome $\ket{m_X^{(A)}}=\ket{110}$ indicates a Z-error on vertex $3$.
  Alice sends Bob the corrected outcome $(111)$, such that Bob knows from the Schmidt decomposition that he possesses the state $\ket{\phi_{A\rfloor B}^{(B)}(\{4\})}$.
  (b) Binary representation and incidence structure after $A\rfloor B$-factorization.
  }
  \label{fig::projection_against_error}
\end{figure*}
In this section, we consider the localization of entanglement \cite{PoppVMCirac2005-LocalizableEnt} on graph states shared between Alice and Bob ($A|B$-bipartition); see Fig. \ref{fig::projection_against_error_bistar}.
Alice measures the graph state with Pauli-measurements on her system, then tells Bob her measurement results via a classical channel.
At the end, Bob should possess a bipartite maximally entangled state which he knows.
A connected graph state is maximally ``connected'' with respect to entanglement localization,
if every pair of vertices can be projected onto a Bell pair with local measurements \cite{GraphStateReviews2006}.
The simplest approach to localize the entanglement of $\ket{G}$ in the subsystem $\{B_1,B_2\}$ is finding a path between $B_1$ and $B_2$, then removing vertices outside the path with Z-measurements, and, at the end, measuring each vertex on the path between $\{B_1,B_2\}$ in the X-direction.
However, the resulting state depends on the measurement outcomes.
If errors occur in Alice's measurements, it will lead to a wrong state of Bob.
Therefore, error correction would be a nice feature in the entanglement localization of graph states.
\par
Graph states are stabilizer states.
These states can be exploited as quantum stabilizer codes \cite{Gottesman1996-QECSQHB, Gottesman1997-phdQEC, NielsenChuang2004-QCQI, GraphStateReviews2006}, which are linear codes and protect against errors.
In the Schmidt decomposition, the measurement outcomes on the system $A$ imply which states are projected in the system $B$.
The existence of X-chains on Alice's side can provide simple repetition codes as the Schmidt basis in the Schmidt decomposition in the X-basis.
Therefore, instead of removing the vertices outside a selected path between $B_1$ and $B_2$, we will make X-measurements on them to take the benefit of X-chains for the error correction.
\par
The graph state $\ket{G}$ in Fig. \ref{fig::projection_against_error_bistar} is taken as an example.
This state has the X-chain generating set $\Gamma_G=\{\{1,2\},\{1,3\},\{4,5\}\}$.
The generating set of the three correlation groups (Eq. \eqref{eq::def_KappaB_in_AB_bipartition}, \eqref{eq::def_KappaAA_in_AB_bipartition} and \eqref{eq::def_KappaA_B_in_AB_bipartition}) for the Schmidt decomposition are $\mathcal{K}_A^{(A)}=\mathcal{K}_{\sim B}^{(A)}=\emptyset$ and $\mathcal{K}^{(B)}=\{\{1\}\}$, while the generating set of the $(A\rightharpoonup B)$-correlation group is $\mathcal{K}^{(A\rightharpoonup B)}=\{\{4\}\}$.
According to Theorem \ref{theorem::Schmidt_decomposition_in_AB-correlation_states} and with the help of Algorithm \ref{algo::factorization_diagram_AB_correlation}, one has
\begin{equation}
|\psi_{A\rfloor B}(\emptyset)\rangle=\left\vert 000\right\rangle
\frac{\left\vert 00\right\rangle +\left\vert 11\right\rangle }{\sqrt{2}}%
\end{equation}
and%
\begin{equation}
|\psi_{A\rfloor B}(\{4\})\rangle=\left\vert 111\right\rangle \frac{\left\vert
00\right\rangle -\left\vert 11\right\rangle }{\sqrt{2}}.
\end{equation}
As a result, the Schmidt decomposition of the graph state is
\begin{equation}
|G\rangle=\frac{1}{\sqrt{2}}\left(  \left\vert 000\right\rangle \frac
{\left\vert 00\right\rangle +\left\vert 11\right\rangle }{\sqrt{2}}+\left\vert
111\right\rangle \frac{\left\vert 00\right\rangle -\left\vert 11\right\rangle
}{\sqrt{2}}\right)
.\label{eq::Schmidt_decomposition_in_unilateral_projection}%
\end{equation}
In this example, one observes that there are $2$ X-chain generators $\{1,2\}$ and $\{1,3\}$ on Alice's $3$-qubit system.
This encodes the following $[3,1,3] $ repetition code \cite{Gottesman1996-QECSQHB, Gottesman1997-phdQEC, NielsenChuang2004-QCQI} in the Schmidt vectors on Alice's system:
\begin{equation}
|\phi_{A\rfloor B}^{(A)}(0)\rangle=|000\rangle \text{ and }
|\phi_{A\rfloor B}^{(A)}(\{4\})\rangle=|111\rangle.
\end{equation}
These codes have the Hamming distance $3$.
Thus, a single Z-error can be corrected.
After a measurement in the X-basis, Alice can therefore correct her result before sending it to Bob.
In this approach, Bob will gain the correct acknowledgment of his maximally entangled state after Alice's measurement with confidence.
Although the repetition code cannot correct phase errors (the X-errors in X-measurements), it is already sufficient for our task, since a phase error on Alice's side does not change the measurement outcomes.
\par
This application may be useful for quantum repeaters \cite{BriegelDCZoller1998-QRepeater}.
The parties $B_1$ and $B_2$ can be at a large distance,
such that they are not able to directly create an entangled state between them.
In this case, they need the help from Alice as a repeater station to project the entanglement onto $B_1$ and $B_2$.

\section{Conclusions}
In this paper, we discussed properties of the representation of graph states in the computational X-basis.
We introduced the framework of X-resources and correlation indices
and linked them to the binary representation of graph states. 
A special type of X-resources was defined as X-chains: an X-chain is a subset of vertices for a given graph, such that the product of the stabilizer generators associated with these vertices contains only $\sigma_X$-Pauli operators.
The set of X-chains of a graph state is a group, which can be calculated efficiently \cite{WuKampermannBruss2015-XChainAlgo}.
The X-chain groups revealed structures of graph states and showed how to distinguish them by local $\sigma_X$ measurements.
We introduced X-chain factorization (Lemma \ref{lemma::X-chain_factorized_group}, \ref{theorem::graph_states_as_X-factorized_states}) for deriving the representation of graph states in the X-basis,
and it was shown that a graph state can be represented as superposition of all X-chain states (Theorem \ref{theorem::graph_states_as_X-factorized_states}).
This approach was illustrated in the so-called factorization diagram (Algorithm \ref{algo::factorization_diagram_graph_states}).
The larger the X-chain group is, the fewer X-chain states are needed for representing the graph state.
%
\par
We demonstrated various applications of the X-chain factorization.
An important application is its usefulness for efficiently determining the overlap of two graph states (Corollary \ref{coro::balanced_and_orthogonal_graph_states}) using our algorithm.
\par
Further, we generalized the X-chain factorization approach such that it allows to find the Schmidt decomposition of graph states, which is the superposition of appropriately selected correlation states (Theorem \ref{theorem::Schmidt_decomposition_in_AB-correlation_states}, Algorithm \ref{algo::factorization_diagram_AB_correlation} and MATHEMATICA package in the
Supplemental Material \cite{XchainMpackage_Wu}).
\par
Further benefits of the X-chain factorization are error correction procedures in entanglement localization of graph states in bipartite systems.
This could be useful for quantum repeaters \cite{BriegelDCZoller1998-QRepeater}.
\par
The results of this paper can be extended to general multipartite graph states, e.g. weighted graph states \cite{DurHHLBriegel2005-EntSpinChain, HartmannCDB2007-04} and hypergraph states \cite{RossiHBMacchiavello2013-11,QuWLBao2013-02,GuhneCSMRBKM2014}.
Another possible extension of these results is to consider the representation of graph states in a hybrid basis, i.e. for a subset of the qubits one adopts the X-basis, while for the other parties one uses the Z-basis.
The graph state in such a hybrid basis can even have a simpler representation (i.e. a smaller number of terms in the superposition) than the one obtained by X-chain factorization.
Besides, in \cite{HeinEisertBriegel2004-06,GraphStateReviews2006,CosentinoSimone2009-Weight,HajdusekMurao2013-Direct}, various multipartite entanglement measures for graph states were studied.
We expect that the approach of X-chain factorization may also be useful in these cases.

\begin{acknowledgments}
  This work was financially supported by the BMBF (Germany). We thank Michael Epping, Mio Murao and Yuki Mori for inspiration and useful discussions.
\end{acknowledgments}

\appendix

\section{Proofs}
\label{apdx::proofs_for_X-chains}
\label{sec::graph_states_proofs}
Proposition \ref{prop::isomorphism_of_vertex-induction_operation} shows the isomorphism between the stabilizer group and power set of the graph vertex set. It is proved as follows.
\begin{customprop}{\ref{prop::isomorphism_of_vertex-induction_operation}}
[Isomorphism of $\xi$-induction]%
Let $\left(
\mathcal{S}_{G},\cdot\right)  $ be the stabilizer group of a graph state
$|G\rangle$ and $\mathcal{P}(V_{G})$ be the power set of the vertex set of $G$.
The vertex-induction operation $s_{G}^{(\xi)}$ is a group isomorphism between
$(\mathcal{P}(V_{G}),\Delta)$ and $(\mathcal{S}_{G},\cdot)$, i.e.%
\begin{equation}
(\mathcal{P}(V_{G}),\Delta)\overset{s_{G}^{(\xi)}}{\sim}\left(  \mathcal{S}%
_{G},\cdot\right)  ,
\end{equation}
where $\Delta$ is the symmetric difference operation.
\end{customprop}
\begin{proof}[Proof of Proposition \ref{prop::isomorphism_of_vertex-induction_operation}]
Let $\xi_{1}$, $\xi_{2}\subseteq V_{G}$ be two vertex subsets. Since the
stabilizer group $\mathcal{S}_{G}$ is Abelian, one can resort the product as
$s_{G}^{(\xi_{1})}s_{G}^{(\xi_{2})}=\prod_{i\in\xi_{1}\Delta\xi_{2}}g_{i}%
\prod_{i^{\prime}\in\xi_{1}\cap\xi_{2}}(g_{i^{\prime}})^{2}$.
The property $(g_{i})^{2}=\id$ leads to $s_{G}^{(\xi_{1})}s_{G}^{(\xi_{2}%
)}=s_{G}^{(\xi_{1}\Delta\xi_{2})}$. Therefore $s_{G}^{(\xi)}$ is a group
homomorphism $(\mathcal{P}(V_{G}),\Delta)\overset{s_{G}^{(\xi)}}{\rightarrow
}\left(  \mathcal{S}_{G_{k}},\cdot\right)  $. The kernel of $s_{G}^{(\xi)}$ is
$\emptyset$, therefore $(\mathcal{P}(V_{G}),\Delta)\overset{s_{G}^{(\xi)}%
}{\sim}\left(  \mathcal{S}_{G},\cdot\right)  $.
\end{proof}
\par
Proposition \ref{prop::induced_stabilizer_math_formula} provides us with a mathematical expression of $\xi$-induce graph state stabilizer. It is proven by counting of the exchanging times of Pauli $X$ and $Z$ operators.
\begin{customprop}{\ref{prop::induced_stabilizer_math_formula}}
[Induced stabilizer]%
Let $\xi$ be a vertex subset of a graph $G$. The $\xi$-induced stabilizer (see Def. \ref{def::induced_stabilizer}) of a graph
state $|G\rangle$ is given by
\begin{equation}
s_{G}^{(\xi)}=\pi_{G}\left(  \xi\right)  \sigma_{X}^{(\xi)}\sigma_{Z}%
^{(c_{\xi})}%
\end{equation}
where $c_{\xi}$ is the \emph{correlation index} of $\xi$
and $\pi_{G}\left(\xi\right)$ is the stabilizer parity of $\xi$.
\end{customprop}
\begin{proof}
[Proof of Proposition \ref{prop::induced_stabilizer_math_formula}]
Let $\xi=\left\{  j_{1},...,j_{m}\right\}  $. Once we write down the $\xi
$-induced stabilizers explicitly, we have
\begin{equation}
s_{G}^{\left(  \xi\right)  }=\sigma_{x}^{(j_{1})}\sigma_{z}^{(N_{j_{1}})}%
\cdots\sigma_{x}^{(j_{m})}\sigma_{z}^{(N_{j_{m}})}%
\end{equation}
with $N_{j}$ being the neighborhood of $j$. Now we shift $\sigma_{x}$
operators to re-sort the expression such that all the $\sigma_{x}$ are on the
left side of $\sigma_{z}$. First, let us consider the last $X$-operator,
$\sigma_{x}^{(j_{m})}$. The number of $\sigma_{z}^{(j_{m})}$ on the left hand
side of $\sigma_{x}^{(j_{m})}$ indicates how many times one needs to exchange
$\sigma_{x}^{(j_{m})}$ and $\sigma_{z}^{(j_{m})}$. It is equal to the number
of neighbors of $j_{m}$ in the $\xi$-induced graph $G[\xi]$, namely $d_{j_{m}%
}(G[\xi])$. Due to the anti-commutativity of $\sigma_{x}$ and $\sigma_{z}$,
the shifting brings us a prefactor $(-1)^{d_{j_{m}}(G\left[  \xi\right]  )}$.
Recursively, shifting $\sigma_{x}^{(j_{m}-1)}$ to the left side of $\sigma
_{z}^{(j_{m}-1)}$ brings us a prefactor $(-1)^{d_{j_{k-1}}(G\left[
\xi\right]  -j_{m})}$, and so on.
In total, the times that one needs to exchange
$\sigma_{x}$ and $\sigma_{z}$ is
\begin{equation}
d_{j_{k}}(G\left[  \xi\right]  )+d_{j_{k-1}}(G\left[  \xi\right]
-j_{m})+\cdots+d_{2}(G\left[  \left\{  j_{1},j_{2}\right\}  \right]  ),
\end{equation}
which is equal to the edge number $\left\vert E(G[\xi])\right\vert $. Hence,
after the shifting, we obtain a product of re-sorted $\sigma_{x}$
and $\sigma_{z}$ operators with a prefactor $(-1)^{\left\vert E(G[\xi
])\right\vert }$, i.e.,%
\begin{equation}
s_{G}^{(\xi)}=(-1)^{\left\vert E(G[\xi])\right\vert }\sigma_{x}^{(j)}%
\sigma_{z}^{(n_{G}\left(  j_{1}\right)  )}\cdots\sigma_{z}^{(n_{G}\left(
j_{m}\right)  )},
\end{equation}
while $\sigma_{z}^{(N(j_{1}))}\cdots\sigma_{z}^{(N(j_{m}))}=\sigma_{z}%
^{c_{G}(\xi)}$.
\end{proof}
\bigskip
\par
Lemma \ref{lemma::X-chain_factorized_group} regroups the power set of vertices with factorization regarding the X-chain group into the correlation group. Accordingly, one can regroup the graph state projector by stabilizers induced by the correlation group. It is a result of Proposition \ref{prop::isomorphism_of_vertex-induction_operation}.
\begin{customlemma}{\ref{lemma::X-chain_factorized_group}}
[X-chain groups and correlation groups]%
Let $|G\rangle$ be a graph state.
The set of X-chains together with the symmetric difference
$(\mathcal{X}_{G}^{\left(  \emptyset\right)  },\Delta)$,
is a normal subgroup of $\left(  \mathcal{P}\left(  V_{G}\right)  ,\Delta\right)$.
The quotient group $(\mathcal{P} \left(  V_{G}\right)  /\mathcal{X}_{G}^{\left(\emptyset\right)},\Delta)$ is identical to the set of all resource sets
\begin{equation}
\mathcal{P}\left(  V_{G}\right)  /\mathcal{X}_{G}^{\left(  \emptyset\right)
}=\left\{  \mathcal{X}_{G}^{\left(  c\right)  }:c\in\mathcal{C}_{G}\right\},
\end{equation}
which we call call the \emph{correlation group} of $\ket{G}$.
Let $\Gamma_{G}$ and $\mathcal{K}_{G}$ denote the generating sets of
$(\mathcal{X}_{G}^{\left(  \emptyset\right)  },\Delta)$ and $(\mathcal{P}(V_{G})/\mathcal{X}_{G}^{(\emptyset)}, \Delta
)$, respectively.
The stabilizer group $\left(  \mathcal{S}_{G},\cdot\right)  $ is isomorphic to the direct product of the X-chain group and the correlation group,
\begin{equation}
\left(  \mathcal{S}_{G},\cdot\right)  \sim\left(  \left\langle \Gamma
_{G}\right\rangle ,\Delta\right)  \times\left(  \left\langle \mathcal{K}%
_{G}\right\rangle ,\Delta\right)
,
\end{equation}
As a result, the graph state $|G\rangle$ is the product of the X-chain group and correlation group inducing stabilizers, i.e.%
\begin{equation}
|G\rangle\langle G|=\prod_{\kappa\in\mathcal{K}_{G}}\frac{1+s_{G}^{(\kappa)}%
}{2}\prod_{\gamma\in\Gamma_{G}}\frac{1+s_{G}^{(\gamma)}}{2}%
.%
\end{equation}
\end{customlemma}
\begin{proof}
[Proof of Lemma \ref{lemma::X-chain_factorized_group}]
Let $\xi_{1}$ and $\xi_{2}$ be two elements of $\mathcal{X}_{G}^{\left(
c\right)  }$. The correlation index mapping, $c_{G}:\left(  \mathcal{P}%
(V_{G}),\Delta\right)  \rightarrow\left(  \mathcal{C}_{G},\Delta\right)  $, is
a group homomorphism, since $c_{G}\left(  \xi_{1}\Delta\xi_{2}\right)
=c_{G}\left(  \xi_{1}\right)  \Delta c_{G}\left(  \xi_{2}\right)  $. Due to
the definition of X-chains that $c_{G}(\xi)=0$, $(\mathcal{X}_{G}^{\left(
\emptyset\right)  },\Delta)$ is the kernel of the mapping $c_{G}$. Since
$\left(  \mathcal{P}\left(  V_{G}\right)  ,\Delta\right)  $ is Abelian, the
kernel $(\mathcal{X}_{G}^{\left(  \emptyset\right)  },\Delta)$ and the
correlation group $\mathcal{P}\left(  V_{G}\right)  /\mathcal{X}_{G}^{\left(  \emptyset\right)
}$ are both normal subgroups.
The correlation group $\braket{\mathcal{K}_G}$ is obtained via
\begin{align}
\braket{\mathcal{K}_G}  & =\mathcal{P}\left(  V_{G}\right)  /\mathcal{X}%
_{G}^{\left(  \emptyset\right)  }\nonumber\\
& =\left\{  \xi\Delta\mathcal{X}_{G}^{\left(  \emptyset\right)  }:\xi
\in\mathcal{P}\left(  V_{G}\right)  \right\}  \nonumber\\
& =\left\{  \mathcal{X}_{G}^{\left(  c\right)  }:c\in\mathcal{C}_{G}\right\}.
\end{align}
As a result of group theory,
\begin{equation}
\left(  \mathcal{P}(V_{G}),\Delta\right)  =\left(  \left\langle \Gamma
_{G}\right\rangle ,\Delta\right)  \times\left(  \left\langle \mathcal{K}%
_{G}\right\rangle ,\Delta\right)  .
\end{equation}
According to Proposition \ref{prop::isomorphism_of_vertex-induction_operation}%
, one obtains the isomorphism%
\begin{equation}
\left(  \mathcal{S}_{G_{k}},\cdot\right)  \overset{s_{G_{k}}^{(\xi)}}{\sim
}\left(  \left\langle \Gamma_{G}\right\rangle ,\Delta\right)  \times\left(
\left\langle \mathcal{K}_{G}\right\rangle ,\Delta\right)  .
\end{equation}
The projector of graph state $|G\rangle\langle G|$ is the sum of all $\xi
$-induced stabilizers, $s_{G}^{(\xi)}$, with $\xi\in\left\langle \Gamma
_{G}\right\rangle \times\left\langle \mathcal{K}_{G}\right\rangle $. As a
result,%
\begin{align}
|G\rangle\langle G| &  =\sum_{\xi\in\mathcal{P}\left(  V_{G}\right)  }%
s_{G}^{(\xi)}\nonumber\\
&  =\prod_{\kappa\in\mathcal{K}_{G}}\frac{1+s_{G}^{(\kappa)}}{2}\prod
_{\gamma\in\mathcal{\Gamma}_{G}}\frac{1+s_{G}^{(\gamma)}}{2}.
\end{align}

\end{proof}
\begin{customprop}{\ref{prop::X-chain_states_in_X-basis}}
[X-chain states in X-basis]
Let $|G\rangle$ be a graph state with the X-chain group $\left\langle \Gamma
_{G}\right\rangle $ and the correlation group $\left\langle \mathcal{K}%
_{G}\right\rangle $. Let $\Gamma_{G}=\left\{  \gamma_{1},\gamma_{2}%
,...\right\}  $, and $\gamma_{i}=\left\{  v_{i_{1}},v_{i_{2}},\cdots\right\}
$. The generating set $\Gamma_{G}$ and $\mathcal{K}_{G}$ can be chosen as

\begin{enumerate}
\item $\Gamma_G=\{\gamma_1,...,\gamma_k\}$ such that $\gamma_i\not\subseteq\gamma_j$ for all $\gamma_i, \gamma_j\in\Gamma_G$,
\item $\mathcal{K}_{G}=\left\{  \left\{  v\right\}  :v\in V_{G}\backslash
\bigcup_{i=1}^{k}\left\{  v_{i_{1}}\right\}  \right\}  $.
\end{enumerate}
Here, the first element of $\gamma_i=\{v_{i_1},v_{i_2},...\}$ is selected in a way such that $v_{i_{1}}\neq v_{j_{1}}$ for all $i\not =j$.
Then the X-chain state $\ket{\psi_{\emptyset}(\emptyset)}$ of $|G\rangle$ is an X-basis state, $|i^{(x_{\Gamma})}\rangle$, with%
\begin{equation}
x_{\Gamma}=\left\{  v_{i_{1}}:\pi_{G}\left(  \gamma_{i}\right)  =-1\right\}  .
\end{equation}
\end{customprop}
\begin{proof}[Proof of Proposition \ref{prop::X-chain_states_in_X-basis}]
Let $\gamma_{i}^{-}$ be an X-chain generator with negative parity $\pi
_{G}\left(  \gamma_{i}^{-}\right)  =-1$, then $v_{i_{1}}\in x_{\Gamma}$. Since
$v_{i_{1}}\in\gamma_{i}$ and $v_{i_{1}}\not \in \gamma_{j}$ for all $j\not =%
i$, the intersection $\gamma_{i}^{-}\cap x_{\Gamma}=\left\{  v_{i_{1}%
}\right\}  $, hence $\pi_{G}\left(  \gamma_{i}^{-}\right)  \sigma_{x}%
^{(\gamma_{i})}|i^{(x_{\Gamma})}\rangle=|i^{(x_{\Gamma})}\rangle$. For an
X-chain generator $\gamma_{j}^{+}$ with positive parity $\pi_{G}\left(
\gamma_{i}^{+}\right)  =1$, the intersection $\gamma_{i}^{+}\cap x_{\Gamma
}=\emptyset$, and therefore $\pi_{G}\left(  \gamma_{i}^{+}\right)  \sigma
_{x}^{(\gamma_{i})}|i^{(x_{\Gamma})}\rangle=|i^{(x_{\Gamma})}\rangle$. Hence
the condition \ref{item::def_X-chain_state_condition_1} in Definition
\ref{def::X-chain_states_K-correlation_states} is fulfilled.

Let $\left\{  v\right\}  \in\mathcal{K}_{G}$ be a generator of correlation
group, then $\sigma_{x}^{(\{v\})}|i^{(x_{\Gamma})}\rangle=\left(  -1\right)
^{\left\vert x_{\Gamma}\cap\left\{  v\right\}  \right\vert }|i^{(x_{\Gamma}%
)}\rangle=|i^{(x_{\Gamma})}\rangle$, since $\left\vert x_{\Gamma}\cap\left\{
v\right\}  \right\vert =0$ according to the choice of $\mathcal{K}_{G}$.
Hence, the condition \ref{item::def_X-chain_state_condition_2} in Definition
\ref{def::X-chain_states_K-correlation_states} is fulfilled.
\end{proof}
\bigskip
The Proposition \ref{prop::correlation_states_explicit_form} derives the correlation states as the summation of X-chain states. It follows directly from their definition.
\begin{customprop}{\ref{prop::correlation_states_explicit_form}}
[Form of X-chain states, $\mathcal{K}$-correlation states]%
Let $\xi\in\left\langle
\mathcal{K}_{G}\right\rangle $ be an X-resource and $\left\langle
\mathcal{K}\right\rangle \subseteq\left\langle \mathcal{K}_{G}\right\rangle $.
An X-chain state is given as
\begin{equation}
|\psi_{\emptyset}\left(  \xi\right)  \rangle=\pi_{G}\left(  \xi\right)
\left\vert i^{(x_{\Gamma})}\oplus i^{(c_{\xi})}\right\rangle,
\end{equation}
where $\pi_{G}\left(  \xi\right)  $ is the stabilizer parity of $\xi$ (see Eq. \eqref{eq::G-parity_formula}),
and $c_{\xi}$ is the correlation index of $\xi$.

A $\mathcal{K}$-correlation state is the superposition of \emph{X-chain
states},\emph{ }
\begin{equation}
|\psi_{\mathcal{K}}\left(  \xi\right)  \rangle=\frac{1}{2^{\left\vert
\mathcal{K}\right\vert /2}}\sum_{\xi^{\prime}\in\left\langle \mathcal{K}%
\right\rangle }|\psi_{\emptyset}\left(  \xi\Delta\xi^{\prime}\right)  \rangle.
\end{equation}
\end{customprop}

\begin{proof}[Proof of Proposition \ref{prop::correlation_states_explicit_form}]
According to Proposition \ref{prop::isomorphism_of_vertex-induction_operation}%
, $s_{G}^{(\xi)}\Delta s_{G}^{(\xi^{\prime})}=s_{G}^{(\xi\Delta\xi^{\prime})}%
$, the product of the operators in Eq. (\ref{eq::def_correlation_states}) can
be reformulated to the sum of%
\begin{equation}
|\psi_{\mathcal{K}}\left(  \xi\right)  \rangle=\frac{1}{2^{\left\vert
\mathcal{K}\right\vert /2}}\sum_{\xi^{\prime}\in\left\langle \mathcal{K}%
\right\rangle }s_{G}^{(\xi\Delta\xi^{\prime})}\left\vert i^{(x_{\Gamma}%
)}\right\rangle .
\end{equation}
With the formulas in Proposition \ref{prop::induced_stabilizer_math_formula},
\begin{equation}
|\psi_{\mathcal{K}}\left(  \xi\right)  \rangle=\frac{1}{2^{\left\vert
\mathcal{K}\right\vert /2}}\sum_{\xi^{\prime}\in\left\langle \mathcal{K}%
\right\rangle }\pi_{G}\left(  \xi\Delta\xi^{\prime}\right)  \sigma
_{z}^{(c_{\xi\Delta\xi^{\prime}})}\sigma_{x}^{(\xi\Delta\xi^{\prime}%
)}\left\vert i^{(x_{\Gamma})}\right\rangle .
\end{equation}
Since $\sigma_{x}^{(\kappa)}\left\vert i^{(x_{\Gamma})}\right\rangle
=\left\vert i^{(x_{\Gamma})}\right\rangle $ for all $\kappa\in\left\langle
\mathcal{K}\right\rangle $, one obtains%
\begin{equation}
|\psi_{\mathcal{K}}\left(  \xi\right)  \rangle=\frac{1}{2^{\left\vert
\mathcal{K}\right\vert /2}}\sum_{\xi^{\prime}\in\left\langle \mathcal{K}%
\right\rangle }\pi_{G}\left(  \xi\Delta\xi^{\prime}\right)  \left\vert
i^{(x_{\Gamma})}\oplus i^{(c_{\xi\Delta\xi^{\prime}})}\right\rangle .
\end{equation}

\end{proof}
\bigskip
The following calculation is the proof of Corollary \ref{coro::properties_of_correlation_states}, which is employed in the proof of Theorem \ref{theorem::graph_states_as_X-factorized_states}.
\begin{proof}[Proof of Corollary \ref{coro::properties_of_correlation_states}]
The properties (\ref{item::correlation_states_property_2}) and
(\ref{item::correlation_states_property_3}) are direct results of Proposition
\ref{prop::correlation_states_explicit_form}.
Property (\ref{item::correlation_states_property_1}) follows from the commutativity of graph
state stabilizers. Let $\kappa=\gamma\Delta\kappa_{0}\in\left\langle
\Gamma_{G}\right\rangle \times\left\langle \mathcal{K}\right\rangle $ with
$\gamma\in\left\langle \Gamma_{G}\right\rangle $ and $\kappa_{0}%
\in\left\langle \mathcal{K}\right\rangle $; then%
\begin{align*}
s_{G}^{(\kappa)}\phi_{\mathcal{K}}\left(  \xi\right)   &  =s_{G}^{(\gamma
)}s_{G}^{(\kappa_{0})}s_{G}^{(\xi)}\prod_{\kappa^{\prime}\in\mathcal{K}}%
\frac{1+s_{G}^{(\kappa^{\prime})}}{\sqrt{2}}|x_{\Gamma}\rangle\\
&  =s_{G}^{(\xi)}\frac{1}{2^{\left\vert \mathcal{K}\right\vert /2}}%
s_{G}^{(\kappa_{0})}\sum_{\kappa^{\prime}\in\left\langle \mathcal{K}%
\right\rangle }s_{G}^{(\kappa^{\prime})}s_{G}^{(\gamma)}|x_{\Gamma}\rangle.
\end{align*}
Due to the definition of $|x_{\Gamma}\rangle$, it holds that $s_{G}^{(\gamma)}%
|x_{\Gamma}\rangle=|x_{\Gamma}\rangle$. Since $\kappa_{0}\in\left\langle \mathcal{K}%
\right\rangle $, the operator $s_{G}^{(\kappa_{0})}\sum_{\kappa^{\prime}%
\in\left\langle \mathcal{K}\right\rangle }s_{G}^{(\kappa^{\prime})}%
=\sum_{\kappa^{\prime}\in\left\langle \mathcal{K}\right\rangle }s_{G}%
^{(\kappa^{\prime})}$ is not changed by $s_{G}^{(\kappa_{0})}$, hence%
\begin{equation}
s_{G}^{(\kappa)}\ket{\phi_{\mathcal{K}}\left(  \xi\right)}  =s_{G}^{(\xi)}\frac
{1}{2^{\left\vert \mathcal{K}\right\vert /2}}\sum_{\kappa^{\prime}%
\in\left\langle \mathcal{K}\right\rangle }s_{G}^{(\kappa^{\prime})}%
|x_{\Gamma}\rangle=\ket{\phi_{\mathcal{K}}\left(  \xi\right)}  .
\end{equation}

\end{proof}
\bigskip
\par
Lemma \ref{lemma::biseparable_AB-correlation_states} shows us the $A|B$-separability of the correlation state $|\psi_{A\rfloor B}(\xi)\rangle$. It is a result of the property of multiplication $G$-parities.
\begin{lemma}
[Multiplication of $G$-parity]\label{lemma::multiplication_of_G_parity}Let $G$
be a graph; then the multiplication of two of the parities of two-vertex subset
$\pi_{G}\left(  \xi_{1}\right)  $ and $\pi_{G}\left(  \xi_{2}\right)  $ is
equal to
\begin{equation}
\pi_{G}\left(  \xi_{1}\right)  \pi_{G}\left(  \xi_{2}\right)  =\left(
-1\right)  ^{\left\vert E_{G}(\xi_{1}:\xi_{2})\right\vert }\pi_{G}\left(
\xi_{1}\Delta\xi_{2}\right)  .
\end{equation}

\begin{proof}
Since $\left(  P(V),\Delta\right)  $ is isomorphic to the stabilizer group
$\left(  S_{G},\cdot\right)  $, it holds then that%
\[
s_{G}^{(\xi_{1})}s_{G}^{(\xi_{2})}=s_{G_{k}}^{(\xi_{1}\Delta\xi_{2})}.
\]
Reorder the $\sigma_{x}$ and $\sigma_{z}$ in both sides, such that $\sigma
_{x}$ are on the left side of $\sigma_{z}$; then one obtains
\begin{equation}
\pi_{G}\left(  \xi_{1}\right)  \pi_{G}\left(  \xi_{2}\right)  =\left(
-1\right)  ^{\left\vert E_{G}(\xi_{1}:\xi_{2})\right\vert }\pi_{G}\left(
\xi_{1}\Delta\xi_{2}\right)  .
\end{equation}

\end{proof}
\end{lemma}

With this lemma one can prove Lemma \ref{lemma::biseparable_AB-correlation_states} as follows.
\begin{widetext}

\begin{customlemma}{\ref{lemma::biseparable_AB-correlation_states}}
[$A|B$-Separability of $A\rfloor B$-correlation states]%
For $\xi\in\langle
\mathcal{K}_{G}^{A\rightharpoonup B}\rangle$, the $\left(  A\rightharpoonup
B\right)  $-correlation states%
\begin{equation}
|\psi_{A\rfloor B}(\xi)\rangle=\pi_{G}\left(  \xi\right)  |\phi_{A\rfloor
B}^{(A)}(\xi)\rangle|\phi_{A\rfloor B}^{(B)}(\xi)\rangle
\end{equation}
are $A|B$-separable with $|\phi_{A\rfloor B}^{(A)}(\xi)\rangle:=|\psi
_{\mathcal{K}_{A}^{(A)}\cup\mathcal{K}_{\sim B}^{(A)}}^{(A)}(\xi)\rangle$
and $|\phi_{A\rfloor B}^{(B)}(\xi)\rangle:=|\psi_{\mathcal{K}^{(B)}}^{(B)}(\xi)\rangle$ being the $(  \mathcal{K}_{A}^{(A)}\cup\mathcal{K}%
_{\sim B}^{(A)})$- and $\mathcal{K}^{(B)}$-correlation states
projected into the subspaces of $A$ and $B$, respectively.
\end{customlemma}

\begin{proof}
[Proof of Lemma \ref{lemma::biseparable_AB-correlation_states}]
According to Proposition \ref{prop::correlation_states_explicit_form},
\begin{equation}
|\psi_{A\rfloor B}(\xi)\rangle=\sum_{\xi^{\prime}\in\left\langle
\mathcal{K}_{A}^{(A)}\cup\mathcal{K}_{\sim B}^{(A)}\cup
\mathcal{K}_{{}}^{(B)}\right\rangle }\pi_{G}\left(  \xi^{\prime}\Delta
\xi\right)  |x_{\Gamma}\oplus c_{\xi^{\prime}}\oplus c_{\xi}\rangle
.%
\end{equation}
Each X-resource $\xi^{\prime}\in\left\langle \mathcal{K}_{A}^{(A)}%
\cup\mathcal{K}_{\sim B}^{(A)}\cup\mathcal{K}_{{}}^{(B)}%
\right\rangle $ can be decomposed as $\xi^{\prime}=\alpha\Delta\beta
=\alpha_{A}\Delta\alpha_{B}\Delta\beta$ with $\alpha_{A}\in\left\langle
\mathcal{K}_{A}^{(A)}\right\rangle $ and $\alpha_{B}\in\left\langle
\mathcal{K}_{\sim B}^{(A)}\right\rangle $ and $\beta\in\left\langle
\mathcal{K}_{{}}^{(B)}\right\rangle $. Due to Lemma
\ref{lemma::multiplication_of_G_parity}
\begin{equation}
\pi_{G}\left(  \xi^{\prime}\Delta\xi\right)  =\left(  -1\right)  ^{\left\vert
E_{G}(\xi:\alpha\Delta\beta)\right\vert }\pi_{G}\left(  \alpha\Delta
\beta\right)  \pi_{G}\left(  \xi\right)  \text{ and }\pi_{G}\left(
\alpha\Delta\beta\right)  =\left(  -1\right)  ^{\left\vert E_{G}(\alpha
:\beta)\right\vert }\pi_{G}\left(  \alpha\right)  \pi_{G}\left(  \beta\right)
\end{equation}
Since $\alpha_{A}\subseteq A$ and $c_{\beta}\subseteq B$, it holds $\left\vert
E_{G}(\alpha_{A}:\beta)\right\vert \overset{\operatorname{mod}2}{=}\left\vert
\alpha_{A}\cap c_{\beta}\right\vert \overset{\operatorname{mod}2}{=}0$, while
$\alpha_{B}\in\left\langle \mathcal{K}_{\sim B}^{(A)}\right\rangle $ is
defined by $\left\vert E_{G}(\alpha_{B}:\beta)\right\vert \overset
{\operatorname{mod}2}{=}$ $0$. Therefore the edge number is $\left\vert
E_{G}(\alpha:\beta)\right\vert =\left\vert E_{G}(\alpha_{A}\Delta\alpha
_{B}:\beta)\right\vert \overset{\operatorname{mod}2}{=}0$. In addition,
$\left\vert E_{G}(\xi:\alpha\Delta\beta)\right\vert \overset
{\operatorname{mod}2}{=}\left\vert E_{G}(\xi:\alpha)\right\vert +\left\vert
E_{G}(\xi:\beta)\right\vert $, and therefore%
\[
\pi_{G}\left(  \xi^{\prime}\Delta\xi\right)  =\left(  -1\right)  ^{\left\vert
E_{G}(\xi:\alpha)\right\vert }\left(  -1\right)  ^{\left\vert E_{G}(\xi
:\beta)\right\vert }\pi_{G}\left(  \alpha\right)  \pi_{G}\left(  \beta\right)
\pi_{G}\left(  \xi\right)
\]
According to Eq. (\ref{eq::G-parity_formula}) in Proposition
\ref{prop::induced_stabilizer_math_formula}, the following equation holds
\begin{equation}
\pi_{G}\left(  \xi\right)  \left(  -1\right)  ^{\left\vert E_{G}(\xi
:\beta)\right\vert }\pi_{G}\left(  \beta\right)  =\left(  -1\right)
^{\left\vert E(G\left[  \beta\right]  )\right\vert +\left\vert E_{G}(\xi
:\beta)\right\vert +\left\vert E(G\left[  \xi\right]  )\right\vert }=\left(
-1\right)  ^{\left\vert E(G\left[  \beta\Delta\xi\right]  )\right\vert }%
=\pi_{G}\left(  \xi\Delta\beta\right)  .
\end{equation}
This equality also holds for $\alpha$, and therefore%
\begin{equation}
\pi_{G}\left(  \xi^{\prime}\Delta\xi\right)  =\pi_{G}\left(  \alpha\Delta
\xi\right)  \pi_{G}\left(  \beta\Delta\xi\right)  \pi_{G}\left(  \xi\right)  .
\end{equation}
Insert this equality into Eq. \eqref{eq::graph_state_in_X_basis}, and one obtains%
\begin{equation}
|\psi_{A\rfloor B}(\xi)\rangle=\pi_{G}\left(  \xi\right)  \sum_{\alpha
\in\left\langle \mathcal{K}_{A}^{(A)}\right\rangle }\sum_{\beta\in\left\langle
\mathcal{K}_{{}}^{(B)}\right\rangle }\pi_{G}\left(  \alpha\Delta\xi\right)
\pi_{G}\left(  \beta\Delta\xi\right)  |i^{(x_{\Gamma})}\oplus i^{(c_{\alpha}%
)}\oplus i^{(c_{\beta})}\oplus i^{(c_{\xi})}\rangle=\pi_{G}\left(  \xi\right)
|\phi_{A\rfloor B}^{(A)}(\xi)\rangle|\phi_{A\rfloor B}^{(B)}(\xi)\rangle
\end{equation}
with%
\begin{equation}
|\phi_{A\rfloor B}^{(A)}(\xi)\rangle=|\psi_{\mathcal{K}_{A}^{(A)}%
\cup\mathcal{K}_{\sim B}^{(A)}}^{(A)}(\xi)\rangle=\sum_{\alpha
\in\left\langle \mathcal{K}_{A}^{(A)}\cup\mathcal{K}_{\sim B}%
^{(A)}\right\rangle }\pi_{G}\left(  \alpha\Delta\xi\right)  |i^{(x_{\Gamma
}^{(A)})}\oplus i^{(c_{\alpha})}\oplus i^{(c_{\xi}^{(A)})}\rangle,
\end{equation}
and%
\begin{equation}
|\phi_{A\rfloor B}^{(B)}(\xi)\rangle=|\psi_{\mathcal{K}^{(B)}}^{(B)}%
(\xi)\rangle=\sum_{\beta\in\left\langle \mathcal{K}_{{}}^{(B)}\right\rangle
}\pi_{G}\left(  \beta\Delta\xi\right)  |i^{(x_{\Gamma}^{(B)})}\oplus
i^{(c_{\beta})}\oplus i^{(c_{\xi}^{(B)})}\rangle.
\end{equation}
\end{proof}

Lemma \ref{lemma::orthonormality_of_AB-correlation_states} is the key to derive Theorem \ref{theorem::Schmidt_decomposition_in_AB-correlation_states}. Its proof is as follows.
\begin{customlemma}{\ref{lemma::orthonormality_of_AB-correlation_states}}
[Orthonormality of $(A\rightharpoonup B)$-correlation states]
The components of $A\rfloor B$-correlation states on subspace $A$ and $B$, $|\phi_{A\rfloor B}^{(A)}%
(\xi)\rangle$ and $|\phi_{A\rfloor B}^{(B)}(\xi)\rangle$, are orthonormal with
respect to $\xi\in\langle\mathcal{K}^{A\rightharpoonup B}\rangle$ within the subspaces
$A$ and $B$, respectively, i.e.,
\begin{equation}
\langle\phi_{A\rfloor B}^{(A)}(\xi_{1})|\phi_{A\rfloor B}^{(A)}(\xi
_{2})\rangle=0
\end{equation}
and%
\[
\langle\phi_{A\rfloor B}^{(B)}(\xi_{1})|\phi_{A\rfloor B}^{(B)}(\xi
_{2})\rangle=0
\]
for all $\xi_{1},\xi_{2}\in\langle\mathcal{K}^{A\rightharpoonup B}\rangle$ and
$\xi_{1}\not =\xi_{2}$.
\end{customlemma}
\begin{proof}
[Proof of Lemma \ref{lemma::orthonormality_of_AB-correlation_states}]
According to the definition of correlation states (Def.
\ref{def::X-chain_states_K-correlation_states}) and the unitarity of
stabilizer $s_{G}^{(\xi_{1})}$ and $s_{G}^{(\xi_{2})}$, it holds that%
\begin{equation}
\langle\phi_{\mathcal{K}}(\xi_{1})|\phi_{\mathcal{K}}(\xi_{2})\rangle
=\langle\phi_{\mathcal{K}}(0)|\phi_{\mathcal{K}}(\xi_{1}\Delta\xi_{2})\rangle.
\end{equation}
One just needs to consider the overlap $\langle\phi_{A\rfloor B}^{(A)}%
(0)|\phi_{A\rfloor B}^{(A)}(\xi)\rangle$ and $\langle\phi_{A\rfloor B}%
^{(B)}(0)|\phi_{A\rfloor B}^{(B)}(\xi)\rangle$ with $\xi\in\langle
\mathcal{K}_{G}^{A\rightharpoonup B}\rangle$. That means for all $\alpha
\in\langle\mathcal{K}_{A}^{(A)}\cup\mathcal{K}_{\sim B}^{(A)}\rangle$
and $\beta\in\langle\mathcal{K}_{{}}^{(B)}\rangle$, it holds that%
\begin{equation}
c_{\alpha}\oplus c_{\beta}\not =c_{\xi}%
.\label{eq::proof_of_orthonormality_AB_correlation_states}%
\end{equation}
For $|\phi_{A\rfloor B}^{(A)}(\xi)\rangle$, due to the commutativity of graph
state stabilizers, it holds that
\begin{align}
\langle\phi_{A\rfloor B}^{(A)}(0)|\phi_{A\rfloor B}^{(A)}(\xi)\rangle &
=\frac{1}{2^{\left\vert \mathcal{K}_{A}^{(A)}\cup\mathcal{K}_{\sim %
B}^{(A)}\right\vert }}\sum_{\alpha,\alpha^{\prime}\in\left\langle
\mathcal{K}_{A}^{(A)}\cup\mathcal{K}_{\sim B}^{(A)}\right\rangle }%
\pi_{G}\left(  \alpha^{\prime}\right)  \pi_{G}\left(  \alpha\Delta\xi\right)
\langle i(x_{\Gamma}^{(A)})\oplus i(c_{\alpha^{\prime}})|i(x_{\Gamma}^{(A)})\oplus
i(c_{\alpha})\oplus i(c_{\xi}^{(A)})\rangle\\
&  =\frac{1}{2^{\left\vert \mathcal{K}_{A}^{(A)}\cup\mathcal{K}%
_{\sim B}^{(A)}\right\vert }}\sum_{\alpha,\alpha^{\prime}%
\in\left\langle \mathcal{K}_{A}^{(A)}\cup\mathcal{K}_{\sim B}%
^{(A)}\right\rangle :c_{\alpha\Delta\alpha^{\prime}}=c_{\xi}^{(A)}}\pi
_{G}\left(  \alpha^{\prime}\right)  \pi_{G}\left(  \alpha\Delta\xi\right)  .
\end{align}
If there exists $\lambda\in\left\langle \mathcal{K}_{A}^{(A)}\cup
\mathcal{K}_{\sim B}^{(A)}\right\rangle $ such that $c_{\lambda}%
=c_{\xi}^{(A)}$, then $\xi\Delta\lambda\in\mathcal{K}^{(B)}$ (since
$c_{\xi\Delta\lambda}\subseteq B$). This means $\xi\in\left\langle
\mathcal{K}_{A}^{(A)}\cup\mathcal{K}_{\sim B}^{(A)}\right\rangle
\times\left\langle \mathcal{K}^{(B)}\right\rangle $, which is in contradiction to
the definition of $\xi\in\langle\mathcal{K}^{A\rightharpoonup B}\rangle$.
Hence there are no pairs $\alpha,\alpha^{\prime}\in\mathcal{K}_{A}^{(A)}$ such
that $c_{\alpha\Delta\alpha^{\prime}}=c_{\xi}^{(A)}$, and therefore%
\begin{equation}
\langle\phi_{A\rfloor B}^{(A)}(0)|\phi_{A\rfloor B}^{(A)}(\xi)\rangle=0.
\end{equation}
Analogously for $|\phi_{A\rfloor B}^{(B)}(\xi)\rangle$, it holds that%
\begin{equation}
\langle\phi_{A\rfloor B}^{(B)}(0)|\phi_{A\rfloor B}^{(B)}(\xi)\rangle=\frac
{1}{2^{\left\vert \mathcal{K}^{(B)}\right\vert }}\sum_{\beta,\beta^{\prime}%
\in\left\langle \mathcal{K}^{(B)}\right\rangle :c_{\beta\Delta\beta^{\prime}%
}=c_{\xi}^{(B)}}\pi_{G}\left(  \beta^{\prime}\right)  \pi_{G}\left(
\beta\Delta\xi\right)  .
\end{equation}
If there exists no $\lambda\in\mathcal{K}^{(B)}$ such that $c_{\lambda}%
=c_{\xi}^{(B)}$, then $\langle\phi_{A\rfloor B}^{(B)}(0)|\phi_{A\rfloor
B}^{(B)}(\xi)\rangle=0$. If there exist such $\lambda\in\mathcal{K}^{(B)}$
then we substitute $\xi$ by $\xi^{\prime}:=\xi\Delta\lambda$, and then
$\xi^{\prime}\in\left\langle \mathcal{K}^{A\rightharpoonup B}\right\rangle $
still holds and $c_{\xi^{\prime}}^{(B)}=0$. Hence the overlap becomes%
\begin{align}
\langle\phi_{A\rfloor B}^{(B)}(0)|\phi_{A\rfloor B}^{(B)}(\xi)\rangle &
=\frac{1}{2^{\left\vert \mathcal{K}^{(B)}\right\vert }}\sum_{\beta
\in\left\langle \mathcal{K}^{(B)}\right\rangle }\pi_{G}\left(  \beta\right)
\pi_{G}\left(  \beta\Delta\xi^{\prime}\right)  =\frac{1}{2^{\left\vert
\mathcal{K}^{(B)}\right\vert }}\sum_{\beta\in\left\langle \mathcal{K}%
^{(B)}\right\rangle }\left(  -1\right)  ^{\left\vert E_{G}\left(  \beta
:\xi^{\prime}\right)  \right\vert }\nonumber\\
&  =\frac{1}{2^{\left\vert \mathcal{K}^{(B)}\right\vert }}\prod_{\beta
\in\mathcal{K}^{(B)}}\left(  1+\left(  -1\right)  ^{\left\vert E_{G}\left(
\beta:\xi^{\prime}\right)  \right\vert }\right)  .
\end{align}
Since $\xi^{\prime}\in\left\langle \mathcal{K}^{A\rightharpoonup B}%
\right\rangle $ and $\tilde{c}_{\xi^{\prime}}\subseteq A$ implies that
$\xi^{\prime}\not \subseteq A$, under the assumption that $\left\vert
E_{G}\left(  \beta:\xi^{\prime}\right)  \right\vert \overset
{\operatorname{mod}2}{=}0$ for all $\beta\in\mathcal{K}^{(B)}$, one can infer
(according to the definition in Eq. (\ref{eq::def_KappaA_B_in_AB_bipartition})
) that $\xi^{\prime}\in\langle\mathcal{K}_{\sim B}^{(A)}\rangle$. This is
in contradiction to the condition that $\xi^{\prime}\in\left\langle \mathcal{K}%
^{A\rightharpoonup B}\right\rangle $. Therefore, there must be at least one
$\beta_{0}$, which has odd number of edges to $\xi^{\prime}$, i.e., $\left\vert
E_{G}\left(  \beta:\xi^{\prime}\right)  \right\vert \overset
{\operatorname{mod}2}{=}1$.
Hence,
\begin{equation}
\langle\phi_{A\rfloor B}^{(B)}(0)|\phi_{A\rfloor B}^{(B)}(\xi)\rangle=0.
\end{equation}
\end{proof}

\end{widetext}

\newpage
\clearpage
\newpage

\section{The list of notations}
\label{apdx::list_of_notations}
Here we present a list of symbols together with the page number where they occur for the first time.
\printnoidxglossary[type=notation]

\bibliography{X-chains_factorizaton}
\bibliographystyle{myunsrt}

\end{document}